\documentclass[10pt,journal,compsoc]{IEEEtran}
\usepackage{enumerate}
\usepackage{amsmath,amssymb,amsfonts}
\usepackage{mathrsfs}
\usepackage{textcomp}
\usepackage{xcolor}
\usepackage{graphicx}
\usepackage{algorithm}
\usepackage{algpseudocode}
\usepackage{booktabs}
\usepackage{cleveref}
\usepackage{caption} 
\DeclareCaptionStyle{ruled}{labelfont=normalfont,labelsep=colon,strut=off} 
\usepackage[labelformat=simple]{subcaption}

\newcommand{\argmin}{\mathop{\mathrm{argmin}}}

\usepackage{amsthm}

\newtheorem{assumption}{Assumption}

\newtheorem{theorem}{Theorem}[section]
\newtheorem{corollary}{Corollary}[theorem]

\newcommand{\zhou}[1]{{\color{black}{#1}}}
\newcommand{\liu}[1]{{\color{black}{#1}}}

\newcommand{\liuR}[1]{{\color{black}{#1}}}

\ifCLASSOPTIONcompsoc
  \usepackage[nocompress]{cite}
\else
  \usepackage{cite}
\fi
\ifCLASSINFOpdf
\else
\fi

\begin{document}
%
\title{\liu{Multi-Job Intelligent Scheduling with Cross-Device Federated Learning}}
%
%
%
%

\author{Ji Liu\IEEEauthorrefmark{1}\IEEEauthorrefmark{3},
        Juncheng Jia\IEEEauthorrefmark{1}\IEEEauthorrefmark{4},
        Beichen Ma\IEEEauthorrefmark{3},
        Chendi Zhou\IEEEauthorrefmark{4},
        Jingbo Zhou\IEEEauthorrefmark{3},
        Yang Zhou\IEEEauthorrefmark{5},\\
        Huaiyu Dai\IEEEauthorrefmark{6},~and~
        Dejing Dou\IEEEauthorrefmark{3}
\thanks{
\IEEEauthorrefmark{1} Corresponding author.}
\thanks{\IEEEauthorrefmark{3} J. Liu, B. Ma, J. Zhou, and D. Dou are with Baidu Inc., Beijing, China.}
\thanks{\IEEEauthorrefmark{4} J. Jia and C. Zhou are with Soochow University, China.}
\thanks{\IEEEauthorrefmark{5} Y. Zhou is with Auburn University, United States.}
\thanks{\IEEEauthorrefmark{6} H. Dai is with North Carolina State University, United States.}
}

\IEEEtitleabstractindextext{%
\begin{abstract}
Recent years have witnessed a large amount of decentralized data in various (edge) devices of end-users, while the decentralized data aggregation remains complicated for machine learning jobs because of regulations and laws. 
As a practical approach to handling decentralized data, Federated Learning (FL) enables collaborative global machine learning model training without sharing sensitive raw data. 
The servers schedule devices to jobs within the training process of FL. 
In contrast, device scheduling with multiple jobs in FL remains a critical and open problem.
In this paper, we propose a novel multi-job FL framework, which enables the training process of multiple jobs in parallel. 
The multi-job FL framework is composed of a system model and a scheduling method.
The system model enables a parallel training process of multiple jobs, with a cost model based on the data fairness and the training time of diverse devices during the parallel training process. 
We propose \liu{a novel intelligent scheduling approach based on multiple scheduling methods, including an original reinforcement learning-based scheduling method and an original Bayesian optimization-based scheduling method,} which corresponds to a small cost while scheduling devices to multiple jobs.
We conduct extensive experimentation with diverse jobs and datasets. 
The experimental results reveal that our proposed approaches significantly outperform baseline approaches in terms of training time (up to \liu{12.73} times faster) and accuracy (up to \liu{46.4}\% higher).  
\end{abstract}

}

\maketitle

\IEEEdisplaynontitleabstractindextext

%
\IEEEpeerreviewmaketitle

\IEEEraisesectionheading{\section{Introduction}\label{sec:introduction}}

In recent years, we have witnessed a large amount of decentralized data over various Internet of Things (IoT) devices, mobile devices, etc. \cite{liu2021distributed}, which can be exploited to train machine learning models of high accuracy for diverse artificial intelligence applications. Since the data contain sensitive information of end-users, a few stringent legal restrictions \cite{GDPR, CCL, CCPA, chik2013singapore} have been put in place to protect data security and privacy. In this case, it is difficult or even impossible to aggregate the decentralized data into a single server or a data center to train machine learning models. To enable collaborative training with decentralized data, Federated Learning (FL) \cite{mcmahan2017communication}, which does not transfer raw data, have emerged as \liuR{a} practical approach.

FL was first introduced to collaboratively train a global model with non-Independent and Identically Distributed (non-IID) data distributed across mobile devices \cite{mcmahan2017communication}. 
During the training process of FL, the raw data remains decentralized without being transferred to a single server or a single data center \cite{kairouz2019advances, yang2019federated}.
FL only allows the intermediate data to be transferred from the distributed devices, which can be the weights or the gradients of a model. 
FL generally utilizes a parameter server architecture \cite{smola2010architecture,liu2022large,liu2021heterps}, where a server (or a group of servers) coordinates the training process with numerous devices. 
To collaboratively train a global model, the server selects (schedules) several devices to perform local model updates based on their local data, and then it aggregates the local models to obtain a new global model. 
This process is repeated multiple times to generate a global model of high accuracy.

While current FL solutions \cite{mcmahan2017communication, pilla2021optimal} focus on a single-task job or a multi-task job \cite{Smith2017Multi-Task}, FL with multiple jobs \cite{han2020marble} remains an open problem. The major difference between the multi-task job and multiple jobs is that the tasks of the multi-task job share some common parts of the model, while the multiple jobs do not interact with each other in terms of the model. The multi-job FL deals with the simultaneous training process of multiple independent jobs. Each job corresponds to multiple updates during the training process of a global model with the corresponding decentralized data. While the FL with a single job generally selects a portion of devices to update the model, the other devices remain idle, and the efficiency thus is low. The multi-job FL can well exploit diverse devices for multiple jobs simultaneously, which brings high efficiency. The available devices are generally heterogeneous \cite{li2020federated, Li2021Heterogenous}, i.e., the computing and communication capacity of each device is different, and the data in each device may also differ. \liuR{For instance, multiple machine learning jobs, e.g., CTR models \cite{Zhao2019AIBox,Guo2021ScaleFreeCTR}, mobile keyboard prediction \cite{hard2018federated}, and travel time prediction \cite{fang2021ssml}, may be concurrently executed with FL. The concurrent execution can be carried out with the same group of users. In addition, industry-level machine learning jobs, e.g., recommendation system jobs \cite{qin2019duerquiz}, speech recognition \cite{masterson2015baidu}, etc., may be adapted to be executed in parallel with FL for privacy issues. 
}

During the training process of multiple jobs, the devices need to be scheduled for each job. At a given time, a device can be scheduled to one job. However, only a portion of the available devices are scheduled to one job to reduce the influence of stragglers \cite{mcmahan2017communication}. Powerful devices should be scheduled to jobs to accelerate the training process, while other eligible devices should also participate in the training process to increase the fairness of data to improve the accuracy of the final global models. The fairness of data refers to the fair participation of the data in the training process of FL, which can be indicated by the standard deviation of the times to be scheduled to a job \cite{pitoura2007load,finkelstein2008fairness}.

While the scheduling problem of devices is typical NP-hard \cite{Du1989NPHard, liu2020job}, some solutions have already been proposed for the training process of FL \cite{McMahan2017Communication-efficien, Nishio2019Client, Li2021Heterogenous, Abdulrahman2021FedMCCS} or distributed systems \cite{barika2019scheduling}, which generally only focus on a single job with FL. In addition, these methods either cannot address the heterogeneity of devices \cite{McMahan2017Communication-efficien}, or do not consider the data fairness during the training process \cite{Nishio2019Client, Li2021Heterogenous, Abdulrahman2021FedMCCS}, which may lead to low accuracy. 

In this paper, we propose a Multi-Job Federated Learning (MJ-FL) framework to enable the efficient training of multiple jobs with heterogeneous edge devices. The MJ-FL framework consists of a system model and \liu{a novel intelligent scheduling approach}. The system model enables the parallel training process of multiple jobs. With the consideration of both the efficiency of the training process, i.e., the time to execute an iteration, and the data fairness of each job for the accuracy of final models, we propose a cost model based on the training time and the data fairness within the system model. We propose \liu{an intelligent scheduling approach based on multiple scheduling methods, including} two \liu{original} scheduling methods, i.e., reinforcement learning-based and Bayesian optimization-based, to schedule the devices for each job. To the best of our knowledge, we are among the first to study FL with multiple jobs. \liuR{This paper is an extension of a conference version \cite{liu2022Efficient}, with an extra meta-scheduling approach, additional theoretical proof, and extensive experimental results.} We summarize our contributions as follows:

\begin{itemize}
    \item We propose MJ-FL, a multi-job FL framework composed of a parallel training process for multiple jobs and a cost model for scheduling methods.
    We propose combining the capability and data fairness in the cost model to improve the efficiency of the training process and the accuracy of the global model.
    \item We propose two scheduling methods, i.e., Reinforcement Learning (RL)-based and Bayesian Optimization (BO)-based methods, to schedule the devices to diverse jobs \liuR{(more details including the method to estimate the loss in Section \ref{subsec:problem} are added compared with \cite{liu2022Efficient})}. Each method has advantages in a specific situation. The BO-based method performs better for simple jobs, while the RL-based method is more suitable for complex jobs. \liuR{In addition, we provide theoretical convergence analysis in Section \ref{sec:conv}.}
    \item \liu{We propose a novel intelligent scheduling approach based on multiple scheduling methods \liuR{(extra contribution compared with \cite{liu2022Efficient})}. The novel intelligent scheduling approach is a meta-scheduling approach that coordinates multiple scheduling methods to achieve excellent performance \liuR{with an adapted dynamic cost model}.}
    \item We carry out extensive experimentation to validate the proposed approach. We exploit multiple jobs, composed of Resnet18, CNN, AlexNet, VGG, and LeNet, to demonstrate the advantages of our proposed approach using both IID and non-IID datasets \liuR{(with extra extensive experimental results for the meta-scheduling approach)}.
\end{itemize}

The rest of the paper is organized as follows. We present the related work in Section \ref{sec:related_work}. Then, we explain the system model and formulate the problem with a cost model in Section \ref{sec:system_model}. We present the scheduling methods in Section \ref{sec:solutions}. We provide theoretical convergence analysis in Section \ref{sec:conv}. The experimental results with diverse models and datasets are given in Section \ref{sec:experiment}. Finally, Section \ref{sec:conclusion} concludes the paper.

\section{Related Work}
\label{sec:related_work}

In order to protect the security and privacy of decentralized raw data, FL emerges as a promising approach, which enables training a global model with decentralized data  \cite{mcmahan2017communication, yang2019federated, li2020federated, liu2021distributed}. Based on the data distribution, FL can be classified into three types, i.e., horizontal, vertical, and hybrid \cite{yang2019federated, liu2021distributed}. The horizontal FL addresses the decentralized data of the same features, while the identifications are different. The vertical FL handles the decentralized data of the same identifications with different features. The hybrid FL deals with the data of different identifications and different features. In addition, FL includes two variants: cross-device FL and cross-silo FL \cite{kairouz2019advances}. The cross-device FL trains global machine learning models with a huge number of mobile or IoT devices, while the cross-silo FL handles the collaborative training process with the decentralized data from multiple organizations or geo-distributed datacenters. In this paper, we focus on the horizontal and cross-device FL. 

Current FL approaches \cite{Bonawitz19, liu2020fedvision, yurochkin2019bayesian,  Wang2020Federated, Zhang2022FedDUAP,JinAccelerated2022} generally deal with a single job, i.e., with a single global model. While some FL approaches have been proposed to handle multiple tasks \cite{Smith2017Multi-Task, chen2021matching}, the tasks share some common parts of a global model and deal with the same types of data. In addition, the devices are randomly selected (scheduled) in these approaches. 

A few scheduling approaches \cite{McMahan2017Communication-efficien, Nishio2019Client, Li2021Heterogenous, Abdulrahman2021FedMCCS, barika2019scheduling, Nishio2019Client, Li2021Heterogenous, Abdulrahman2021FedMCCS, sun2020deepweave} exist for single-job scheduling while the device scheduling with multi-job FL is rarely addressed. The scheduling methods in the above works are mainly based on some heuristics. For instance, the greedy method \cite{shi2020device} and the random scheduling method \cite{McMahan2017Communication-efficien} are proposed for FL, while genetic algorithms \cite{barika2019scheduling} are exploited for distributed systems. However, these methods do not consider the fairness of data, which may lead to low accuracy for multi-job FL. The black-box optimization-based methods, e.g., RL \cite{sun2020deepweave}, BO \cite{kim2020probabilistic}, and deep neural network \cite{zang2019hybrid}, have been proposed to improve the efficiency, i.e., the reduction of execution time, in distributed systems. They do not consider data fairness either, which may lead to low accuracy for multi-job FL. \liuR{Although ensemble learning or ensemble method consisting of multiple models \cite{chen2019edge}, has been exploited for scheduling parallel tasks \cite{emani2015celebrating}, proper cost models and ensemble mechanism should be well designed.}

\zhou{
\begin{table*}[htbp]
\caption{Summary of Main Notations}
\vspace{-4mm}
\begin{center}
\begin{tabular}{cc}
\toprule
Notation & Definition \\
\hline
\label{tab:summary}
$\mathcal{K}$; $|\mathcal{K}|$ & Set of all devices; size of $\mathcal{K}$ \\
$M$; $m$; $T$ & The total number of jobs; index of jobs; total training time \\
$\mathcal{D}_k^m$; $D_k^m$; $d_k^m$ & Local dataset of Job $m$ on Device $k$; size of $\mathcal{D}_k^m$; batch size of the local update of Device $k$ \\
$\mathcal{D}^m$; $D^m$ & Global dataset of Job $m$; size of $\mathcal{D}^m$ \\
$F_k^m(\boldsymbol{w})$; $F^m(\boldsymbol{w})$ & Local loss function of Job $m$ in Device $k$; global loss function of Job $m$ \\
$\boldsymbol{w}_{k,r}^m(j)$ & Local model of Job $m$ in Device $k$ in the $j$-th local update of Round $r$ \\
$R_m$ & The maximum rounds for Job $m$ during the execution \\
$R_m'$ & The minimum rounds for Job $m$ to achieve the required performance (loss value or accuracy) \\
$l_m$ & The desired loss value for Job $m$\\
$\tau_m$; $C_m$ &Number of local epochs of Job $m$; the ratio between the number of devices scheduled to Job $m$ and $|\mathcal{K}|$\\
$S_m, s_{k,m}^r$ & The frequency vector for Job $m$; the frequency of Device $k$ scheduled to Job $m$ at Round $r$\\
$\mathcal{V}_m^r$ & A set of devices scheduled to Job $m$ at Round $r$\\
$\mathcal{V}_{o}; \mathcal{V}_{o}^r$ & A set of occupied devices; the set of occupied devices in Round $r$ \\
$\zeta^{m,r}_{k,j}$ & The sampled dataset for Job $m$ at local iteration $h$ Round $r$ on Device $k$\\
\bottomrule
\end{tabular}
\end{center}
\vspace{-6mm}
\end{table*}
}

Different from all existing works, we propose a system model for the multi-job FL with the consideration of both efficiency and accuracy. \liu{In addition, we propose a novel intelligent scheduling approach based on multiple scheduling methods. To improve the efficiency of the training process, we propose two original scheduling methods, i.e., RL and BO,} for multi-job FL, which are suitable for diverse models and for both IID and non-IID datasets.

\section{System Model and Problem Formulation}
\label{sec:system_model}

In this section, we first explain the motivation for multi-job FL. Then, we propose our multi-job FL framework 
, consisting of a multi-job FL process and a cost model. Afterward, we formally define the problem to address in this paper. \liu{Please see the meanings of the major notations in Table \ref{tab:summary}. }

\subsection{Motivation for Multi-Job Federated Learning}
Let us assume a scenario where multiple FL jobs are processed simultaneously, e.g., image classification, speech recognition, and text generation. These jobs can be trained in parallel to exploit the available devices efficiently. However, while each device can only update the model of one job at a given time slot, it is critical to schedule devices to different jobs during the training process. 
As the devices are generally heterogeneous, some may possess high computation or communication capability while others may not. 
In addition, the data fairness of multiple devices may also impact the convergence speed of the training process. 
For instance, if only specific powerful devices are scheduled for a job, the model can only learn from the data stored on these devices, while the knowledge from the data stored on other devices may be missed. In order to accelerate the training process of multiple jobs with high accuracy, it is critical to consider how to schedule devices while considering both the computing and communication capability and the data fairness.

A straightforward approach is to train each job separately using the mechanism explained in \cite{McMahan2017Communication-efficien}, while exploiting the existing scheduling of single-job FL, e.g., FedAvg \cite{McMahan2017Communication-efficien}.
In this way, simple parallelism is considered while the devices are not fully utilized and the system is of low efficiency. 
In addition, a direct adaptation of existing scheduling methods to multi-job FL cannot address the efficiency and the accuracy at the same time.
Thus, it is critical to propose a reasonable and effective approach for the multi-job FL.

\subsection{Multi-job Federated Learning Framework}

In this paper, we focus on an FL environment composed of a server module and multiple devices. The server module (Server) may consist of a single parameter server or a group of parameter servers \cite{li2014scaling}. In this section, we present a multi-job FL framework, which is composed of a process for the multi-job execution and a cost model to estimate the cost of the execution.

\subsubsection{Multi-job FL Process}

Within the multi-job FL process, we assume that $K$ devices, denoted by the set $\mathcal{K}$, collaboratively train machine learning models for $M$ jobs, denoted by the set $\mathcal{M}$. Each Device $k$ is assumed to have $M$ local datasets corresponding to the $M$ jobs without loss of generality, and the dataset of the $m$-th job on Device $k$ is expressed as $\mathcal{D}_k^m = \{\boldsymbol{x}_{k,d}^m \in \mathbb{R}^{n_m}, y_{k,d}^m \in \mathbb{R} \}_{d=1}^{D_k^m}$ with $D_k^m = |\mathcal{D}_k^m|$ as the number of data samples, $\boldsymbol{x}_{k,d}^m$ representing the $d$-th $n_m$-dimentional input data vector of Job $m$ at Device $k$, and $y_{k,d}^m$ denoting the labeled output of $\boldsymbol{x}_{k,d}^m$. The whole dataset of Job $m$ is denoted by $\mathcal{D}^m = \bigcup_{k \in \mathcal{K}} \mathcal{D}_k^m$ with $D^m=\sum_{k \in \mathcal{K}} D_k^m$.
The objective of multi-job FL is to learn respective model parameters $\{\boldsymbol{w}^m\}$ based on the decentralized datasets. 

The global learning problem of multi-job FL can be expressed by the following formulation:
\vspace{-3mm}
\begin{equation}\label{eq:eq1}
\min_{\boldsymbol{W}}\sum_{m=1}^M\mathbb{L}_m\textrm{, with }\mathbb{L}_m = \sum_{k=1}^K\frac{D_k^m}{D^m}F_k^m(\boldsymbol{w}^m),
\vspace{-2mm}
\end{equation}
where $\mathbb{L}_m$ is the loss value of Job $m$, $F_k^m(\boldsymbol{w}^m)=\frac{1}{D_k^m}\sum_{\{\boldsymbol{x}_{k,d}^m,y_{k,d}^m\}\in\mathcal{D}_k^m} f^m(\boldsymbol{w}^m;\boldsymbol{x}_{k,d}^m,y_{k,d}^m)$ is the loss value of Job $m$ at Device $k$, $\boldsymbol{W}: \equiv\{\boldsymbol{w}^1,\boldsymbol{w}^2,...,\boldsymbol{w}^M\}$ is the set of weight vectors for all jobs, and $f^m(\boldsymbol{w}^m;\boldsymbol{x}_{k,d}^m, y_{k,d}^m)$ captures the error of the model parameter $\boldsymbol{w}^m$ on the data pair $\{\boldsymbol{x}_{k,d}^m, y_{k,d}^m\}$.

\begin{figure*}[htbp]
\centerline{\includegraphics[width=0.8\linewidth]{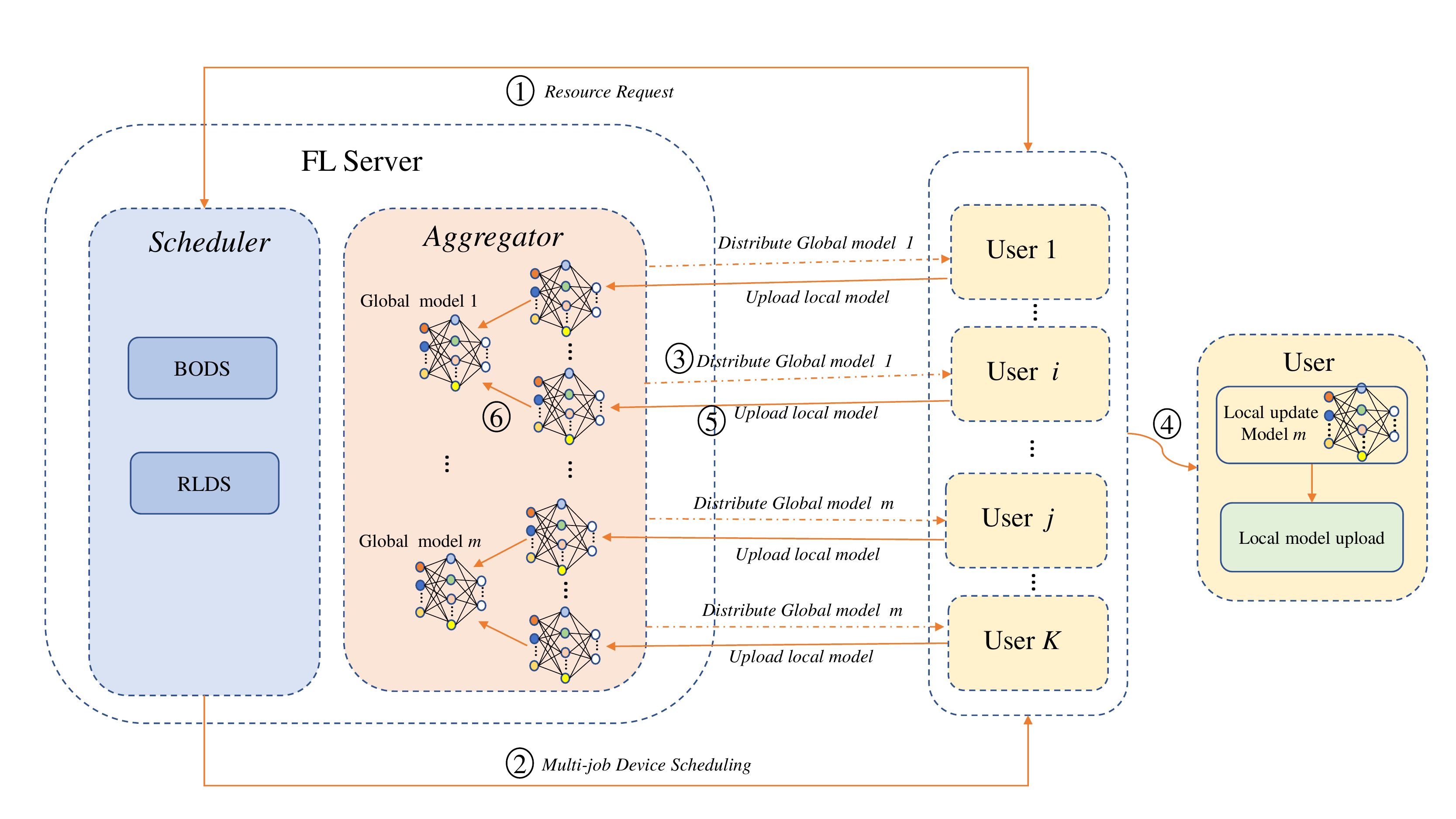}}
\vspace{-3mm}
\caption{The training process within the Multi-job Federated Learning Framework.}
\label{framework}
\vspace{-6mm}
\end{figure*}

In order to solve the problem defined in Formula \ref{eq:eq1}, the Server needs to continuously schedule devices for different jobs to iteratively update the global models until the training processes of the corresponding job converge or achieve the target performance requirement (in terms of accuracy or loss value). We design a multi-job FL process as shown in Figure \ref{framework}. The Server first initializes a global model for each job. The initialization can be implemented randomly or from the pre-training process with public data. To know the current status of devices, the Server sends requests to available devices in Step \textcircled{1}. Then, in Step \textcircled{2}, the Server schedules devices for the current job based on the scheduling plan generated by the scheduling method (see details in Section \ref{sec:solutions}). The scheduling plan is a set of devices selected to execute the local training process of the current job. Note that the scheduling process generates a scheduling plan for each job during the training process of multiple jobs, i.e., with an online strategy, while the scheduling processes for multiple jobs are carried out in parallel. The Server distributes the latest global model for the current job to the scheduled devices in Step \textcircled{3}, and then the model is updated in each device based on the local data in Step \textcircled{4}. Afterward, each device sends the updated model to the Server after its local training in Step \textcircled{5}. Finally, the Server aggregates the models of scheduled devices to generate a new global model in Step \textcircled{6}. The combination of Steps \textcircled{1} - \textcircled{6} is denoted by a round, which is repeated for each job until the corresponding global model reaches the expected performance (accuracy, loss value, or convergence). Please note that multiple jobs are executed in parallel asynchronously, while a device can only be scheduled to one job at a given time. In addition, we assume that each job is of equal importance.

As an FL environment may contain GPUs or other high-performance chips, it is beneficial to train multiple jobs simultaneously to reduce training time while achieving target accuracy. 
Within each round, Step \textcircled{6} exploits FedAvg \cite{McMahan2017Communication-efficien} to aggregate multiple models within each job, which can ensure the optimal convergence \cite{li2019convergence, Zhou2018On}.Within our framework, the sensitive raw data is kept within each device, while only the models are allowed to be transferred. Other methods, e.g., homomorphic encryption \cite{paillier1999public} and differential privacy \cite{dwork2008differential}, can be exploited to protect the privacy of sensitive data. 

\subsubsection{Cost Model}

In order to measure the performance of each round, we exploit a cost model defined in Formula \ref{eq:totalCost}, which is composed of time cost and data fairness cost. The data fairness has a significant impact on convergence speed. 
\vspace{-1mm}
\begin{equation}\label{eq:totalCost}
\vspace{-1mm}
Cost_{m}^{r}(\mathcal{V}_m^r) = \alpha * \mathscr{T}_m^{r}(\mathcal{V}_m^r) + \beta * \mathscr{F}_m^{r}(\mathcal{V}_m^r),
\end{equation}
where $\alpha$ and $\beta$ are the weights of time cost and fairness cost respectively, $\mathscr{T}_m^{r}(\cdot)$ represents the execution time of the training process in Round $r$ with the set of scheduled devices $\mathcal{V}_m^r$, and $\mathscr{F}_m^{r}(\cdot)$ is the corresponding data fairness cost. \liu{We choose the linear combination because of its convenience and excellent performance. In practice, we empirically set $\alpha$ and $\beta$ based on the information from previous execution and adjust them using small epochs. We increase $\alpha$ for fast convergence and increase $\beta$ mainly for high accuracy.}

As defined in Formula \ref{eq:time}, the execution time of a round depends on the slowest device in the set of scheduled devices. 
\vspace{-1mm}
\begin{equation}\label{eq:time}
\mathscr{T}_{m}^{r}(\mathcal{V}_m^r) = \max_{k \in \mathcal{V}_m^r}\{t_m^k\},
\vspace{-1mm}
\end{equation}
where $t_m^k$ is the execution time of Round $r$ in Device $k$ for Job $m$. $t_m^k$ is composed of the communication time and the computation time, which is complicated to estimate and differs for different devices. In this study, we assume that the execution time of each device follows the shift exponential distribution as defined in Formula \ref{eq:distribution} \cite{Shi2021Joint, Lee2018Speeding}: 
\vspace{-2mm}
\begin{equation}
\vspace{-2mm}
\label{eq:distribution}
    \begin{split}
        P[t_m^k \text{\textless} t]= \left \{
            \begin{array}{ll}
                1-e^{-\frac{\mu_k}{\tau_m D^m_k}(t - \tau_m a_k D^m_k)}, & t\geq \tau_m a_k D^m_k,\\
                0,   & \rm{otherwise},
            \end{array}
        \right.
    \end{split}
\end{equation}
where the parameters $a_k > 0$ and  $\mu_k > 0$ are the maximum and fluctuation of the computation and communication capability, which is combined into one quantity, of Device $k$, respectively. 
Moreover, we assume that the calculation time of model aggregation has little impact on the training process because of the strong computation capability of the Server and the low complexity of the model.

The data fairness of Round $r$ corresponding to Job $m$ is indicated by the deviation of the frequency of each device to be scheduled to Job $m$ defined in Formula \ref{eq:fairness}.
\vspace{-1mm}
\begin{equation}\label{eq:fairness}
\mathscr{F}_m^{r}(\mathcal{V}_m^r)=\frac{1}{|\mathcal{K}|}\sum_{k \in \mathcal{K}}(s_{k,m}^r-\frac{1}{|\mathcal{K}|}\sum_{k \in \mathcal{K}}s_{k,m}^r)^2,
\vspace{-2mm}
\end{equation}
where $s_{k,m}^r$ is the frequency of Device $k$ to be scheduled to Job $m$, and $\mathcal{K}$ and $|\mathcal{K}|$ are the set of all devices and the size, respectively. $s_{k,m}^r$ is calculated by counting the total number of the appearance of Device $k$ to be scheduled to Job $m$ in the set of scheduling plans for Job $m$, i.e., $\{\mathcal{V}_m^1,..., \mathcal{V}_m^r\}$.
\liu{
In particular, $S^r_m = \{s^r_{1, m}, ..., s^r_{|\mathcal{K}|,m}\}$ represents the frequency vector of Job $m$ at Round $r$. At the beginning, i.e., Round 0, each $s^0_{k, m} \in S^0_m$ is 0. $s_{k, m}^r$ represents the frequency of Device $k$ scheduled to Job $m$ at Round $r$. Then, we can calculate $s_{k, m}^{r+1}$ using the following formula:
\vspace{-3mm}
\begin{equation}
s_{k, m}^{r+1} = 
\begin{cases}
\displaystyle s_{k, m}^{r} + 1,\hspace*{0.45in}$if$~$Device$~k \in \mathcal{V}_m^r \\
\displaystyle s_{k, m}^{r},\hspace*{0.7in}$otherwise$
\end{cases}
\end{equation}
\vspace{-3mm}
}

\liu{Please note that the ``data fairness'' is reflected by the deviation of the frequency of each device scheduled to a job \cite{petrangeli2014multi}, which is different from the ``fairness'' (the bias of the machine learning models concerning certain features) in machine learning \cite{mehrabi2021survey}. Formula \ref{eq:fairness} is inspired by \cite{petrangeli2014multi}, and we are among the first to extend this idea from distributed or network systems to FL. When the devices are non-uniformly sampled with low data fairness, the convergence is slowed down \cite{li2019convergence, Zhou2018On}. In addition, data fairness is important due to the underlying data heterogeneity across the devices. Data fairness can help arbitrarily select devices without harming the learning performance.}

\subsection{Problem Formulation}
\label{subsec:problem}

The problem we address is how to reduce the training time when given a loss value for each job. While the execution of each job is carried out in parallel, the problem can be formulated as follows:
\vspace{-2mm}
\begin{equation}
\vspace{-2mm}
\begin{split}
\label{eq:problem}
    &\displaystyle\min_{\mathcal{V}_m^r}\Big\{ \sum_{m = 1}^M \sum_{r=1}^{R_m'} \mathscr{T}_{m}^{r}(\mathcal{V}_m^r)\Big\}\\
    \text{s.t.} &\begin{cases}
    \displaystyle \mathbb{L}_m(R_m') \leq l_m, \\
    \displaystyle\mathcal{V}_m^r \subset \mathcal{K}, \forall m\in \{1,2,...,M\}, \forall r\in \{1,2,...,R_m'\},
    \end{cases}
\end{split}
\end{equation}
where $l_m$ is the given loss value of Job $m$, $R_m'$ represents the minimum number of rounds to achieve the given loss in the real execution, and $\mathbb{L}_m(R_m')$ is the loss value of the trained model at Round $R_m'$, defined in Formula \ref{eq:eq1}.

\liu{
We assume Stochastic Gradient Descent (SGD) is utilized to train models, which converges at a rate of $O(r)$ with $r$ representing the number of rounds \cite{2018Optimus}. Inspired by \cite{li2019convergence}, we exploit Formula \ref{eq:accuracy} to roughly estimate the loss value of the global model for Job $m$ at Round $r$.
\vspace{-2mm}
\begin{equation}
\vspace{-2mm}
\label{eq:accuracy}
Loss_m(r) = \frac{1}{\gamma^0_m r + \gamma^1_m} + \gamma^2_m,
\end{equation}
where $\gamma_m^0$, $\gamma_m^1$ and $\gamma_m^2$ represent non-negative coefficients of the convergence curve of Job $m$. $\gamma_m^0$, $\gamma_m^1$ and $\gamma_m^2$ can be calculated based on previous execution.
In addition, we assume that the real number of rounds corresponding to the same loss value has 30\% error compared with $r$ (from the observation of multiple executions). Given a loss value of a model, we exploit this loss estimation method to calculate the maximum rounds for each job. Given a loss value of a model, we utilize this loss estimation method to calculate the number of rounds as $R_m^c$ and take $(1+0.3) * R_m^c$ as $R_m$ defined in Table \ref{tab:summary}. Please note that this estimation is different from the loss value during the real execution; i.e., $R_m'$ can be different from $R_m$.

}

As it requires the global information of the whole training process, which is hard to predict, to solve the problem, we transform the problem to the following one, which can be solved with limited local information of each round. In addition, in order to achieve the given loss value of Job $m$ within a short time (the first constraint in Formula \ref{eq:problem}), we need to consider the data fairness within the total cost in Formula \ref{eq:problem2}, within which the data fairness can help reduce $R_m'$ to minimize the total training time.
\vspace{-2mm}
\begin{equation}
\vspace{-1mm}
\begin{split}
\label{eq:problem2}
    \min_{\mathcal{V}_m^r}\Big\{& TotalCost(\mathcal{V}_m^r)\Big\}, \\
    Total&Cost(\mathcal{V}_m^r) = \sum_{m' = 1}^M Cost_{m'}^{r}(\mathcal{V}_{m'}^r), \\
    s.t. \qquad & \mathcal{V}_{m'}^r \subset \mathcal{K}, \forall {m'}\in \{1,2,...,M\},
\end{split}
\end{equation}
where $Cost_{m}^{r}(\mathcal{V}_m^r)$ can be calculated based on Formula \ref{eq:totalCost} with a set of scheduled devices $\mathcal{V}_m^r$ to be generated using a scheduling method for Job $m$.
Since the scheduling results of one job may potentially influence the scheduling of other jobs, we consider the cost of other jobs when scheduling devices to the current job in this problem.
As the search space is $O(2^{|\mathcal{K}|})$, this scheduling problem is a combinatorial optimization problem \cite{toth2000optimization} and NP-hard \cite{Du1989NPHard}.

\section{Device Scheduling for Multi-job FL}
\label{sec:solutions}

In this section, we propose two original scheduling methods, i.e., BO-based and RL-based, \liu{and a novel intelligent scheduling approach, i.e., meta-greedy, }to address the problem defined in Formula \ref{eq:problem2}. The scheduling plan generated by a scheduling method is defined in Formula \ref{eq:schedulingPlan}:
\vspace{-2mm}
\begin{equation}
\vspace{-2mm}
\label{eq:schedulingPlan}
    \mathcal{V'}_m^r = \argmin_{\mathcal{V}_m^r\subset {\{\mathcal{K} \backslash \mathcal{V}_{o}^r\}}} TotalCost(\mathcal{V}_m^r),
\end{equation}
where $\mathcal{V'}_m^r$ is a scheduling plan, $\mathcal{K} \backslash \mathcal{V}_{o}^r$ represents the set of available devices to schedule, $TotalCost(\mathcal{V}_m^r)$ is defined in Formula \ref{eq:problem2}, and $\mathcal{K}$ and $\mathcal{V}_{o}^r$ are the set of all devices and the set of occupied devices in Round $r$, respectively.

\liu{The BO-based and RL-based methods are designed for different model complexities, and we choose the better one based on known profiling information with small tests (a few epochs) to avoid possible limitations. RLDS favors complex jobs, as it can learn the influence among diverse devices. The influence refers to the concurrent, complementary, and latent impacts of the data in multiple devices for diverse jobs. However, BODS favors simple jobs, while it relies on simple statistical knowledge. The complexity of jobs is determined by the number of parameters of models and the size of the training dataset. We consider the probability to release the devices in $\mathcal{V}_{o}$ in BODS and RLDS, and possible concurrent occupation of other devices for other jobs, which is not explained in the paper to simplify the explanation. During the execution, we assume that a fraction of the devices is sampled for each job and some devices can be unavailable for the execution. In order to optimize the scheduling process for diverse types of models, we further propose a meta-greedy scheduling approach, which takes advantage of multiple scheduling methods and chooses the most appropriate scheduling plan from the results of the scheduling methods.}

\begin{figure}[t]
\vspace{-4mm}
\begin{algorithm}[H]		
    \caption{Bayesian Optimization-Based Scheduling}
    \label{alg:bayesian}
    {\bf{Input:}}\\
    \hspace*{0.3in}$\mathcal{V}_{o}:$ A set of occupied devices\\
    \hspace*{0.3in}$S_m:$ A vector of the frequency of each device sched- \hspace*{0.6in}uled to Job $m$\\
    \hspace*{0.3in}$R_m:$ The maximum round of the current Job $m$\\
    \hspace*{0.3in}$l_m:$ The desired loss value for Job $m$.\\
    {\bf{Output:}}\\
    \hspace*{0.3in} $\mathcal{V}_m = \{\mathcal{V}_{m}^{*1},...,\mathcal{V}_{m}^{*R_m}\}:$ a set of scheduling plans, \hspace*{0.35in}each with the size $|\mathcal{K}| \times C_m$ 
    \begin{algorithmic}[1]
        \State $\Pi_L$ $\leftarrow$ Randomly generate a set of observation points and calculate the cost \label{line:randomGeneration}
        \For{$r \in \{1,...R_m\}$ and $l_m$ is not achieved}
            \State $\Pi'$ $\leftarrow$ Randomly generate a set of observation points \hspace*{0.2in}with the devices within $\mathcal{K} \backslash \mathcal{V}_{o}$ \label{line:randomGenerationInLoop}
            \State $\mathcal{V}_{m}^{*r} \leftarrow \underset{\mathcal{V}\subset \Pi'}{arg\max} \alpha_{\rm EI}(\mathcal{V};\Pi')$ \label{line:chooseNext}
            \State FL training of Job $m$ with $\mathcal{V}_{m}^{*r}$ and update $S_m$, $\mathcal{V}_{o}$ \label{line:boTraining}
            \State $\mathbb{C}_{r}=TotalCost(\mathcal{V}_{m}^{*r})$ \label{line:calculateCost}
            \State $\Pi_{L + r} \gets \Pi_{L + r - 1} \cup (\mathcal{V}_{m}^{*r}, \mathbb{C}_{r})$\label{line:addPoint}
        \EndFor
   \end{algorithmic}
\end{algorithm}
\vspace{-8mm}
\end{figure}

\subsection{Bayesian Optimization-Based Scheduling}

As the Gaussian Process (GP) \cite{Srinivas2010Gaussian} can well represent linear and non-linear functions, BO-based methods \cite{Shahriari2061BO} can exploit a GP to find a near-optimal solution for the problem defined in Formula \ref{eq:schedulingPlan}.
In this section, we propose a Bayesian Optimization-based Device Scheduling method (BODS).

We adjust a GP to fit the cost function $TotalCost(\cdot)$. The GP is composed of a mean function $\mu$ defined in Formula \ref{eq:mu} and a covariance function $\rm{K}$ defined in Formula \ref{eq:kernel} with a $Matern$ kernel \cite{williams2006Ggaussian}.
\vspace{-2mm}
\begin{equation}\label{eq:mu}
\vspace{-2mm}
\mu(\mathcal{V}_{m}^r)=\underset{\mathcal{V}_{m}^r \subset {\{\mathcal{K} \backslash \mathcal{V}_{o}^r\}}}{\mathbb{E}}[TotalCost(\mathcal{V}_{m}^r)]
\end{equation}
\begin{equation}
\vspace{-2mm}
\begin{split}
\label{eq:kernel}
& {\rm K}(\mathcal{V}_{m}^r, \mathcal{V}_{m}'^r) = \underset{\mathcal{V}_{m}^r \subset {\{\mathcal{K} \backslash \mathcal{V}_{o}^r\}}, \mathcal{V}_{m}'^r \subset {\{\mathcal{K} \backslash \mathcal{V}_{o}^r\}}}{\mathbb{E}}\\
& [(TotalCost(\mathcal{V}_{m}^r)-\mu(\mathcal{V}_{m}^r)) (TotalCost(\mathcal{V}_{m}'^r)-\mu(\mathcal{V}_{m}'^r))]
\end{split}
\end{equation}

The BODS is explained in {\bf{Algorithm \ref{alg:bayesian}}}.
First, we generate a random set of observation points and calculate the cost according to the Formula \ref{eq:totalCost} (Line \ref{line:randomGeneration}). Each observation point is a pair of scheduling plan and cost for the estimation of mean function and the covariance function. Then, in each round, we randomly sample a set of scheduling plans (Line \ref{line:randomGenerationInLoop}), within which we use updated $\mu$ and $K$ based on $\Pi_{L + r - 1}$ (Line \ref{line:chooseNext}) to select the one with the largest reward. Afterward, we perform the FL training for Job $m$ with the generated scheduling plan (Line \ref{line:boTraining}). In the meanwhile, we calculate the cost corresponding to the actual execution (Line \ref{line:calculateCost}) according to Formula \ref{eq:problem2} and update the observation point set (Line \ref{line:addPoint}).

Let ($\mathcal{V}_l^r$, $\mathbb{C}_l$) denote an observation point $l$ for Job $m$ in Round $r$, where $\mathcal{V}_l^r=\{\mathcal{V}_{l, 1}^r,...,\mathcal{V}_{l, M}^r\}$ and $\mathbb{C}_l$ is the cost value of $TotalCost(\mathcal{V}_{l,m}^r)$ while the scheduling plans of other jobs are updated with the ones in use in Round $r$.
At a given time, we have a set of observations $\Pi_{L-1}=\{(\mathcal{V}_1^r, \mathbb{C}_1),...,(\mathcal{V}_{L-1}^r, \mathbb{C}_{L-1})\}$ composed of $L-1$ observation points.
We denote the minimum cost value within the $L - 1$ observations by $\mathbb{C}^{+}_{L-1}$. Then, we exploit Expected Improvement(EI) \cite{Jones1998Efficient} to select a new scheduling plan $\mathcal{V}_{m}^{*r}$ in Round $r$ what improves $\mathbb{C}^{+}_{L-1}$ the most, which is the utility function.
Please note that this is not an exhaustive search since we randomly select several observation points (a subset of the whole search space) at the beginning and add new observation points using the EI method. \liu{
The utility function is defined in Formula \ref{eq:equility}.
\begin{figure}[htbp]
\centering
\centerline{\includegraphics[width=1.0\linewidth]{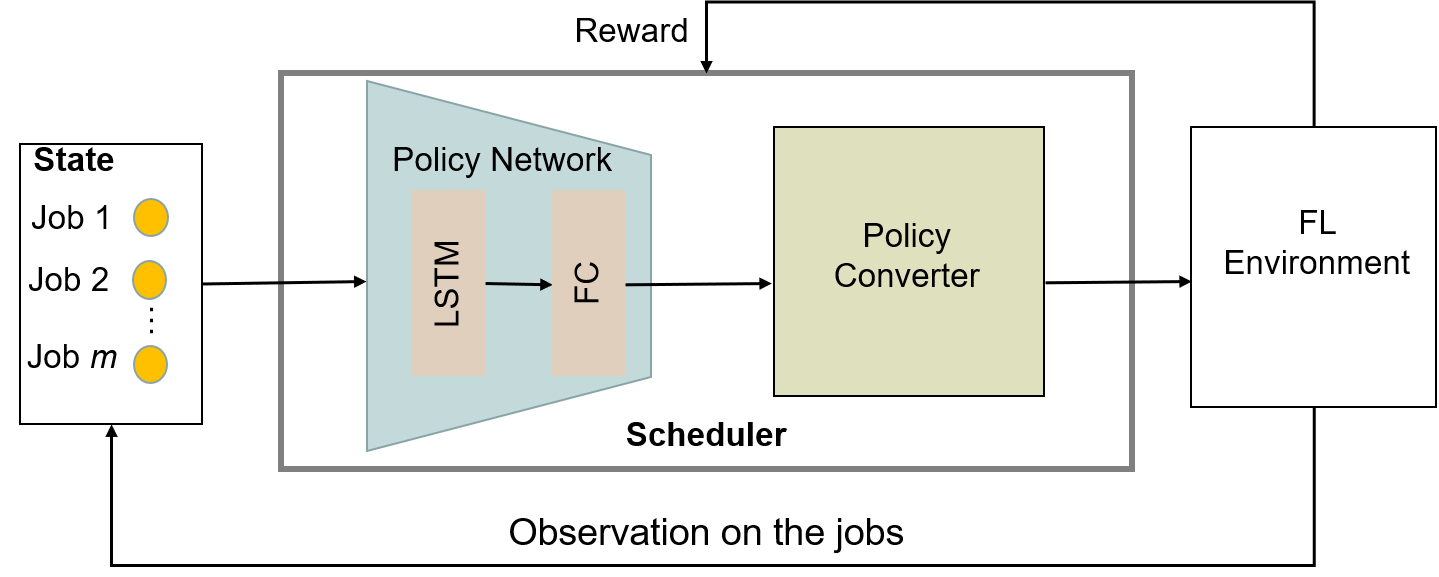}}
\vspace{-2mm}
\caption{The architecture of the RLDS.}
\label{RLDS-framework}
\vspace{-6mm}
\end{figure}
\vspace{-2mm}
\begin{equation}
\vspace{-2mm}
\label{eq:equility}
    u(\mathcal{V}_{m}^{*r})=max(0, \mathbb{C}^{+}_{L-1}-TotalCost(\mathcal{V}_{m}^{*r})),
\end{equation}
where we receive a reward $\mathbb{C}^{+}_{L-1}-TotalCost(\mathcal{V}_{m}^{*r})$ if $TotalCost(\mathcal{V}_{m}^{*r})$ turns out to be less than $\mathbb{C}^{+}_{L-1}$, and no reward otherwise. Then, we use the following formula, which is also denoted an acquisition function, to calculate the expected reward of a given scheduling plan $\mathcal{V}$.
\begin{equation}\label{eq:eq7}
\begin{aligned}
    \alpha_{\rm EI}(\mathcal{V};\Pi_{L-1}) = &\mathbb{E}[u(\mathcal{V})|\mathcal{V}, \Pi_{L-1}]\\
    =&(\mathbb{C}^{+}_{L-1}-\mu(\mathcal{V}))\Phi(\mathbb{C}^{+}_{L-1}; \mu(\mathcal{V}), 
    {\rm K}(\mathcal{V},\mathcal{V}))\\
    &+ {\rm K}(\mathcal{V},\mathcal{V})\mathcal{N}(\mathbb{C}^{+}_{L-1}; \mu(\mathcal{V}), \rm K(\mathcal{V},\mathcal{V})),
 \end{aligned}
\end{equation}
where $\Phi$ is the Cumulative Distribution Function (CDF) of the standard Gaussian distribution. Finally, we can choose the scheduling plan with the largest reward as the next observation point, i.e., $\mathcal{V}_{L,m}^{*r}$.
}

\subsection{Reinforcement Learning-Based Scheduling}
In order to learn more information about the near-optimal scheduling patterns for complex jobs, we further propose a Reinforcement Learning-based Device Scheduling (RLDS) method as shown in Figure \ref{RLDS-framework}. In addition, the method is inspired by \cite{Hongzi2019rniLeang, sun2020deepweave}. The scheduler of RLDS consists of a policy network and a policy converter. During the process of device scheduling, RLDS collects the status information of jobs as the input of the policy network. Afterwards, the policy network generates a list of probabilities on all devices as the output. Finally, the policy converter converts the list into a scheduling plan.

\subsubsection{Policy Network} The policy network is implemented using a Long Short-Term Memory (LSTM) network followed by a fully connected layer, which can learn the device sharing relationship among diverse jobs. We take the computation and communication capability of available devices to be used in Formula \ref{eq:distribution}, and the data fairness of each job defined in Formula \ref{eq:fairness} as the input. The network calculates the probability of each available device to be scheduled for a job.

\subsubsection{Policy Converter} The Policy Converter generates a scheduling plan based on the probability of each available device which \liuR{is} calculated by the policy network with the $\epsilon$-greedy strategy \cite{xia2015online}.

\subsubsection{Training} In the training process of RLDS, we define the reward as $\mathscr{R}^m = -1 * TotalCost(\mathcal{V}_m^r)$. Inspired by \cite{williams1992simple, Zoph2017Neural}, we exploit Formula \ref{eq:rlupdate} to update the policy network:
\vspace{-2mm}
\begin{equation}\label{eq:rlupdate}
\begin{aligned}
    \theta^{'} = \theta + &\frac{\eta}{N} \sum_{n=1}^N \sum_{k \in \mathcal{V}_{n,m}^{r}}^{ \mathcal{V}_{n,m}^{r} \subset \mathcal{K}\backslash\mathcal{V}_{o}^r} \nabla_\theta &\log P(\mathscr{S}_k^m|\mathscr{S}_{(k-1):1}^m; \theta) \\&(\mathscr{R}_n^m - b_m),
\end{aligned}
\end{equation}
where $\theta'$ and $\theta$ represent the updated parameters and the current parameters of the policy network, respectively, $\eta$ represents the learning rate, $N$ is the number of scheduling plans to update the model in Round $r$ ($N > 1$ in the pre-training process and $N = 1$ during the execution of multiple jobs), $P$ represents the probability calculated based on the RL model, $\mathscr{S}_k^m = 1$ represents that Device $k$ is scheduled to Job $m$, and $b_m$ is the baseline value for reducing the variance of the gradient. 

\begin{figure}[t]
\vspace{-4mm}
\begin{algorithm}[H]		
    \caption{Reinforcement Learning Based Pre-Training}
    \label{alg:reinforcement-pre-training}
    {\bf{Input:}}\\
    \hspace*{0.3in}$\mathcal{V}_{o}:$ A set of occupied devices\\
    \hspace*{0.3in}$S_m:$ A vector of the frequency of each device sched- \hspace*{0.6in}uled to Job $m$\\
    \hspace*{0.3in}$N:$ The number of scheduling plans used to train the \hspace*{0.55in}network for each round\\
    \hspace*{0.3in}$R_m:$ The maximum round of the current Job $m$\\
    \hspace*{0.3in}$l_m:$ The desired loss value for Job $m$.\\
    {\bf{Output:}}\\
    \hspace*{0.3in} $\theta:$ Parameters of the pre-trained policy network
    
    \begin{algorithmic}[1]
        \State $\theta$ $\leftarrow$ randomly initialize the policy network, $b_m\leftarrow0$ \label{line:ap:loadModel}
        \For{$r \in \{1,2,...,R_m\}$ and $l_m$ is not achieved}
            \State $\mathcal{V}_{m}^{r} \leftarrow$ generate a set of $N$ scheduling plans\label{line:ap:generatePlan}
            \For{$\mathcal{V}_{n, m}^{r} \in \mathcal{V}_{m}^{r}$}
                \State $\mathscr{R}_n^m \leftarrow$ $-1 *TotalCost(\mathcal{V}_{n,m}^{r})$ \label{line:ap:computeRD}
            \EndFor
            \State Update $\theta$ according to Formula \ref{eq:rlupdate} \label{line:ap:updateTheta}
            \State $b_m$ $\leftarrow$ (1 - $\gamma)$ * $b_m$ + $\frac{\gamma}{N}$ * $\sum_{n=1}^N \mathscr{R}_n^m$ \label{line:ap:updateBm}
            \State $\mathcal{V}_{m}^{*r}$ $\leftarrow$ $\argmin_{\mathcal{V}_{n,m}^{r} \in \mathcal{V}_{m}^{r} }$TotalCost($\mathcal{V}_{n,m}^{r}$)\label{line:ap:bestSchedulingPlan}
            \State Update $S_m$, $\mathcal{V}_{o}$ with $\mathcal{V}_{m}^{*r}$\label{line:ap:updateParameters}
        \EndFor
   \end{algorithmic}
\end{algorithm}
\vspace{-8mm}
\end{figure}

\liu{
We pre-train the policy network using {\bf{Algorithm \ref{alg:reinforcement-pre-training}}}. First, we randomly initialize the policy network (Line \ref{line:ap:loadModel}). We use the latest policy network and the $\epsilon$-Greedy method to generate $N$ scheduling plans (Line \ref{line:ap:computeRD}). The parameters are updated based on the Formula \ref{eq:rlupdate} (Line \ref{line:ap:updateTheta}), and the baseline value $b_m$ is also updated with the consideration of the historical value (Line \ref{line:ap:updateBm}). Afterward, we choose the best scheduling plan that corresponds to the minimum total cost, i.e., the maximum reward (Line \ref{line:ap:bestSchedulingPlan}). Finally, we update the frequency matrix $S_m$ and the set of occupied devices $\mathcal{V}_{o}$, while assuming that the best scheduling plan is used for the multi-job FL (Line \ref{line:ap:updateParameters}).
}

\begin{figure}[t]
\vspace{-4mm}
\begin{algorithm}[H]		
    \caption{Reinforcement Learning-Based Scheduling}
    \label{alg:reinforcement}
    {\bf{Input:}}\\
    \hspace*{0.3in}$\mathcal{V}_{o}:$ A set of occupied devices\\
    \hspace*{0.3in}$S_m:$ A vector of the frequency of each device sched- \hspace*{0.6in}uled to Job $m$\\
    \hspace*{0.3in}$R_m:$ The maximum round of the current Job $m$\\
    \hspace*{0.3in}$l_m:$ The desired loss value for Job $m$.\\
    {\bf{Output:}}\\
    \hspace*{0.3in} $\mathcal{V}_m = \{\mathcal{V}_{m}^{1},...,\mathcal{V}_{m}^{R_m}\}:$ a set of scheduling plans, \hspace*{0.35in}each with the size $|\mathcal{K}| \times C_m$ 
    
    \begin{algorithmic}[1]
        \State $\theta$ $\leftarrow$ pre-trained policy network, $\Delta\theta\leftarrow0$, $b_m\leftarrow0$ \label{line:loadModel}
        \For{$r \in \{1,2,...,R_m\}$ and $l_m$ is not achieved}
            \State $\mathcal{V}_{m}^{r} \leftarrow$ generate a scheduling plan using the policy \hspace*{0.2in}network \label{line:generatePlan}
            \State FL training of Job $m$ and update $S_m$, $\mathcal{V}_{o}$ \label{line:training}
            \State Compute $\mathscr{R}^m$\label{line:computeRD}
            \State Update $\theta$ according to Formula \ref{eq:rlupdate} \label{line:updateTheta}
            \State $b_m$ $\leftarrow$ (1 - $\gamma)$ * $b_m$ + $\gamma$ * $\mathscr{R}_n^m$ \label{line:updateBm}
        \EndFor
   \end{algorithmic}
\end{algorithm}
\vspace{-9mm}
\end{figure}

\liu{After the pre-training, we exploit RLDS during the training process of multiple jobs within the MJ-FL framework as shown in {\bf{Algorithm \ref{alg:reinforcement}}}. First, we load the pre-trained policy network and initialize the parameters $\Delta\theta$, $b_m$ (Line \ref{line:loadModel}). }
When generating a scheduling plan for Job $m$, the latest policy network is utilized (Line \ref{line:generatePlan}). We perform the FL training for Job $m$ with the generated scheduling plan and update the frequency matrix $S_m$ and the set of occupied devices $\mathcal{V}_{o}$ (Line \ref{line:training}). Afterward, we calculate the reward corresponding to the real execution (Line \ref{line:computeRD}). The parameters are updated based on the Formula \ref{eq:rlupdate} (Line \ref{line:updateTheta}), while the baseline value $b_m$ is updated with the consideration of the historical value (Line \ref{line:updateBm}).

\liu{
\subsection{Meta-Greedy Scheduling}

Multiple jobs of diverse structures and layers exist within the multi-job federated learning environment. Some scheduling methods, e.g., BODS, favor simple jobs, while some other scheduling methods, e.g., RLDS, prefer complex jobs. In addition, heuristic scheduling methods, e.g., Greedy and Genetic, can be exploited to schedule devices, as well. In this case, we take advantage of the existing scheduling methods, and further propose a meta-greedy scheduling approach as shown in Figure \ref{Meta-Greedy-framework}. In Round $r$, the Meta-Greedy executes six scheduling methods in parallel, appends the solution generated by each method to the candidate scheduling solutions set $\Theta^r_m$, and selects a scheduling plan from the set of candidate device scheduling solutions according to Formula \ref{eq:RefactoredtotalCost} as the solution for the $r$-th round under the current job.  The meta-greedy scheduling approach chooses the most appropriate one from the scheduling plans generated by multiple methods, which can be closer to the optimal solution compared with that of a single method.

The Meta-Greedy algorithm is shown in {\bf{Algorithm \ref{alg:MetaGreedy}}}. First, for each round (Line \ref{meta:eachRound}), we exploit diverse scheduling methods, e.g., BODS, RLDS, Random \cite{McMahan2017Communication-efficien}, FedCS \cite{Nishio2019Client}, Genetic \cite{barika2019scheduling}, and Greedy \cite{shi2020device}, to generate scheduling plan candidates (Line \ref{meta:generatePlans}). \liuR{As each scheduling method may have superior performance in a specific environment (Random corresponds to high final accuracy; Genetic corresponds to high accuracy at the beginning of the training process for the jobs of moderate complexity; Greedy corresponds to high accuracy at the beginning of the training process of simple jobs; FedCS corresponds to high accuracy at the middle of the the training process; see details in Section \ref{sec:experiment}), we take these 6 methods in Meta-Greedy.} Then, we choose a scheduling plan that corresponds to the smallest total cost according to Formula \ref{eq:RefactoredtotalCost} (Line \ref{meta:selectPlan}). Finally, we utilize the selected scheduling plan for the training process (Line \ref{meta:usePlan}). 

\begin{figure}[htbp]
\centering
\centerline{\includegraphics[width=1.0\linewidth]{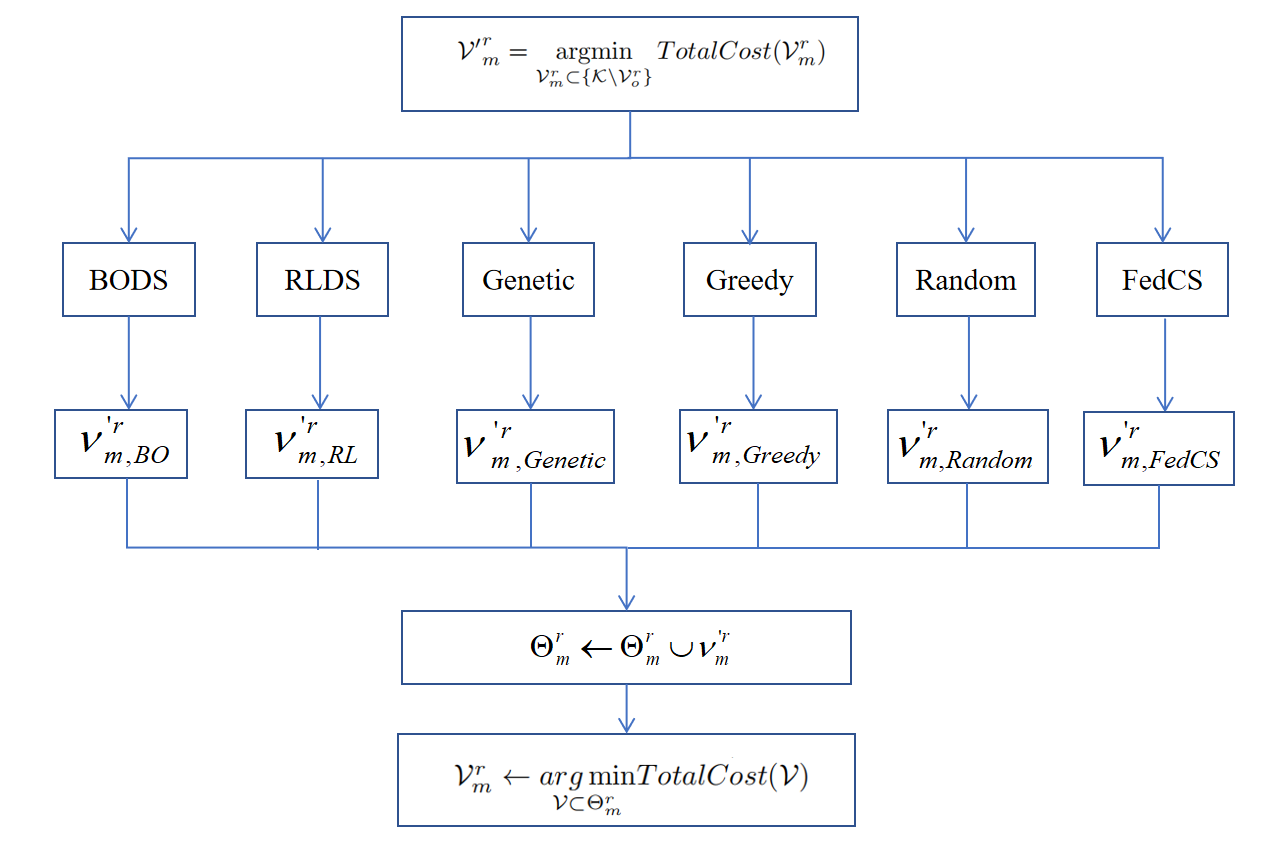}}
\vspace{-2mm}
\caption{The architecture of the Meta-Greedy.}
\label{Meta-Greedy-framework}
\vspace{-2mm}
\end{figure}

\begin{table}[htbp]
  \caption{Experimental Setup of Group A. Size represents the size of training samples and test samples (number of training samples/number of test samples). ``Emnist-L'' represents ``Emnist-Letters'' and ``Emnist-D'' represents ``Emnist-Digitals''.}
  \vspace{-2mm}
  \label{tab:setup-GroupA}
  \begin{tabular}{cccl}
    \toprule
    datasets&  Cifar10& Emnist-L & Emnist-D\\
    \midrule
    Features& 32x32&  28x28& 28x28\\
    Network model& VGG16& CNN& LeNet5\\
    Parameters& 26,233K& 3,785K& 62K\\
    Size & 50k/10k & 124.8k/20.8k & 240k/40k\\
    Local epochs& 5& 5& 5\\
    Mini-batch size& 30& 10& 64\\
  \bottomrule
\end{tabular}
\vspace{-6mm}
\end{table}

We propose a dynamic cost model for the selection in Line \ref{meta:selectPlan}. While Formula \ref{eq:totalCost} can also be exploited for the selection, the influence of the data fairness becomes smaller in the later round of FL as the frequency of the participation has little change with a small increase of one. We reconstruct the cost model using Formula \ref{eq:RefactoredtotalCost} to replace Formula \ref{eq:totalCost}. 
\vspace{-2mm}
\begin{equation}
\vspace{-2mm}
\label{eq:RefactoredtotalCost}
    ReCost(\mathcal{V}_m^r) = \alpha * \mathscr{T}_m^{r}(\mathcal{V}_m^r) + \beta^r * \mathscr{F}_m^{r}(\mathcal{V}_m^r),
\end{equation}
where $\beta^r$ dynamically changes according to $r$ and is defined in Formula \ref{eq:dynamicBeta}.
\vspace{-2mm}
\begin{equation}
\vspace{-2mm}
\label{eq:dynamicBeta}
    \beta^r = \beta*\Omega(r), r \geq 1
\end{equation}
where $\Omega(r)$ is a function based on Round $r$. $\Omega(r)$ should enhance the impact of data fairness when $r$ becomes significant. We take $\Omega(r)=\liu{\sqrt{r}}$ in our algorithm because of its excellent performance (see details in Section \ref{sec:experiment}).
}

\liuR{
\section{Convergence Analysis}
\label{sec:conv}

In this section, we present the theoretical convergence proof for our multi-job FL with arbitrary scheduling methods. We first introduce the assumptions and then present the convergence theorem and corollary. 

\begin{assumption} 
\label{assump:lip}
\textit{Lipschitz gradient: The function $F^m_k$ is $L$-smooth for each device $k \in \mathcal{N}$ i.e., $\parallel \nabla F^m_k(x) - \nabla F^m_k(y) \parallel  \leq  L \parallel x - y \parallel$.}
\end{assumption}

\begin{assumption} 
\label{assump:unbiase}
\textit{Unbiased stochastic gradient: $\mathbb{E}_{\zeta^{m,r}_{k,h} \sim \mathcal{D}_i} [\nabla f^m_k(w^m; \zeta^{m,r}_{k,h})] = \nabla F^m_k(w^m) $.}
\end{assumption}

\begin{assumption} 
\label{assump:localVariance}
\textit{Bounded local variance: For each device $k \in \mathcal{N}$, the variance of its stochastic gradient is bounded: $\mathbb{E}_{\zeta^{m,r}_{k,h} \sim \mathcal{D}_i} \parallel \nabla f^m_k(w^m, \zeta^{m,r}_{k,h}) - \nabla F^m_k(w^m) \parallel^2 \le{\sigma^2}$.}
\end{assumption}

\begin{figure}[t]
\vspace{-4mm}
\begin{algorithm}[H]		
    \caption{Meta-Greedy Scheduling}
    \label{alg:MetaGreedy}
    {\bf{Input:}}\\
    \hspace*{0.3in}$\mathcal{V}_{o}:$ A set of occupied devices\\
    \hspace*{0.3in}$S_m:$ A vector of the frequency of each device sched- \hspace*{0.6in}uled to Job $m$\\
    \hspace*{0.3in}$R_m:$ The maximum round of the current Job $m$\\
    \hspace*{0.3in}$l_m:$ The desired loss value for Job $m$.\\
    {\bf{Output:}}\\
    \hspace*{0.3in} $\mathcal{V}_m = \{\mathcal{V}_{m}^{1},...,\mathcal{V}_{m}^{R_m}\}:$ a set of scheduling plans, \hspace*{0.35in}each with the size $|\mathcal{K}| \times C_m$ 
    
    \begin{algorithmic}[1]
        \State Initialize $\mathcal{V}_m\leftarrow \varnothing$
        \For{$r \in \{1,2,...,R_m\}$ and $l_m$ is not achieved} \label{meta:eachRound}
            \State ${\Theta}_{m}^{r} \leftarrow$ generate a set of candidate scheduling plans using BODS, RLDS, Genetic, Greedy, FedCS and Random within $\mathcal{K}\backslash \mathcal{V}_o$ \label{meta:generatePlans}
            \State $\mathcal{V}_m^r \leftarrow \underset{\mathcal{V} \subset \Theta_m^r}{arg\min ReCost(\mathcal{V})}$ \label{meta:selectPlan}
            \State $\mathcal{V}_m \leftarrow \mathcal{V}_m \cup \mathcal{V}_m^r$ \label{meta:usePlan}
        \EndFor
   \end{algorithmic}
\end{algorithm}
\vspace{-8mm}
\end{figure}

\begin{table}[htbp]
  \caption{Experimental Setup of Group B. Size represents the size of training samples and test samples (number of training samples/number of test samples).}
  \vspace{-2mm}
  \setlength\tabcolsep{5pt}
  \label{tab:setup-GroupB}
  \begin{tabular}{cccl}
    \toprule
    datasets&   Fashion\_{mnist}&  Cifar10&  Mnist\\
    \midrule
    Features& 28x28&  32x32& 28x28\\
    Network model& CNN& ResNet18& AlexNet\\
    Parameters& 225K& 598K& 3,275K\\
    Size & 60K/10K & 50K/10K & 60K/10K\\
    Local epochs& 5& 5& 5\\
    Mini-batch size& 10& 30& 64\\
  \bottomrule
\end{tabular}
\vspace{-6mm}
\end{table}

\begin{assumption} 
\label{assump:sequared}
\textit{Bounded local gradient: For each device $k \in \mathcal{N}$, the expected squared of stochastic gradient is bounded: $\mathbb{E}_{\zeta^{m,r}_{k,h} \sim \mathcal{D}_i}\parallel \nabla f^m_k(w^m, \zeta^{m,r}_{k,h}) \parallel^2 \le{G^2}$.}
\end{assumption}

Assumption \ref{assump:lip} has been made in \cite{li2019convergence,Li2022FedHiSyn}, while Assumptions \ref{assump:unbiase}, \ref{assump:localVariance}, \ref{assump:sequared}, have been exploited in \cite{reddi2020adaptive}. When Assumptions \ref{assump:lip} - \ref{assump:sequared} hold, we get the following theorem.

\begin{theorem}
Suppose that Assumptions 1 to 4 hold, and consider that $F^m_k$ is a non-convex function. When the learning rate satisfies $0 < \eta^m \leq \frac{1}{L}$, then for all $R \ge 1$ we have:
\begin{equation}
\label{eq:eqTheorem}
\begin{aligned}
    &\frac{1}{RH} \sum^{R}_{r=1} \sum^H_{h=1} \mathbb{E} \parallel \nabla 
    F^m(\bar{w}^m_{r,h}) \parallel^2 \\
    \leq &~\frac{2}{{\eta^m_{r,h}}RH}(F^m(\bar{w}^m_{1,1}) - {F^m}^*) + L {\eta^m_{r,h}} \sigma^2 \\
    & + L^2 Q^2 (H-1)^2 {\eta^m_{r,h}}^2 G^2 
\end{aligned}
\end{equation}
where ${F^m}^*$ is the local optimal value, $R$ refers to the number of rounds for Job $m$ during the execution, $H$ represents the number of local iterations, and $Q$ is the upper bound ratio between the learning rate at local iteration $0$ and $h'$, i.e., $\eta^m_{r, 0} \leq Q \eta^m_{r, h'}$.
\end{theorem}

\begin{proof}
The proof can be found in Appendix.
\end{proof}

Then, we can get the following corollary:

\begin{table*}[ht]
  \centering
  \caption{The convergence accuracy and the time required to achieve the target accuracy for different methods in Group A. The numbers in parentheses represent the target accuracy, and "/" represents that the target accuracy is not achieved.}
  \vspace{-2mm}
  \setlength\tabcolsep{1.5pt}
  \label{tab:GroupA}
  \begin{tabular}{cccccccc cccccccc}
    \toprule
    &\multicolumn{7}{c}{Convergence Accuracy} & &\multicolumn{7}{c}{Time (min)}\\
    \cline{2-8} \cline{10-16}
    &  Random& Genetic& FedCS& Greedy& BODS& RLDS& Meta-Greedy  & & Random& Genetic& FedCS& Greedy& BODS& RLDS& Meta-Greedy\\
    \cline{1-16}
    &\multicolumn{15}{c}{Non-IID}\\
    
    \midrule
    VGG& 0.55&	0.54& 0.55&	0.43& 0.57& 0.56& \bf{0.577}& VGG (0.55)& 2486& 1164.3& 1498.5& /& 455.1& 406.8& \bf{291.3}\\
    Cnn& \bf{0.90}&	0.80& 0.80&	0.83& \bf{0.90}& 0.897& \bf{0.90}& Cnn (0.80)& 44.25& 95.85& 27.39& 43.04& 15.88& 17.6& \bf{11.9}\\
    LeNet& 0.990& 0.988& 0.990& 0.986& \bf{0.991}& \bf{0.991}& \bf{0.991}& LeNet (0.984)& 43.81& 30.15&	33.37& 43.76& 28.84& 22.54& \bf{12.37}\\
    
     \cline{1-16}
    &\multicolumn{15}{c}{IID}\\
    \cline{1-16}
    VGG& \bf{0.610}&	0.558&	0.603&	0.522&	0.603&	0.605&	0.602& VGG (0.55)& 126.9& 231.4& 87.5& /& 57.7& 43.81& \bf{30.34}\\
    Cnn& \bf{0.943}& 0.928&	0.942&	0.928&	\bf{0.943}&	0.935&	0.936& Cnn (0.930)& 52.05& 176.85& 27.45& 26.48& 19.25& 13.0& \bf{12.88}\\
    LeNet&	0.9945&	0.9928&	0.9934&	0.990&	\bf{0.9946}&	\bf{0.9946}& 0.9936& LeNet (0.99)& 14.94& 6.03& 7.12& 17.98&	5.18& 4.02& \bf{2.02}\\
  \bottomrule
\end{tabular}
\vspace{-2mm}
\end{table*}

\begin{table*}[tb]
  \centering
  \caption{The convergence accuracy and the time required to achieve the target accuracy for different methods in Group B. The numbers in parentheses represent the target accuracy, and "/" represents that the target accuracy is not achieved.}
  \vspace{-2mm}
  \setlength\tabcolsep{1pt}
  \label{tab:GroupB}
  \begin{tabular}{cccccccc cccccccc}
    \toprule
    &\multicolumn{7}{c}{Convergence Accuracy} & &\multicolumn{7}{c}{Time (min)}\\
    \cline{2-8} \cline{10-16}
    & Random& Genetic& FedCS& Greedy& BODS& RLDS& Meta-Greedy & & Random& Genetic& FedCS& Greedy& BODS& RLDS& Meta-Greedy\\
    \cline{1-16}
    &\multicolumn{15}{c}{Non-IID}\\
    
    \midrule
    ResNet&	0.546& 0.489& 0.523& 0.403&	0.583&	0.562&	\bf{0.590}& ResNet (0.50)& 852.9& 621.5& 402.3& /& 219.8& 168.9& \bf{149.8}\\
    Cnn& 0.824&	0.767&	0.823&	0.764&	0.836&	0.830&	\bf{0.845}& Cnn (0.73)& 47.1& 22.0& 	18.5& 70.8& 13.8& 15.5& \bf{13.0}\\
    AlexNet& 0.989&	0.986&	0.987&	0.871&	0.990&	0.990&	\bf{0.990}& AlexNet (0.976)&  140.57& 	60.0& 	82.87& 	181.2& 59.2& 53.87& \bf{45.0}\\
    
    \cline{1-16}
    &\multicolumn{15}{c}{IID}\\
    \cline{1-16}
    ResNet&	0.787&	0.754& 0.782& 0.743& 0.791& 0.785& \bf{0.799}& ResNet (0.740)& 65.93& 32.51& 31.4& 52.93& 15.9& 12.82& \bf{10.6}\\
    Cnn& 0.868&	0.867& 0.868& 0.868& 0.869& 0.869& \bf{0.871}& Cnn (0.867)& 120.19& 38.99& 89.6& 36.13& 32.83&	19.7& \bf{16.13}\\
    AlexNet& 0.9938& 0.9938& 0.9939& 0.9935& 0.9939& \bf{0.9943}& 0.9940& AlexNet (0.9933)& 35.08& 19.44& 20.97& /& 21.65& 12.9& \bf{10.8}\\
  \bottomrule
\end{tabular}
\vspace{-2mm}
\end{table*}

\zhou{
\begin{table*}[htbp]
  \centering
  \caption{The time required to achieve the target accuracy for jobs executed sequentially with FedAvg. ``*'' indicates that it fails to achieve the target accuracy.}
  \vspace{-2mm}
  \setlength\tabcolsep{10pt}
  \renewcommand{\arraystretch}{1.3}
  \label{tab:Signal}
  \begin{tabular}{cccc cccc}
    \toprule
     & \multicolumn{3}{c}{non-IID/IID}&  & \multicolumn{3}{c}{non-IID/IID}\\
    \cline{2-4} \cline{6-8}
    Job&  VGG& CNN& LeNet& & ResNet& CNN& AlexNet\\
    \cline{1-8}
    Target Accuracy&  0.55/0.55& 0.80/0.93& 0.984/0.99& &
       0.50/0.74& 0.73/0.867& 0.976/0.9933\\
    \midrule
    Time (min)&  2483.4/133.3& 53.1/45.5& 50.5/18.01&  &
               897.2/*& 35.8/322.6& 115.8/65.16\\
 \bottomrule
 \vspace{-9mm}
 \end{tabular}
\end{table*}
}

\begin{corollary}
When we choose $\eta^m_r = \frac{1}{L\sqrt{RH}}$ and $Q \leq (RH)^{\frac{1}{4}}$, we have:
\begin{equation}
\label{eq:eqCorollary}
\begin{aligned}
&\frac{1}{RH} \sum^{R}_{r=1} \sum^H_{h=1} \mathbb{E} \parallel \nabla F^m(\bar{w}^m_{r,h}) \parallel^2 \\
\leq &~\frac{2L}{\sqrt{RH}}(F^m(\bar{w}^m_{1,1}) - {F^m}^*) + \frac{1}{\sqrt{RH}} \sigma^2
    + \frac{1}{\sqrt{RH}} (H-1)^2 G^2 \\
= &~\mathcal{O}(\frac{1}{\sqrt{RH}})
\end{aligned}
\end{equation}
\end{corollary}
We can find that, the training process of multi-job FL converges to a stationary point of $f(w^*)$ with a convergence rate of $\mathcal{O}(\frac{1}{\sqrt{RH}})$ for each job.
}

\section{Experiments}
\label{sec:experiment}
In this section, we present the experimental results to show the efficiency of our proposed scheduling methods within MJ-FL. We compared the performance of \liu{Meta-Greedy, RLDS, and BODS} with six baseline methods, i.e., Random \cite{McMahan2017Communication-efficien}, FedCS \cite{Nishio2019Client}, Genetic \cite{barika2019scheduling}, Greedy \cite{shi2020device}, \liu{Deep Neural Network (DNN), and Simulated Annealing (SA)} \cite{van1987simulated}.

\subsection{Federated Learning Setups}

In the experiment, we take three jobs as a group to be executed in parallel. 
We carry out the experiments with two groups, i.e., Group A with VGG-16 (VGG) \cite{simonyan2015very}, CNN (CNN-A-IID and CNN-A-non-IID) \cite{lecun1998gradient}, and LeNet-5 (LeNet) \cite{lecun1998gradient}, and Group B with Resnet-18 (ResNet) \cite{He2016}, CNN (CNN-B) \cite{lecun1998gradient}, and Alexnet \cite{Krizhevsky2012ImageNet}, while each model corresponds to one job. The complexity of the models is as follows: AlexNet $<$ CNN-B $<$ ResNet and LeNet $<$ CNN (CNN-A-IID and CNN-A-non-IID) $<$ VGG. We exploit the datasets of Cifar-10 \cite{krizhevsky2009learning}, emnist-letters \cite{cohen2017emnist}, emnist-digital \cite{cohen2017emnist}, Fashion-MNIST \cite{xiao2017fashion}, and MNIST \cite{lecun1998gradient} in the training process.

CNN-A-IID comprises of two $3\times 3$ convolution layers, one with 32 channels and the other with 64 channels. Each layer is followed by one batch normalization layer and $2\times 2$ max pooling. Then, there are one flatten layer and three fully-connected layers (1568, 784, and 26 units) after the two convolution layers. In addition, we make a simple modification of CNN-A-IID to CNN-A-non-IID since the convergence behavior of CNN on non-IID in Group A is not satisfiable. CNN-A-non-IID consists of three $3\times 3$ convolution layers (32, 64, 64 channels, each of them exploits ReLU activations, and each of the first two convolution layers is followed by $2\times 2$  max pooling), followed by one flatten layer and two fully-connected layers (64, 26 units). CNN-B consists of two $2\times 2$ convolution layers (64, 32 channels, each of them exploits ReLU activations) followed by a flatten layer and a fully-connected layer, and each convolution layer is followed by a dropout layer with 0.05. In addition, the other parameters are shown in Tables \ref{tab:setup-GroupA} and \ref{tab:setup-GroupB}.

\begin{figure*}[htbp]
\centering
\begin{subfigure}{0.3\linewidth}
\includegraphics[width=\linewidth]{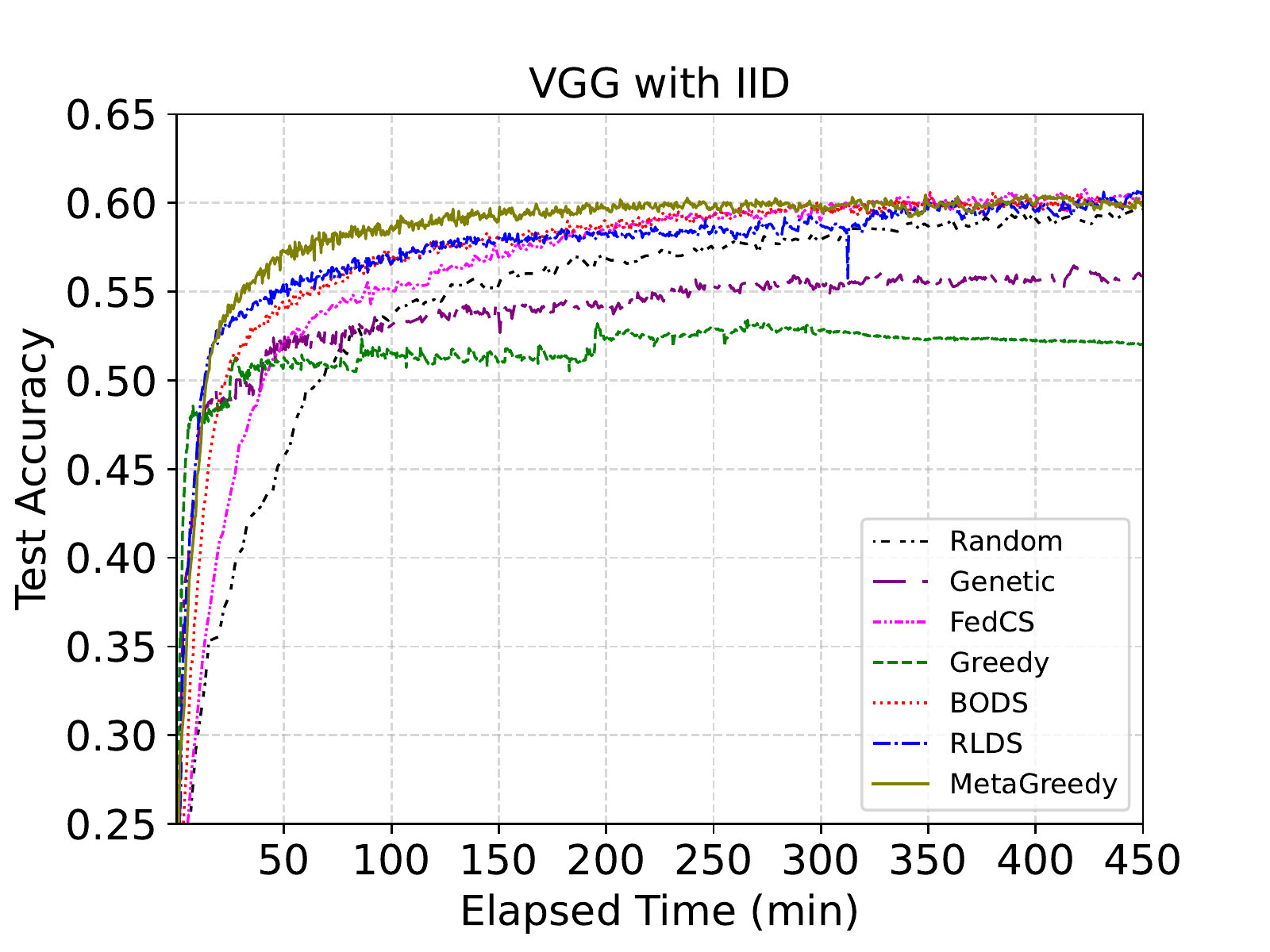}
\vspace{-2mm}
\caption{}
\label{figVgiid}
\end{subfigure}
\begin{subfigure}{0.3\linewidth}
\includegraphics[width=\linewidth]{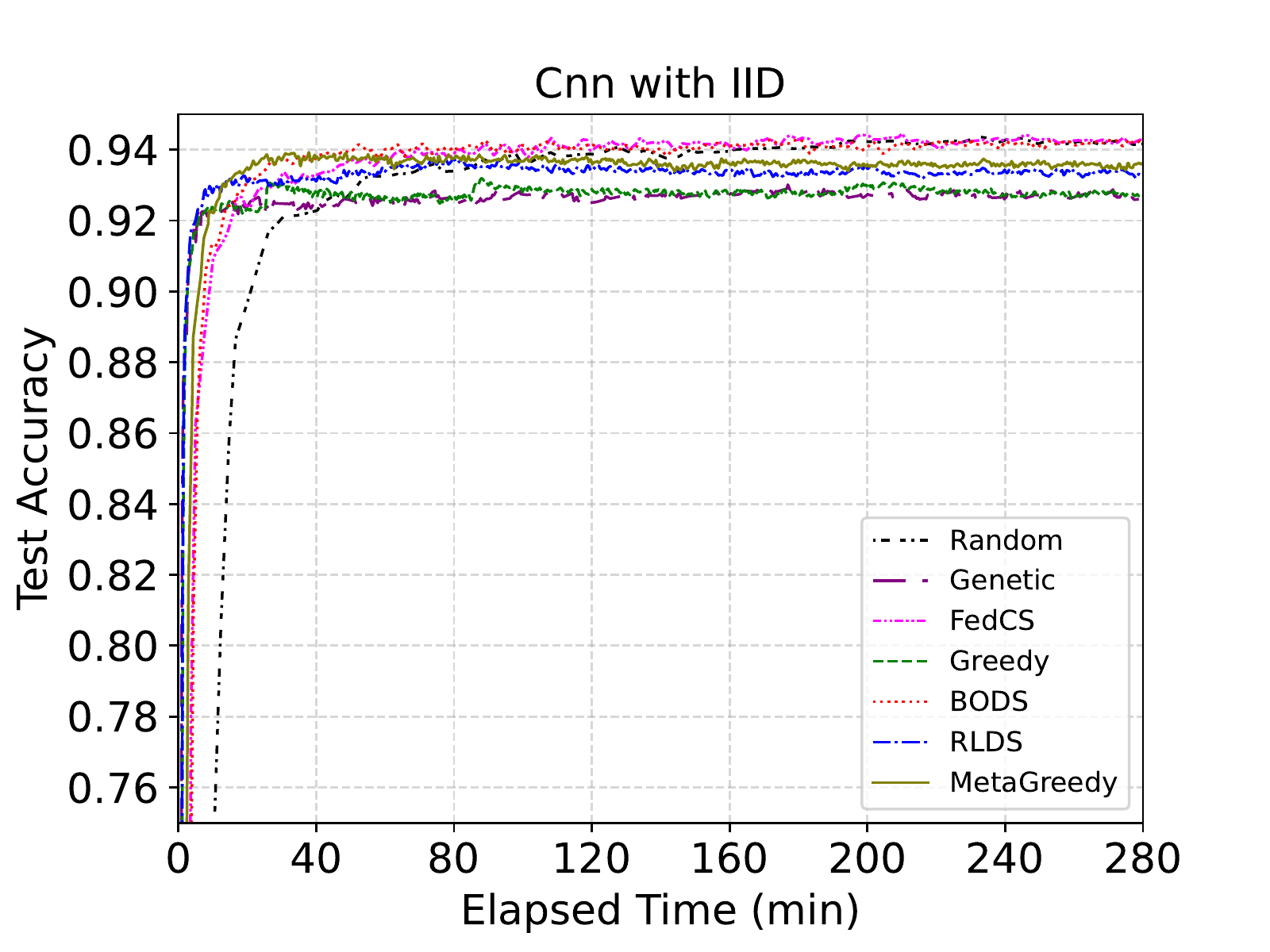}
\vspace{-4mm}
\caption{}
\label{figCiid}
\end{subfigure}
\begin{subfigure}{0.3\linewidth}
\includegraphics[width=\linewidth]{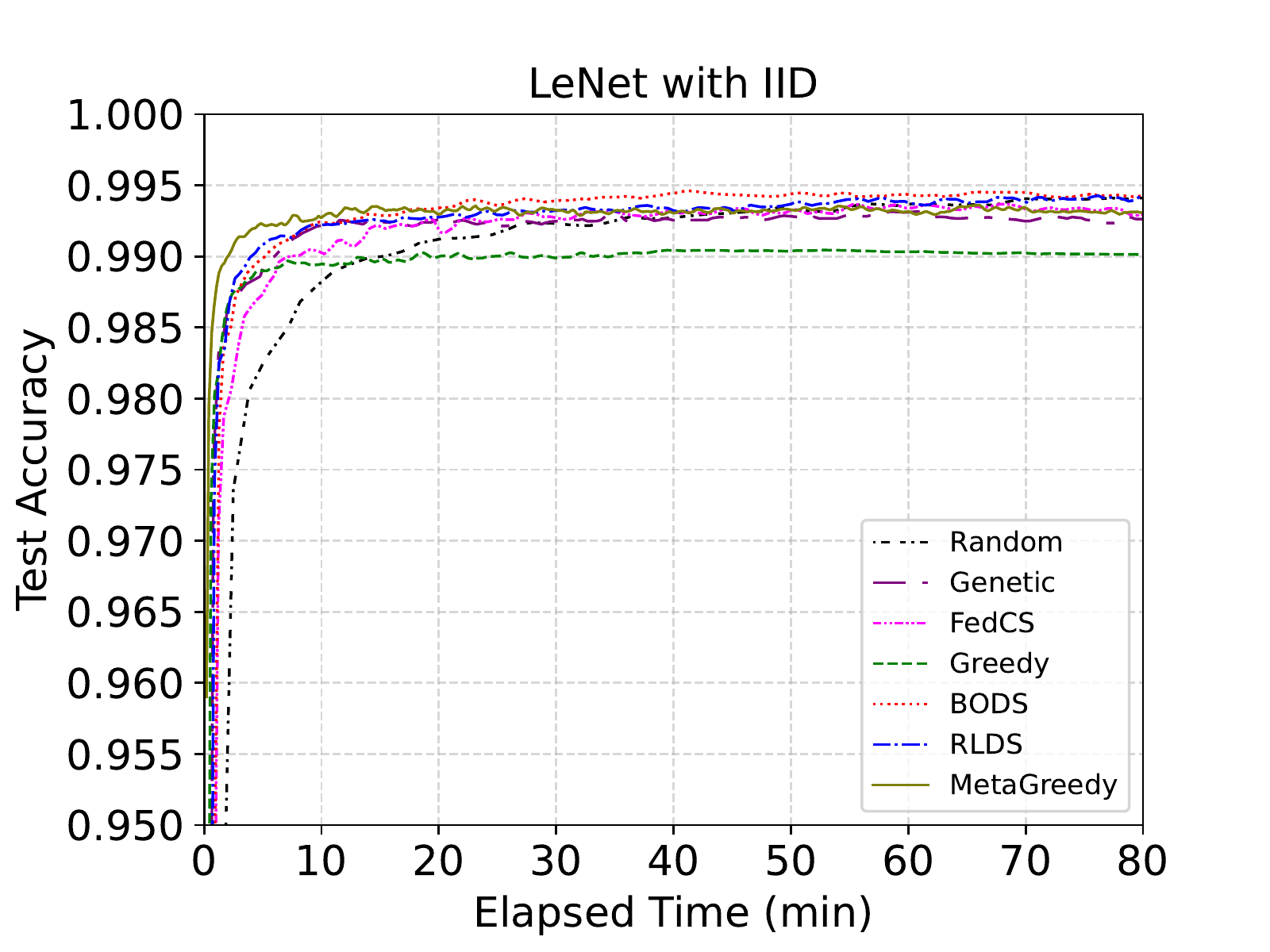}
\vspace{-4mm}
\caption{}
\label{figLeiid}
\end{subfigure}
\vspace{-2mm}
\caption{The convergence accuracy of different jobs in Group A changes over time with the IID distribution.}
\vspace{-5mm}
\label{fig:detail-on-IID-GroupA}
\end{figure*}

\begin{figure*}[htbp]
\centering
\begin{subfigure}{0.3\linewidth}
\includegraphics[width=\linewidth]{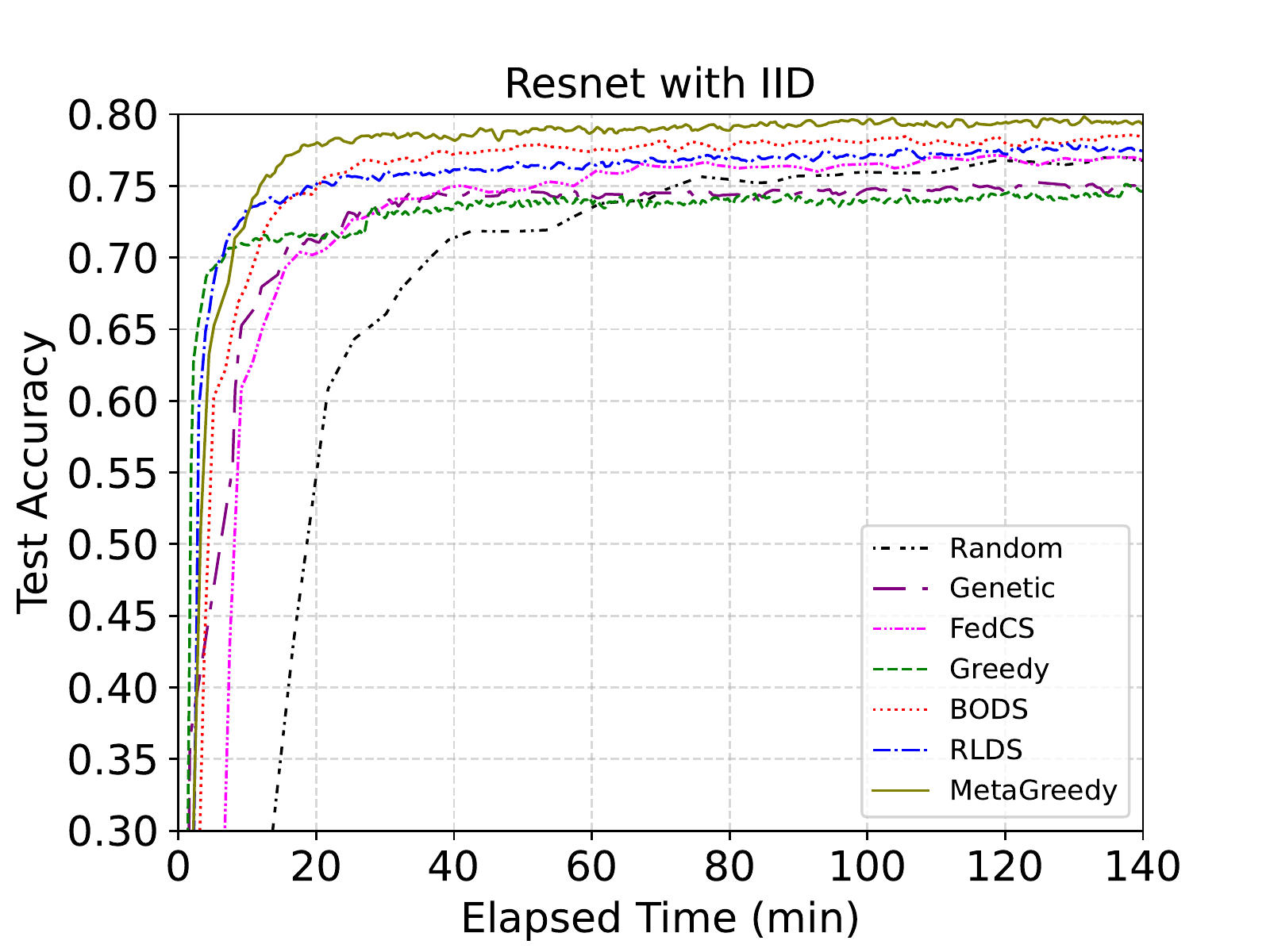}
\vspace{-4mm}
\caption{}
\label{figReiid}
\end{subfigure}
\begin{subfigure}{0.3\linewidth}
\includegraphics[width=\linewidth]{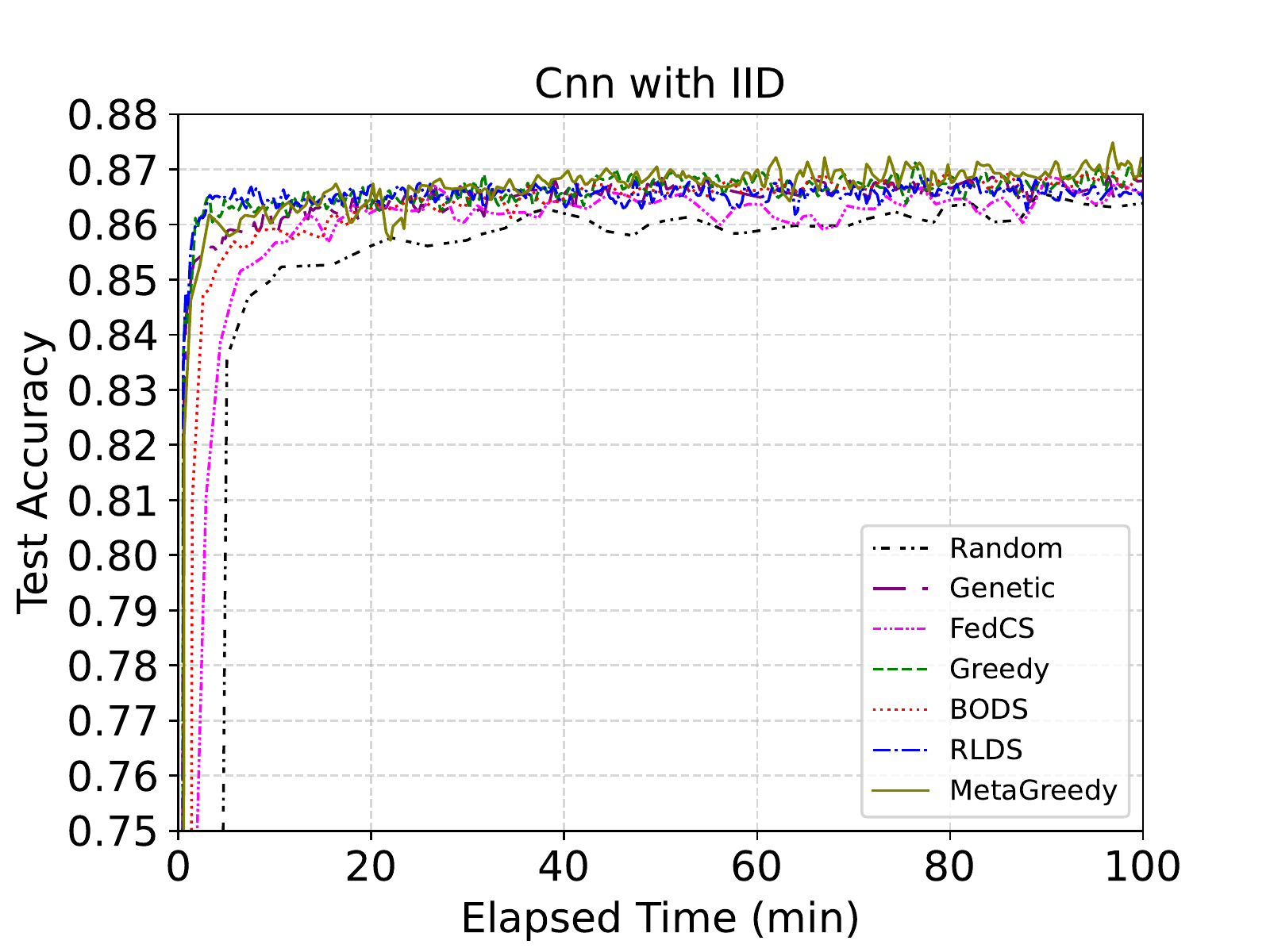}
\vspace{-4mm}
\caption{}
\label{figCniid}
\end{subfigure}
\begin{subfigure}{0.3\linewidth}
\includegraphics[width=\linewidth]{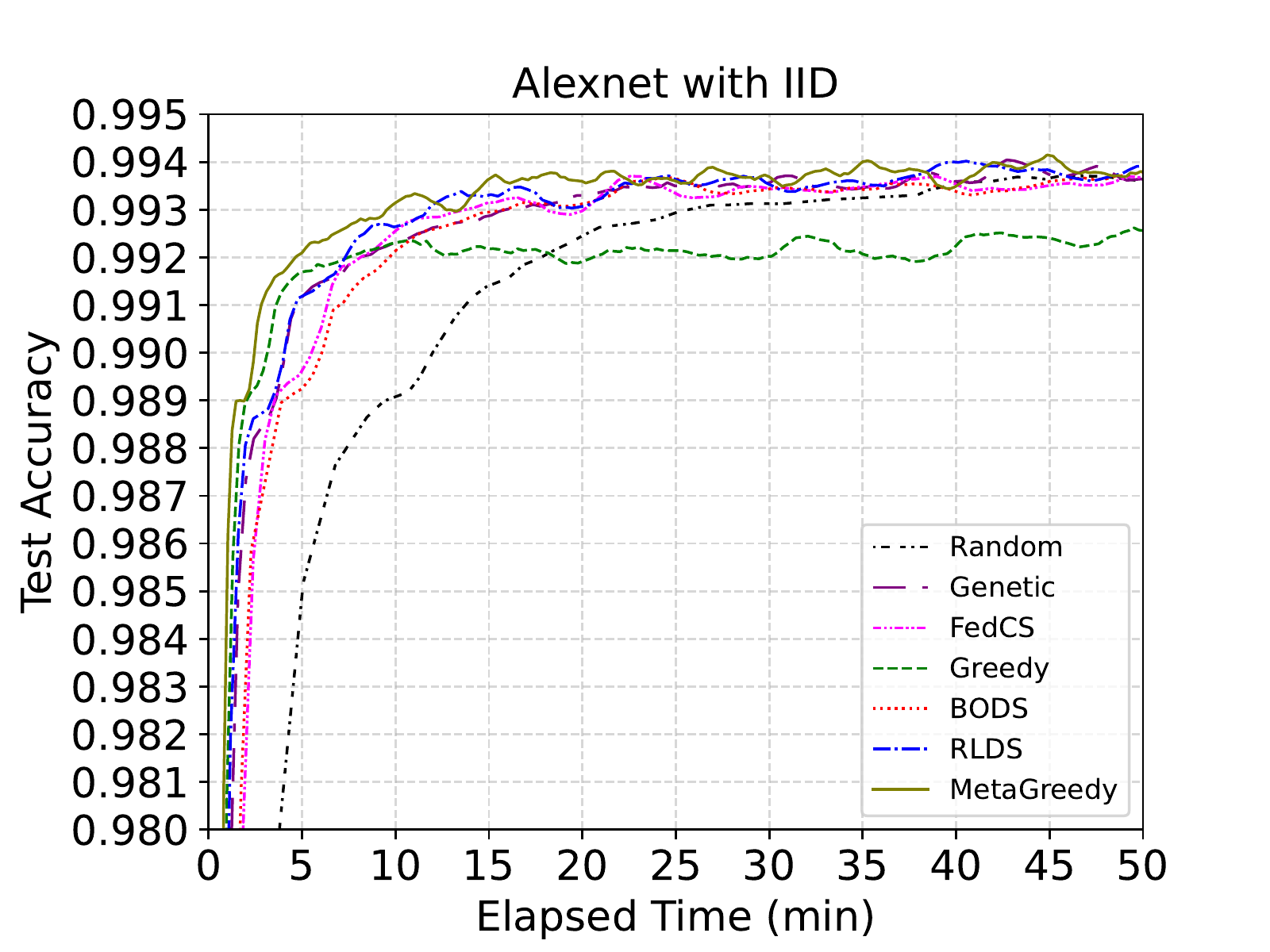}
\vspace{-4mm}
\caption{}
\label{figAliid}
\end{subfigure}
\vspace{-2mm}
\caption{The convergence accuracy of different jobs in Group B changes over time with the IID distribution.}
\vspace{-5mm}
\label{fig:detail-IID-GroupB}
\end{figure*}

\begin{figure*}[htbp]
\centering
\begin{subfigure}{0.3\linewidth}
\includegraphics[width=\linewidth]{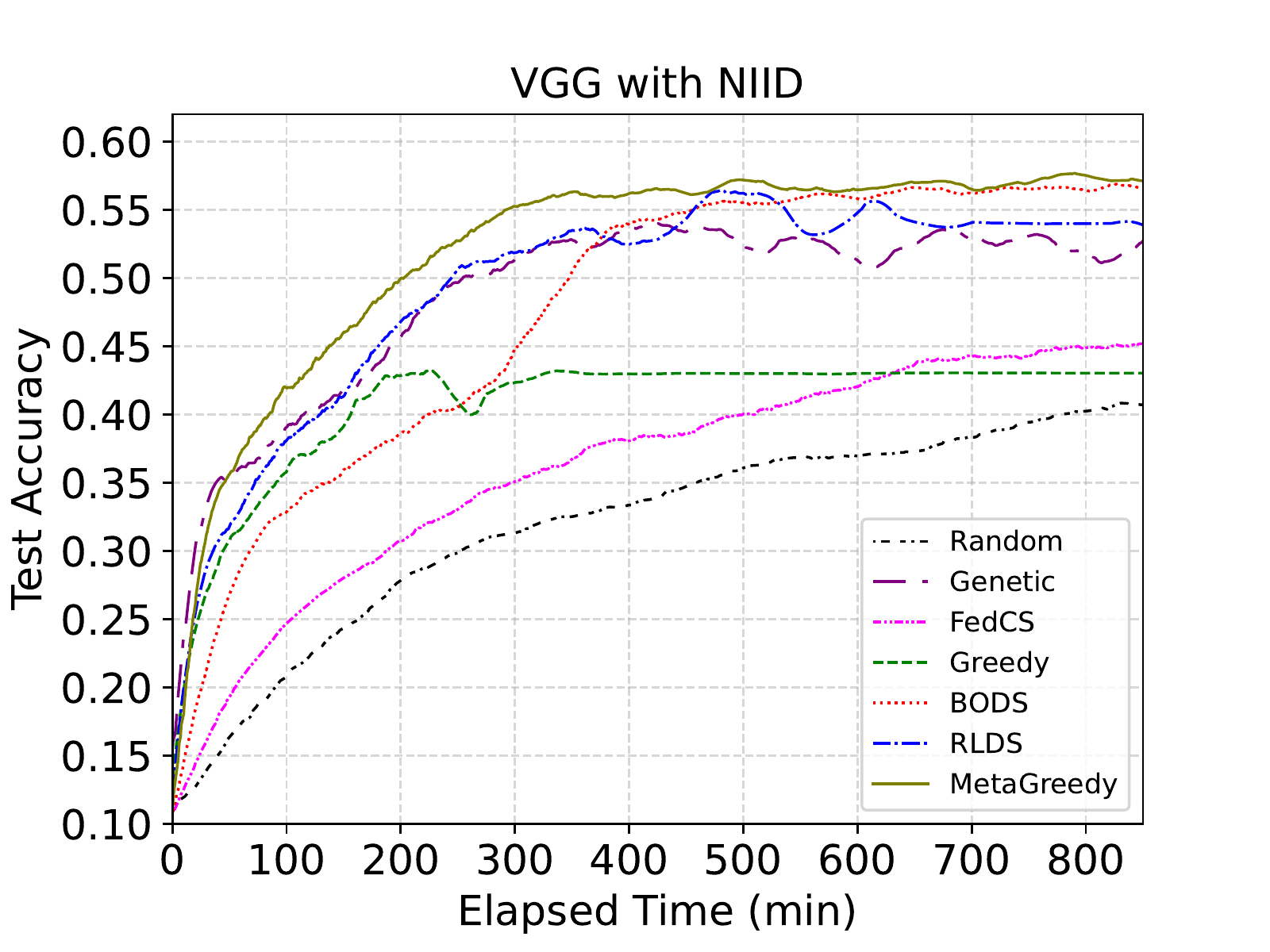}
\caption{}
\label{figVgNiid}
\end{subfigure}
\begin{subfigure}{0.3\linewidth}
\includegraphics[width=\linewidth]{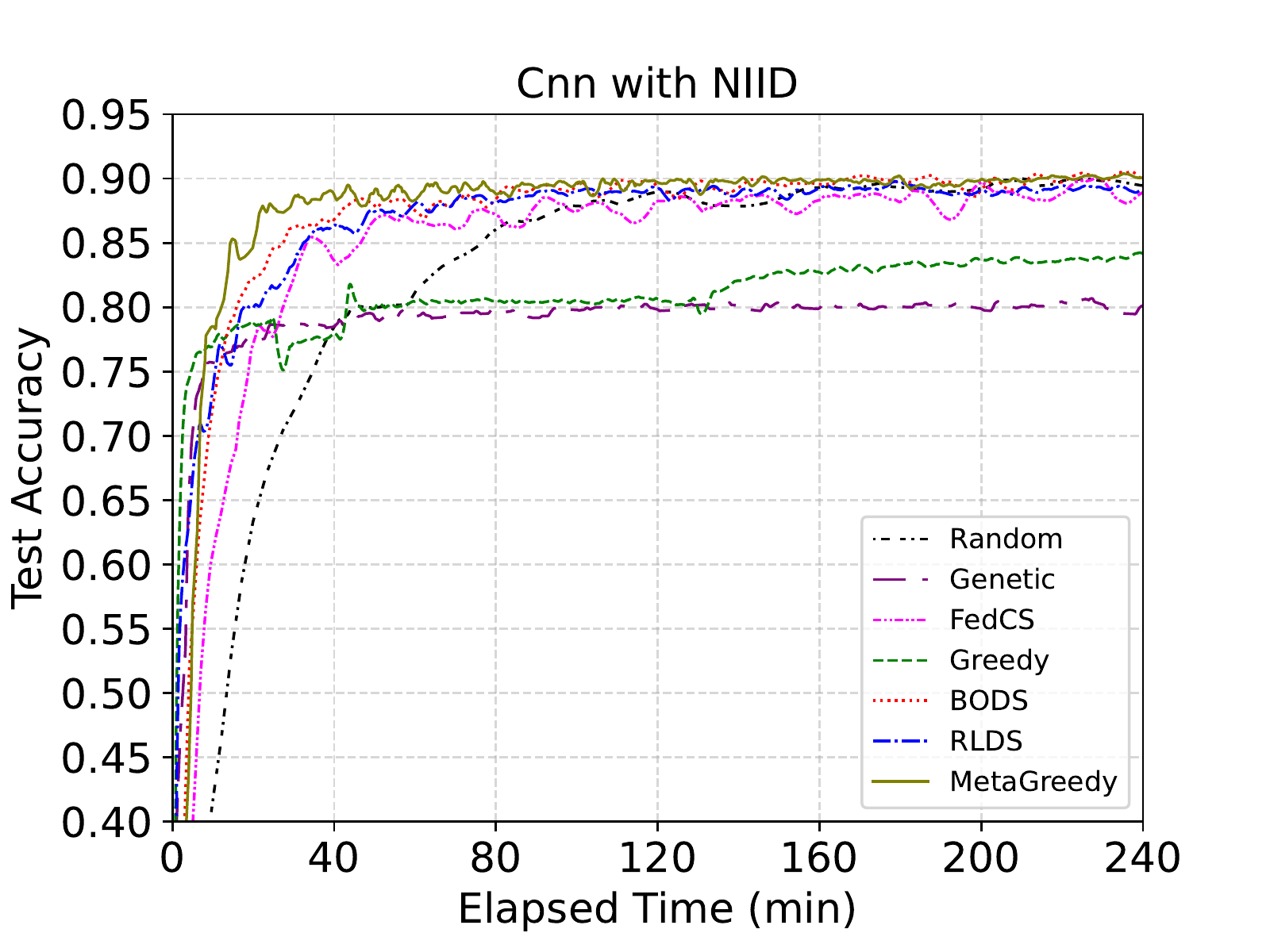}
\caption{}
\label{figCNiid}
\end{subfigure}
\begin{subfigure}{0.3\linewidth}
\includegraphics[width=\linewidth]{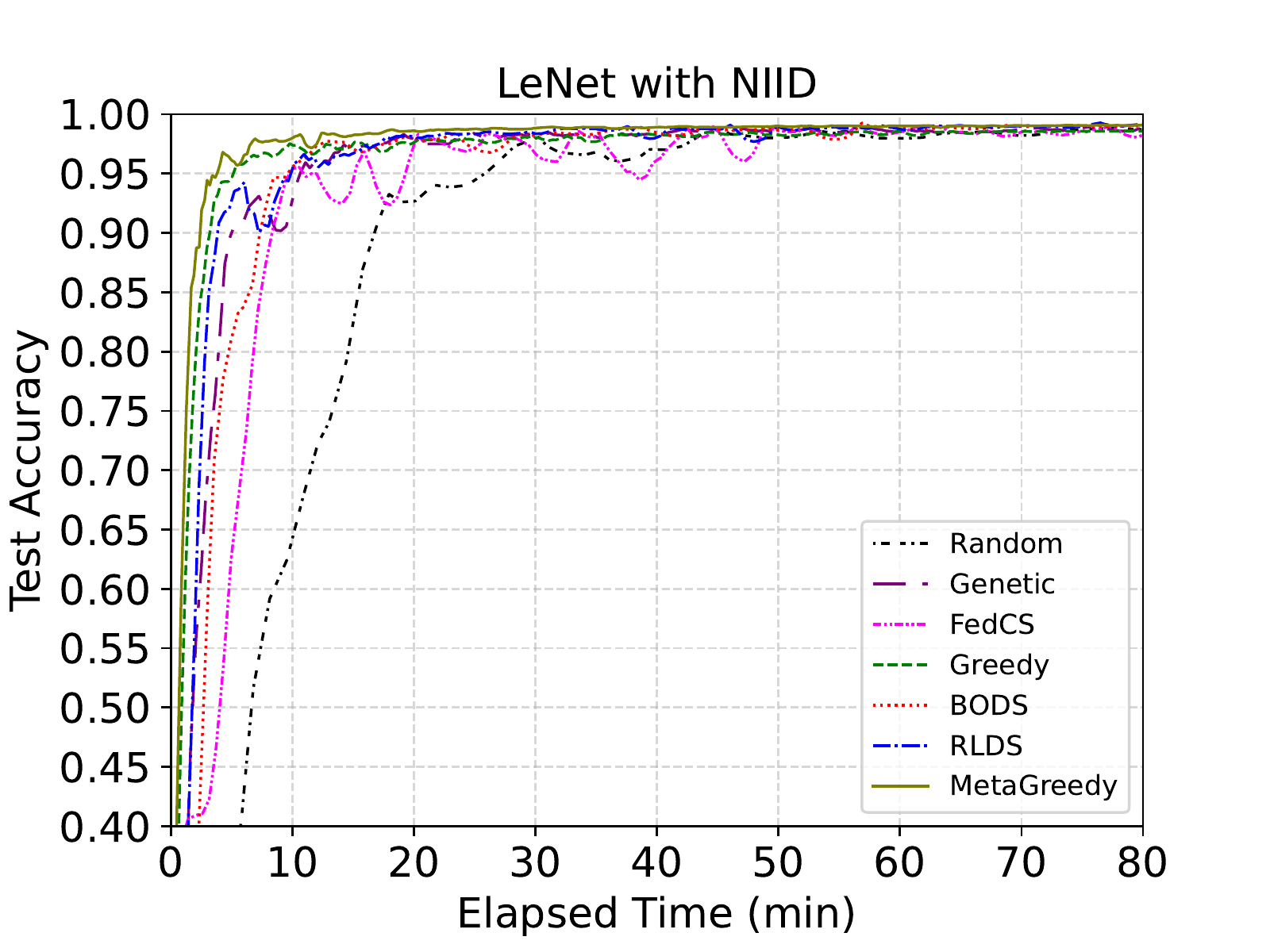}
\caption{}
\label{figLeNiid}
\end{subfigure}
\vspace{-3mm}
\caption{The convergence accuracy of different jobs in Group A changes over time with the non-IID distribution.}
\vspace{-4mm}
\label{fig:groupANonIID}
\end{figure*}

DNN comprises of a flatten layer and 40 hidden layers before an Alpha Dropout layer and a fully-connected layer. The hidden layers consist of 20 units of a fully-connected layer and a batch normalization layer.

For the non-IID setting of each dataset, the training set is classified by category, and the samples of each category are divided into 20 parts. Each device randomly selects two categories and then selects one part from each category to produce its local training set. For the IID setting, each device randomly samples a specified number of images from each training set. 
In addition, we use 12 Tesla V100 GPUs to simulate an FL environment composed of a parameter server and 100 devices. \liuR{We exploit the execution time, including the training time and communication time, measured in real edge devices as the parameters in Formula \ref{eq:distribution} to simulate the real execution.} We use Formula \ref{eq:distribution} to simulate the capabilities of devices in terms of training time with the uniform sampling strategy, while the accuracy is the results from the actual training processes. In the experimentation, we use corresponding target accuracy (for ease of comparison) in the place of target loss value.

\begin{figure*}[htbp]
\centering
\begin{subfigure}{0.3\linewidth}
\includegraphics[width=\linewidth]{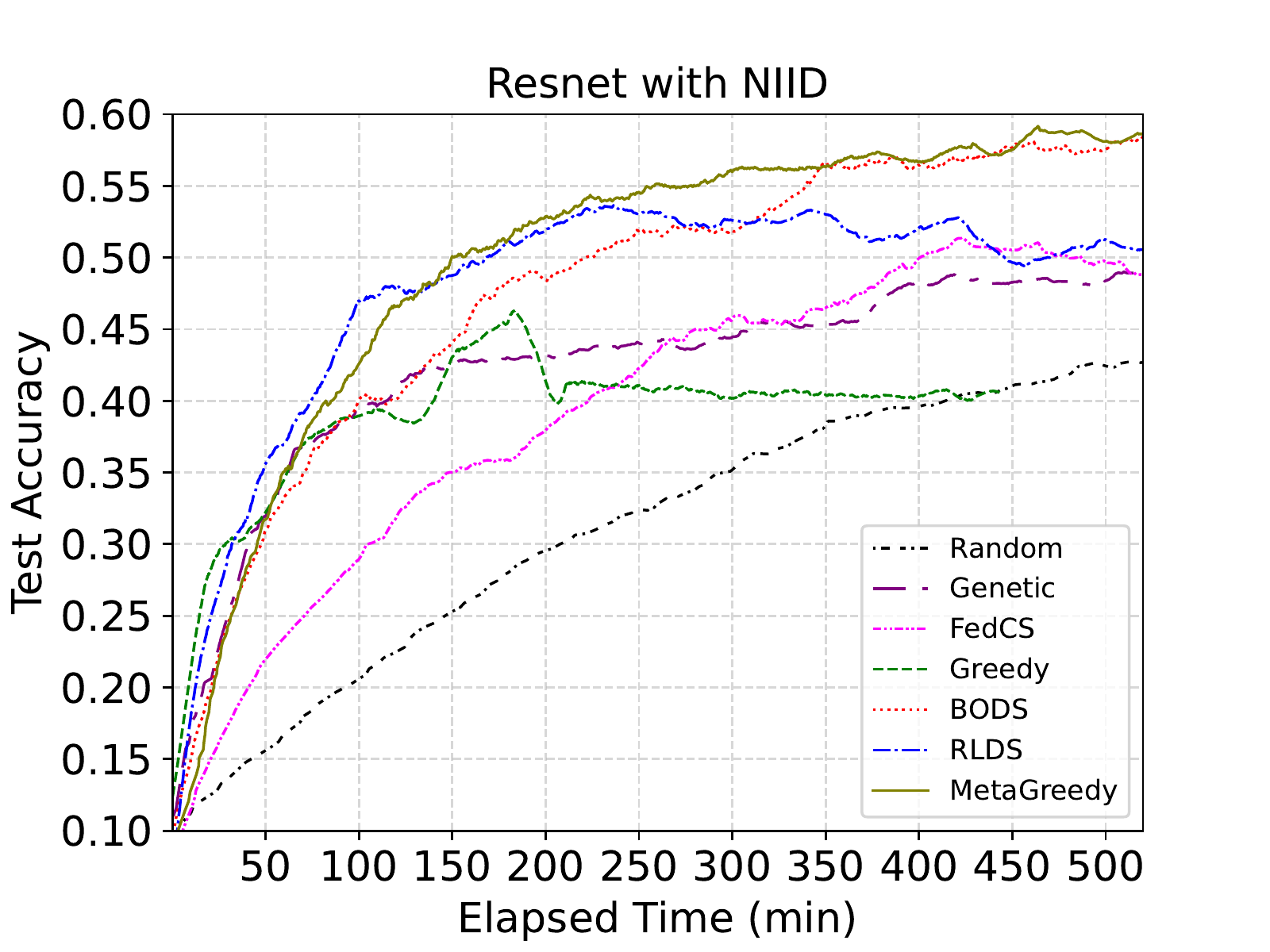}
\vspace{-4mm}
\caption{}
\label{figReNiid}
\end{subfigure}
\begin{subfigure}{0.3\linewidth}
\includegraphics[width=\linewidth]{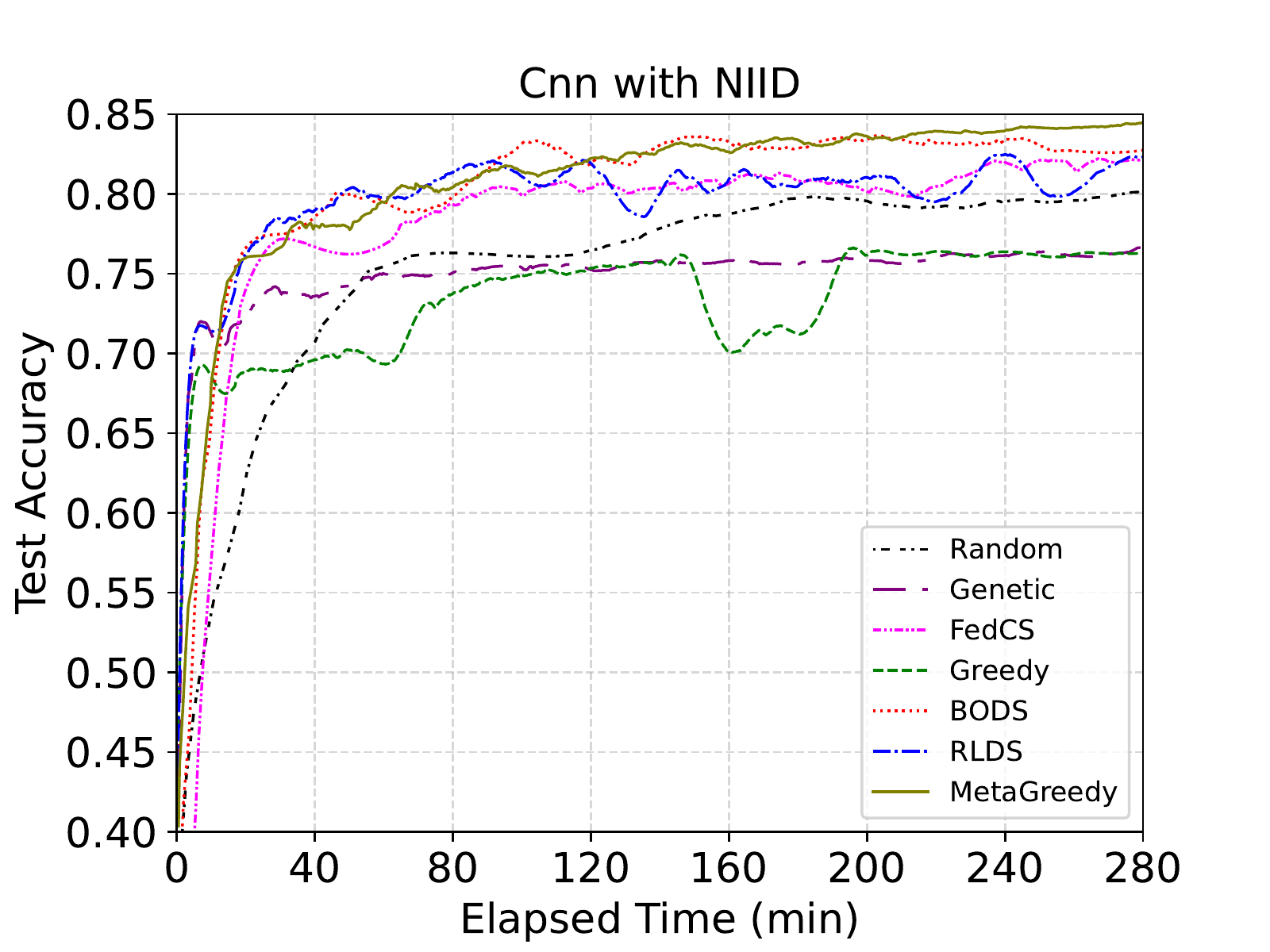}
\vspace{-4mm}
\caption{}
\label{figCnNiid}
\end{subfigure}
\begin{subfigure}{0.3\linewidth}
\includegraphics[width=\linewidth]{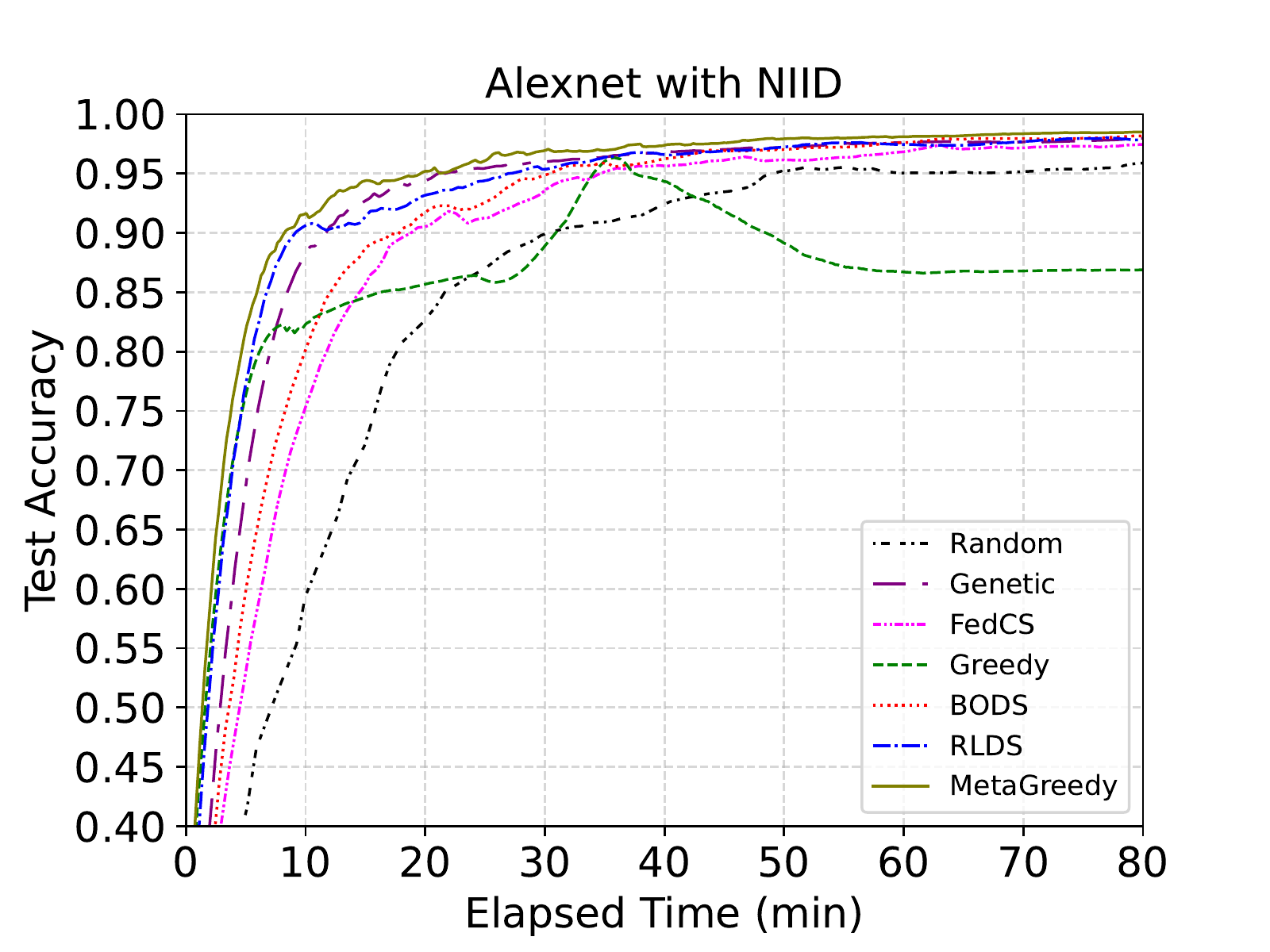}
\vspace{-4mm}
\caption{}
\label{figAlNiid}
\end{subfigure}
\vspace{-2mm}
\caption{The convergence accuracy of different jobs in Group B changes over time with the non-IID distribution.}
\label{fig:groupBNonIID}
\vspace{-5mm}
\end{figure*}

\begin{figure*}[htbp]
\centering
\begin{subfigure}{0.3\linewidth}
\includegraphics[width=\linewidth]{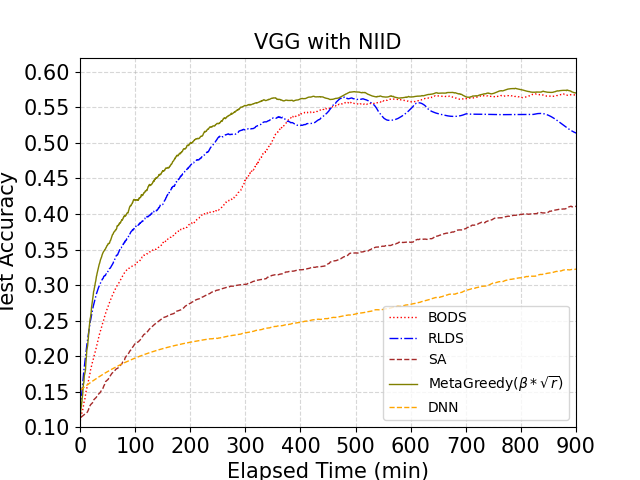}
\vspace{-4mm}
\caption{}
\label{figVgNiidAb-other}
\end{subfigure}
\begin{subfigure}{0.3\linewidth}
\includegraphics[width=\linewidth]{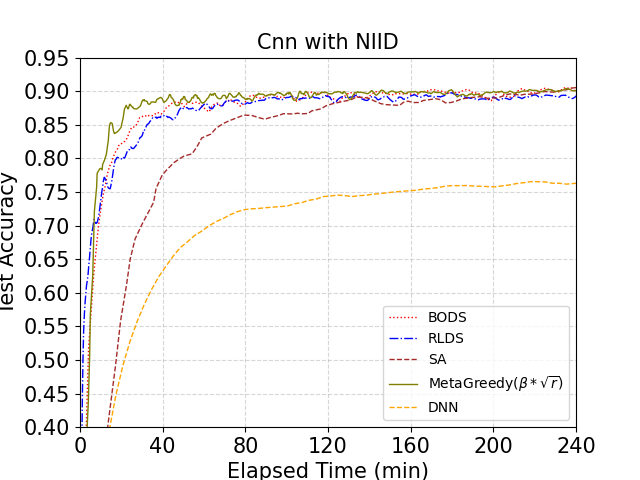}
\vspace{-4mm}
\caption{}
\label{figCnNiidAb-other}
\end{subfigure}
\begin{subfigure}{0.3\linewidth}
\includegraphics[width=\linewidth]{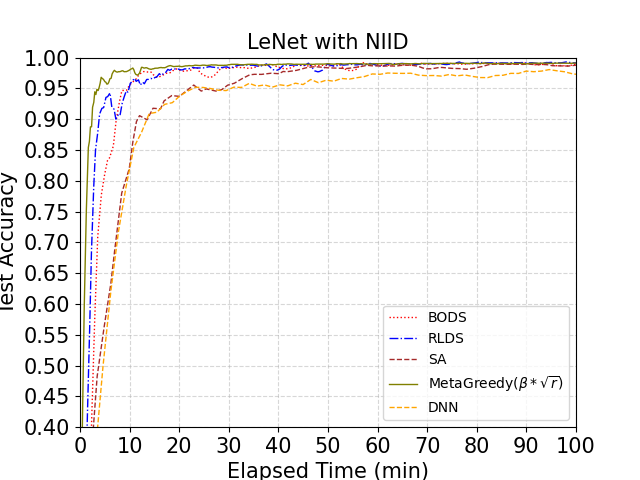}
\vspace{-4mm}
\caption{}
\label{figLeNiidAb-other}
\end{subfigure}
\begin{subfigure}{0.3\linewidth}
\includegraphics[width=\linewidth]{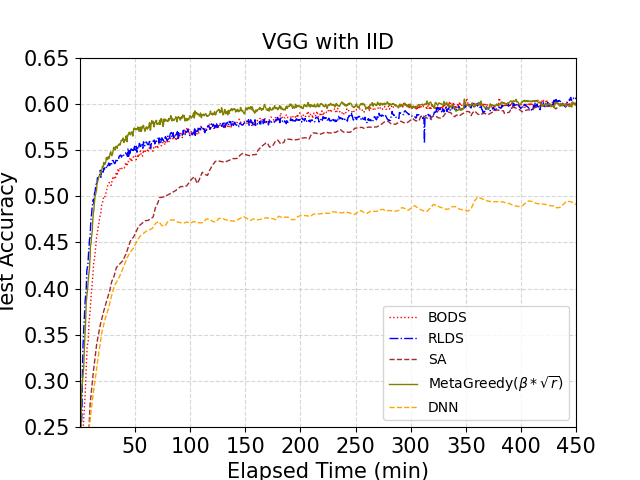}
\vspace{-4mm}
\caption{}
\label{figVgiidAb-other}
\end{subfigure}
\begin{subfigure}{0.3\linewidth}
\includegraphics[width=\linewidth]{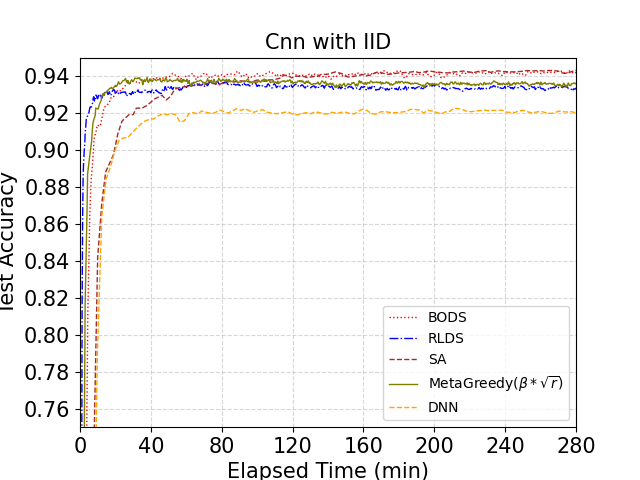}
\vspace{-4mm}
\caption{}
\label{figCniidAb-other}
\end{subfigure}
\begin{subfigure}{0.3\linewidth}
\includegraphics[width=\linewidth]{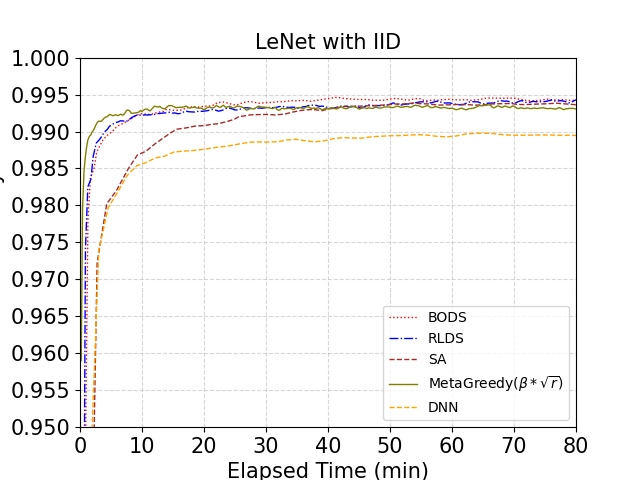}
\vspace{-4mm}
\caption{}
\label{figLeiidAb-other}
\end{subfigure}
\vspace{-2mm}
\caption{The convergence accuracy of different jobs changes over time compared with DNN and SA}
\vspace{-4mm}
\label{fig:DNN and SA}
\end{figure*}

\liu{\subsection{Experimental Results}

We first present the comparison between MJ-FL and Single-Job FL (SJ-FL).
Then, we compare the proposed methods with baselines, e.g., Random \cite{McMahan2017Communication-efficien}, FedCS \cite{Nishio2019Client}, Genetic \cite{barika2019scheduling}, Greedy \cite{shi2020device}, DNN and SA, in both IID and non-IID settings. 
Afterward, we present the ablation experiments to show the impact of execution time and the data fairness in the cost model and the influence of $\Omega(r)$. 

\subsubsection{Comparison with Single-Job FL}}

In order to demonstrate the effectiveness of our proposed framework, i.e., MJ-FL, over the SJ-FL, we execute each group of jobs sequentially with FedAvg, which is denoted the Random method when adapted to multi-job FL. As shown in Tables \ref{tab:GroupA}, \ref{tab:GroupB}, and \ref{tab:Signal}, \liu{MJ-FL outperforms SJ-FL (up to 1.68 times faster) with Random and the same accuracy. RLDS with MJ-FL outperforms Random with SJ-FL up to 15.38 times faster, and the advantage of BODS can be up to 8.83 times faster. Furthermore, Meta-Greedy can achieve much better performance, i.e., 19 times faster. }

\subsubsection{Comparison within MJ-FL}

In this section, we present the experimental results with the IID setup and the non-IID setup.

{\bf{Evaluation with the IID setting:}} \liu{As shown in Figures \ref{fig:detail-on-IID-GroupA} and \ref{fig:detail-IID-GroupB}, the convergence speed of our proposed methods (i.e., BODS, RLDS and Meta-Greedy) is significantly faster than other methods. In addition, the convergence speed of RLDS has significant advantages in terms of both complex and simple jobs compared to BODS, while Meta-Greedy outperforms all others in terms of convergence speed. Besides, Meta-Greedy can achieve higher accuracy within less time compared to other methods for both complex jobs (ResNet in Figure \ref{fig:detail-IID-GroupB}(a)) and simple jobs (AlexNet in Figure \ref{fig:detail-IID-GroupB}(c)). Tables \ref{tab:GroupA} and \ref{tab:GroupB} reveal that IID data correspond to high convergence accuracy, and for VGG in IID setups, the final accuracy of our proposed methods, i.e., BODS (up to 15.52\% compared to Greedy), RLDS (up to 15.90\% compared to Greedy), and Meta-Greedy (up to 15.33\% compared to Greedy), significantly surpasses that of the other methods. In addition, the training time for a target accuracy of our proposed methods is significantly shorter than baseline methods, including the time for an individual job, i.e., the training time of each job (up to 8.19 times faster for BODS, 12.6 times faster for RLDS, and 12.73 times faster for Meta-Greedy), and the time for the whole training process, i.e., the total time calculated according to Formula \ref{eq:problem}, (up to 4.04 times for BODS, 5.81 times for RLDS and 8.16 times for Meta-Greedy). Slightly different from Non-IID setups results, RLDS performs better than BODS for both complex and simple jobs in terms of the convergence speed, while Meta-Greedy still converges the fastest to the target accuracy. }

\begin{figure*}[htbp]
\centering
\begin{subfigure}{0.3\linewidth}
\includegraphics[width=\linewidth]{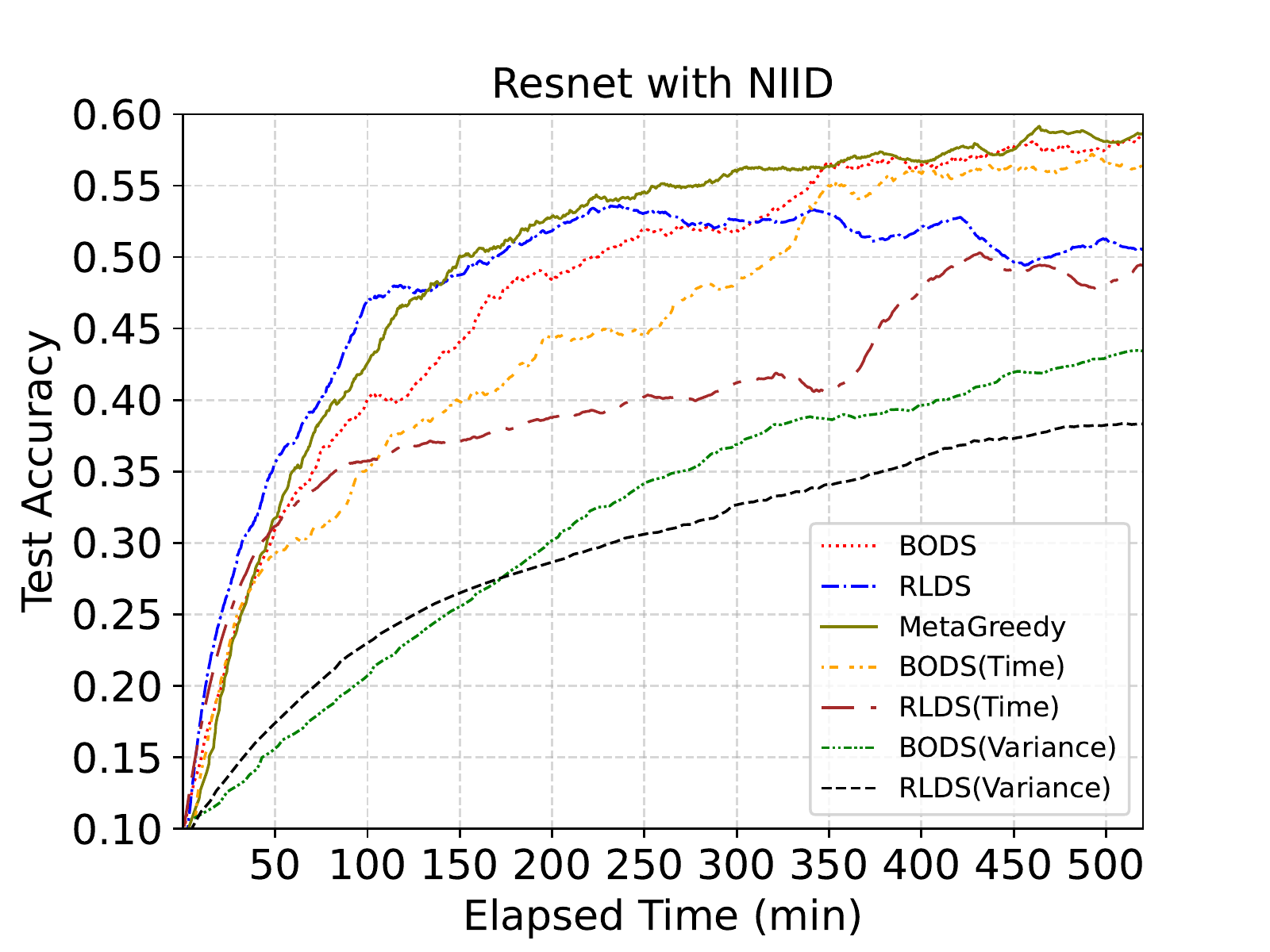}
\vspace{-4mm}
\caption{}
\label{figReNiidAb}
\end{subfigure}
\begin{subfigure}{0.3\linewidth}
\includegraphics[width=\linewidth]{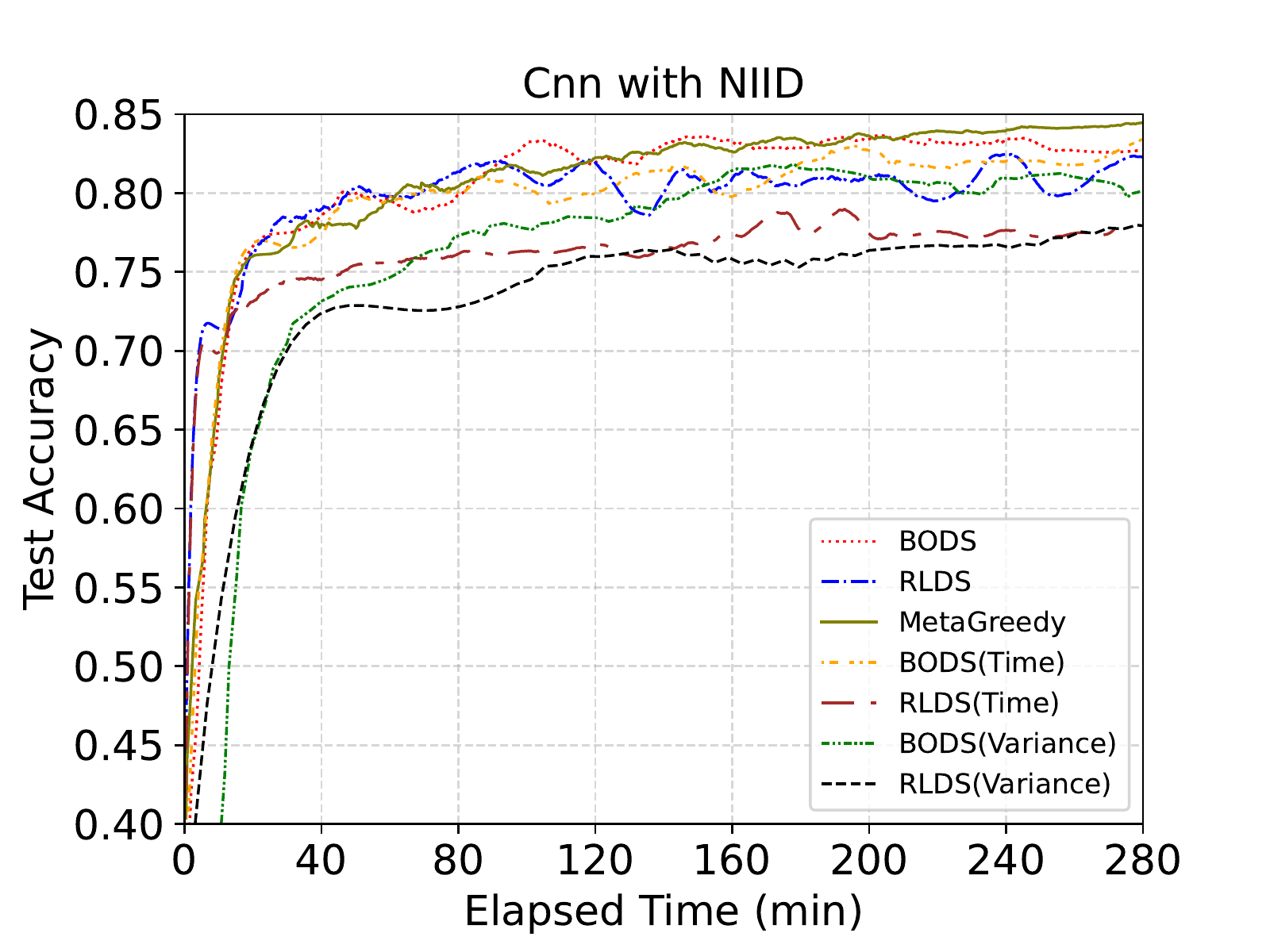}
\vspace{-4mm}
\caption{}
\label{figCBNiidAb}
\end{subfigure}
\begin{subfigure}{0.3\linewidth}
\includegraphics[width=\linewidth]{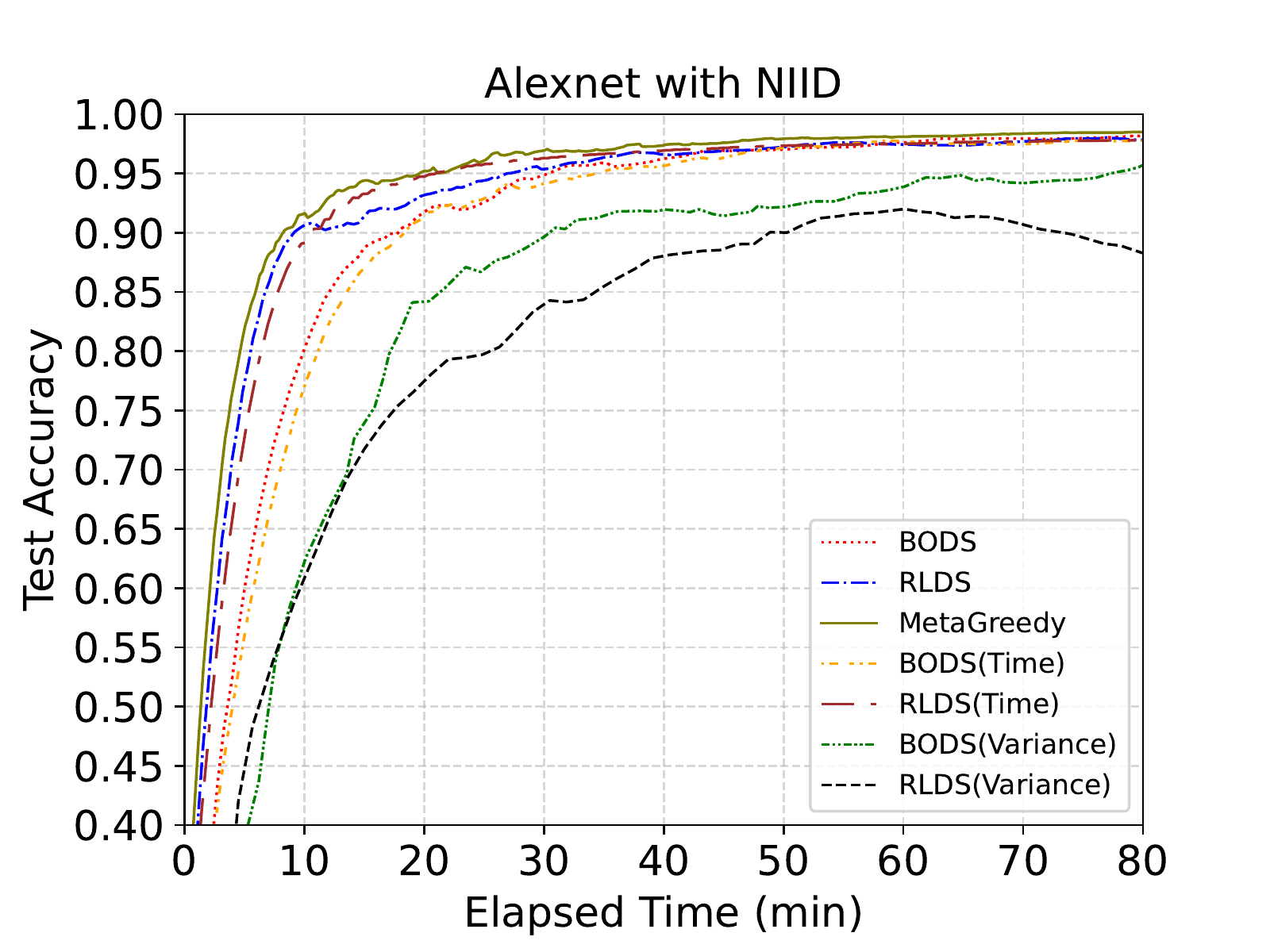}
\vspace{-4mm}
\caption{}
\label{figAlNiidAb}
\end{subfigure}
\begin{subfigure}{0.3\linewidth}
\includegraphics[width=\linewidth]{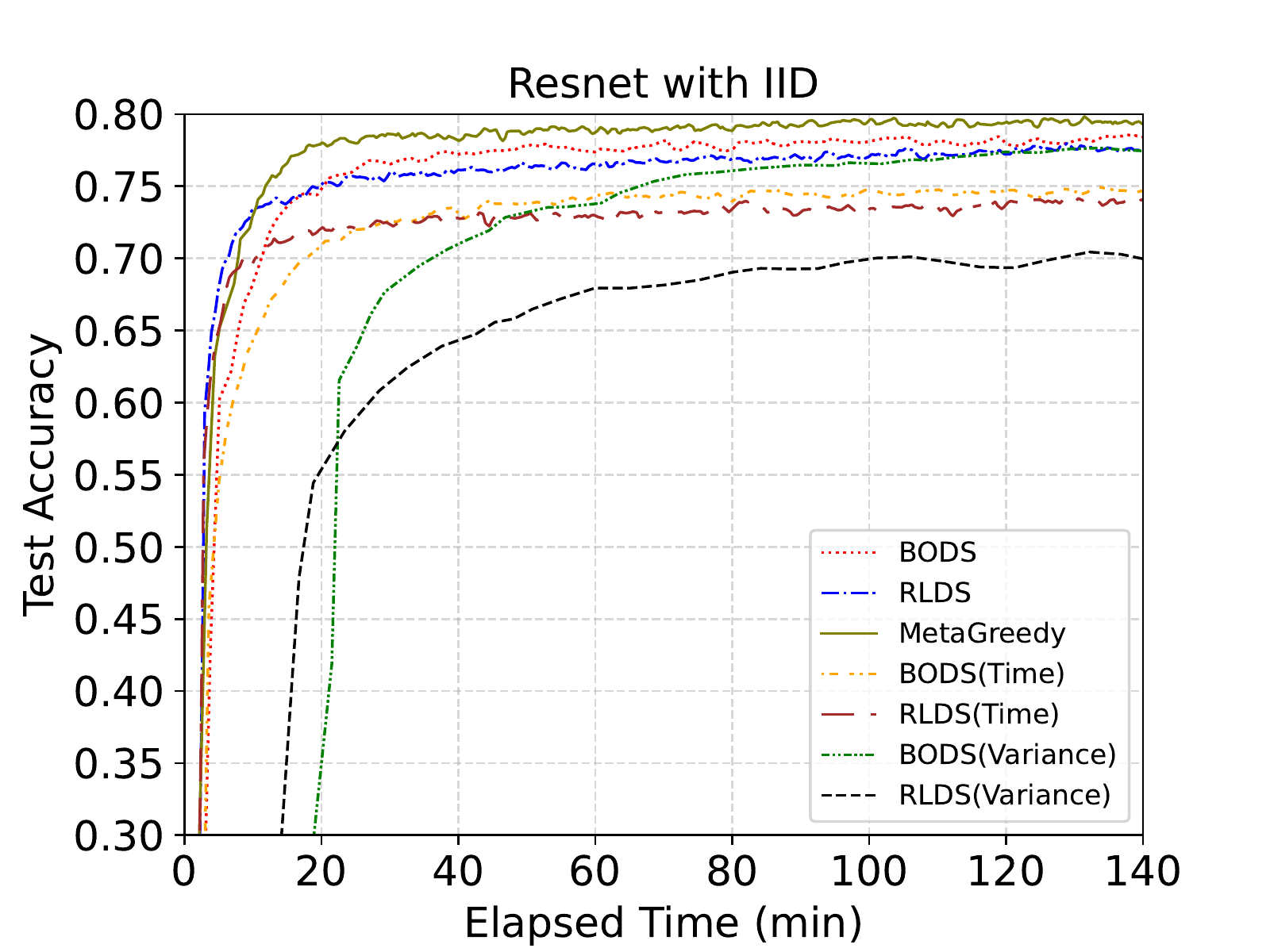}
\vspace{-4mm}
\caption{}
\label{figReiidAb}
\end{subfigure}
\begin{subfigure}{0.3\linewidth}
\includegraphics[width=\linewidth]{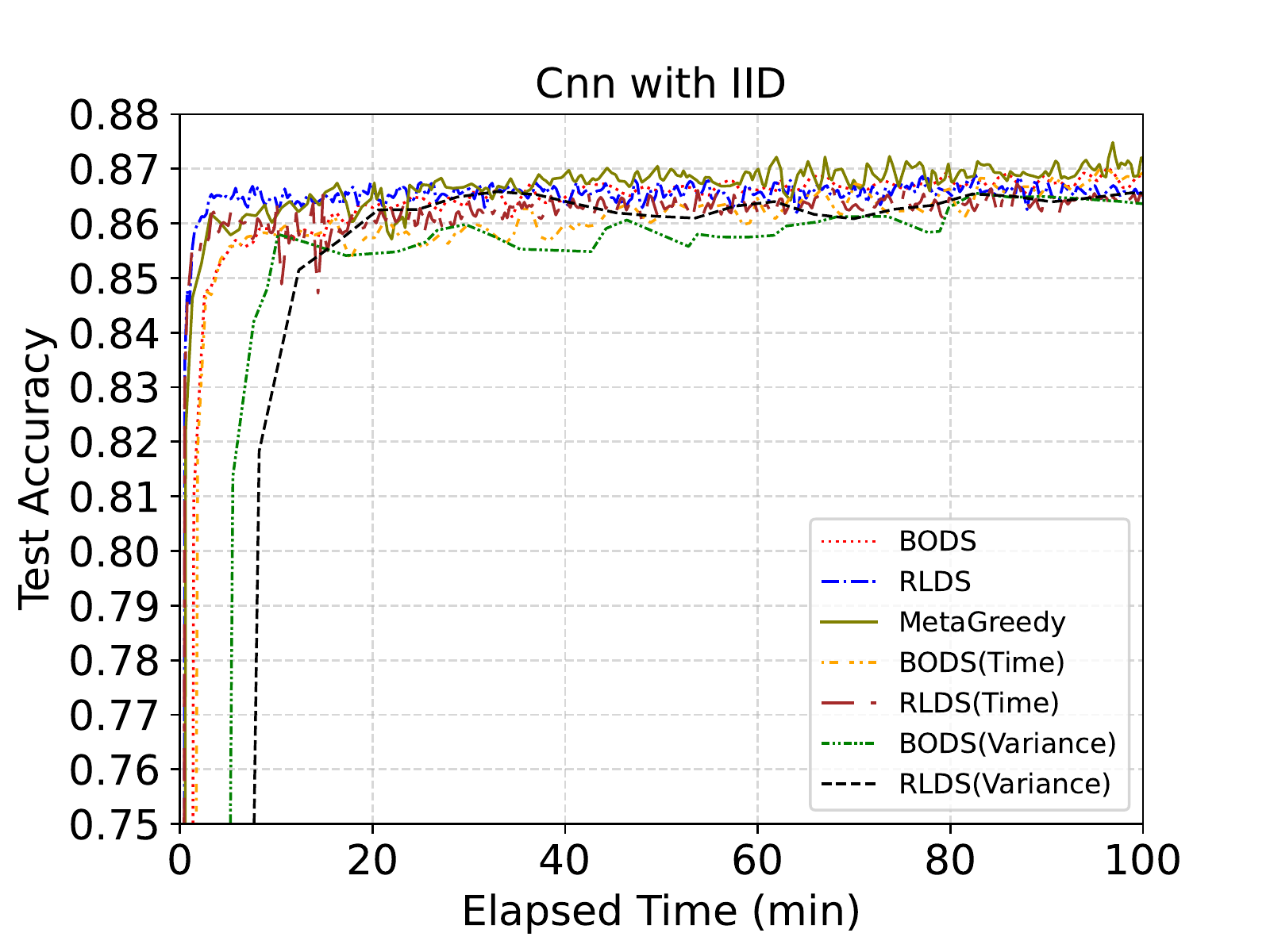}
\vspace{-4mm}
\caption{}
\label{figCBiidAb}
\end{subfigure}
\begin{subfigure}{0.3\linewidth}
\includegraphics[width=\linewidth]{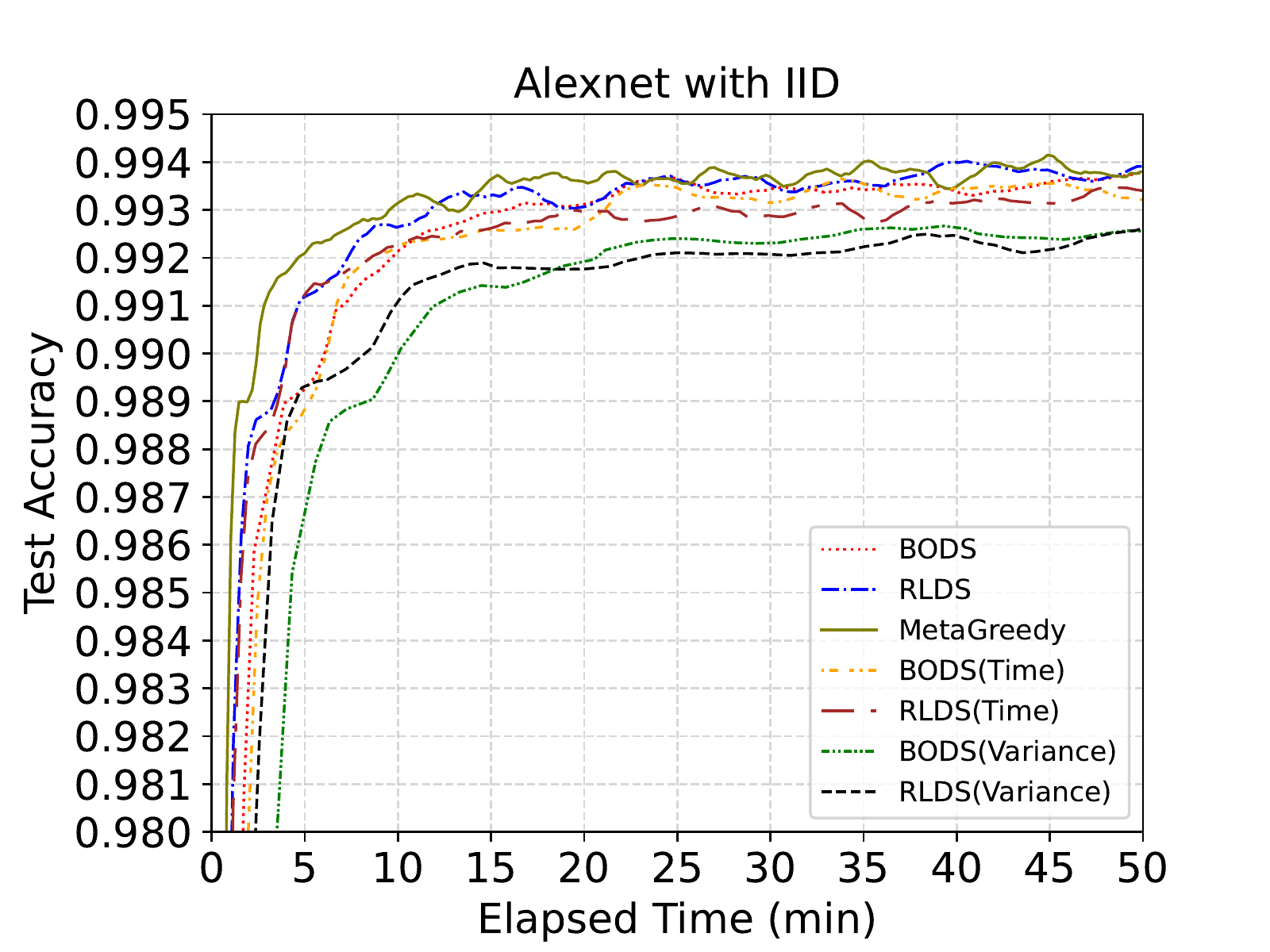}
\vspace{-4mm}
\caption{}
\label{figAliidAb}
\end{subfigure}
\vspace{-2mm}
\caption{The convergence accuracy of different jobs in Group B changes over time with the ablation setting, i.e., ``time'' represents the cost model with execution time and ``variance'' represents the cost model with data fairness.}
\label{fig:groupBAb}
\vspace{-4mm}
\end{figure*}

{\bf{Evaluation with the non-IID setting:}} When the decentralized data is of non-IID, the data fairness defined in Formula \ref{eq:fairness} has a significant impact on the accuracy. As shown in Figures \ref{fig:groupANonIID} and \ref{fig:groupBNonIID}, the convergence speed of our proposed methods, i.e., RLDS, BODS and Meta-Greedy, is significantly faster than other methods. RLDS shows a significant advantage for complex jobs (VGG in Figure \ref{figVgNiid}), while BODS shows advantage for relatively simple jobs in Groups A and B (please see details in Figures \ref{fig:groupANonIID} and \ref{fig:groupBNonIID}). \liuR{This is reasonable as we have much more parameters to adjust for complex jobs with RLDS, e.g., the learning rate, the $\epsilon$ decay rate, the structure of the LSTM etc, which can fit into a complex jobs. However, it may be complicated to fine-tune the parameters for a simple job with RLDS. In contrary, the execution of simple jobs can be directly well addressed by the Bayesian optimization.} \liu{ Meta-Greedy can lead to a good performance for both simple and complex jobs. As shown in Tables \ref{tab:GroupA} and \ref{tab:GroupB}, the final accuracy of RLDS, BODS and Meta-Greedy significantly outperforms other methods (up to 44.6\% for BODS, 39.4\% for RLDS and 46.4\% for Meta-Greedy), as well. Given a target accuracy, our proposed methods can achieve the accuracy within a shorter time, compared with baseline methods, in terms of the time for a single job, i.e., the training time of each job (up to 5.04 times shorter for BODS, 5.11 times shorter for RLDS and 7.53 times shorter for Meta-Greedy), and the time for the whole training process, i.e., the total time calculated based on Formula \ref{eq:problem} (up to 4.15 times for BODS, 4.67 times for RLDS and 7.16 times for Meta-Greedy), for Groups A and B. In addition, among our proposed methods, Meta-Greedy always has the shortest time to achieve the target accuracy. 
We have similar observations with IID, while the advantage of Meta-Greedy is much more significant (up to 19.0 times shorter in terms of the time for a single job) than that of non-IID as shown in Tables \ref{tab:GroupA} and \ref{tab:GroupB}. Besides, we can find that our proposed method (i.e., Meta-Greedy) leads to a better performance in the case of non-IID data due to dynamically enhancing the impact of data fairness to decrease the imbalance of devices.}

{\bf{Comparison in other cases:}} As shown in Figures \ref{fig:detail-IID-GroupB} and \ref{fig:groupBNonIID}, the convergence speed corresponding to RLDS, BODS and Meta-Greedy is much higher than baseline methods with different target accuracy. \liu{The advantage of RLDS and BODS is up to 8.77 times faster for Target 1 (0.845), 13.96 times faster for Target 2 (0.856), and 27.04 times faster for Target 3 (0.865), compared with the baselines. In addition, Meta-Greedy outperforms the other six methods (four baselines and two scheduling methods, i.e., RLDS and BODS) up to 4.42 times faster for Target 1 (0.845), 6.65 times faster for Target 2 (0.856), and 5.61 times faster for Target 3 (0.865). Furthermore,} According to Figure \ref{fig:DNN and SA}, SA \cite{van1987simulated} corresponds to much worse performance (up to 91.4\% slower and 3.5\% lower accuracy) compare with our methods. We carry out experiments to compare out methods with DNNs \cite{zang2019hybrid}, the performance of which is significantly worse (up to 90.5\% slower and 26.3\% lower accuracy) than our methods. In addition, we test other combinations of the two costs, which correspond to worse performance (up to 37.1\% slower and 3.5\% lower accuracy for the sum of squared costs, and 64.4\% slower and 3.3\% lower accuracy for multiplication) compared to the linear one (Formula \ref{eq:totalCost}). 

\begin{figure*}[htbp]
\centering
\begin{subfigure}{0.3\linewidth}
\includegraphics[width=\linewidth]{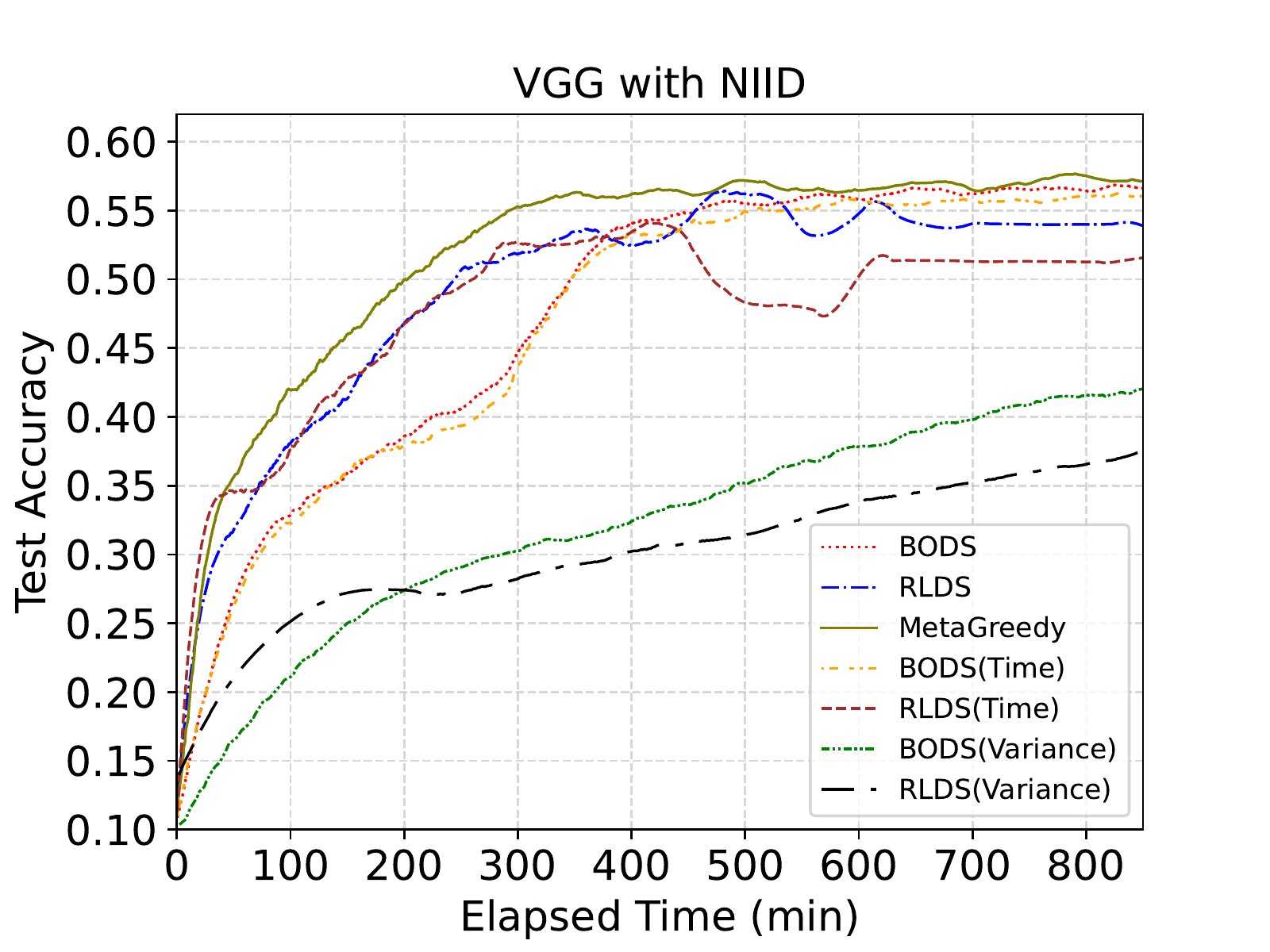}
\vspace{-4mm}
\caption{}
\label{figVgNiidAb}
\end{subfigure}
\begin{subfigure}{0.3\linewidth}
\includegraphics[width=\linewidth]{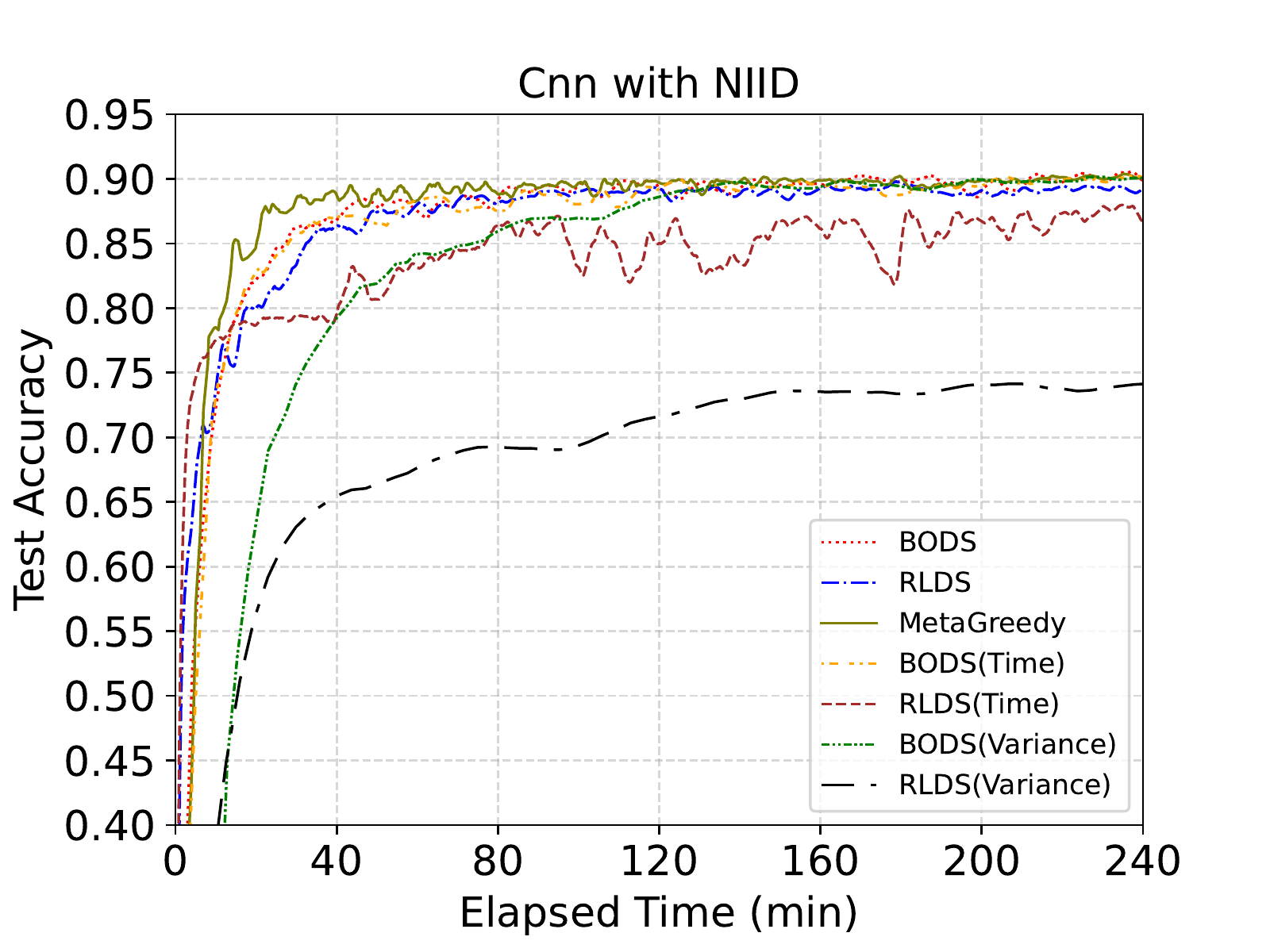}
\vspace{-4mm}
\caption{}
\label{figCnNiidAb}
\end{subfigure}
\begin{subfigure}{0.3\linewidth}
\includegraphics[width=\linewidth]{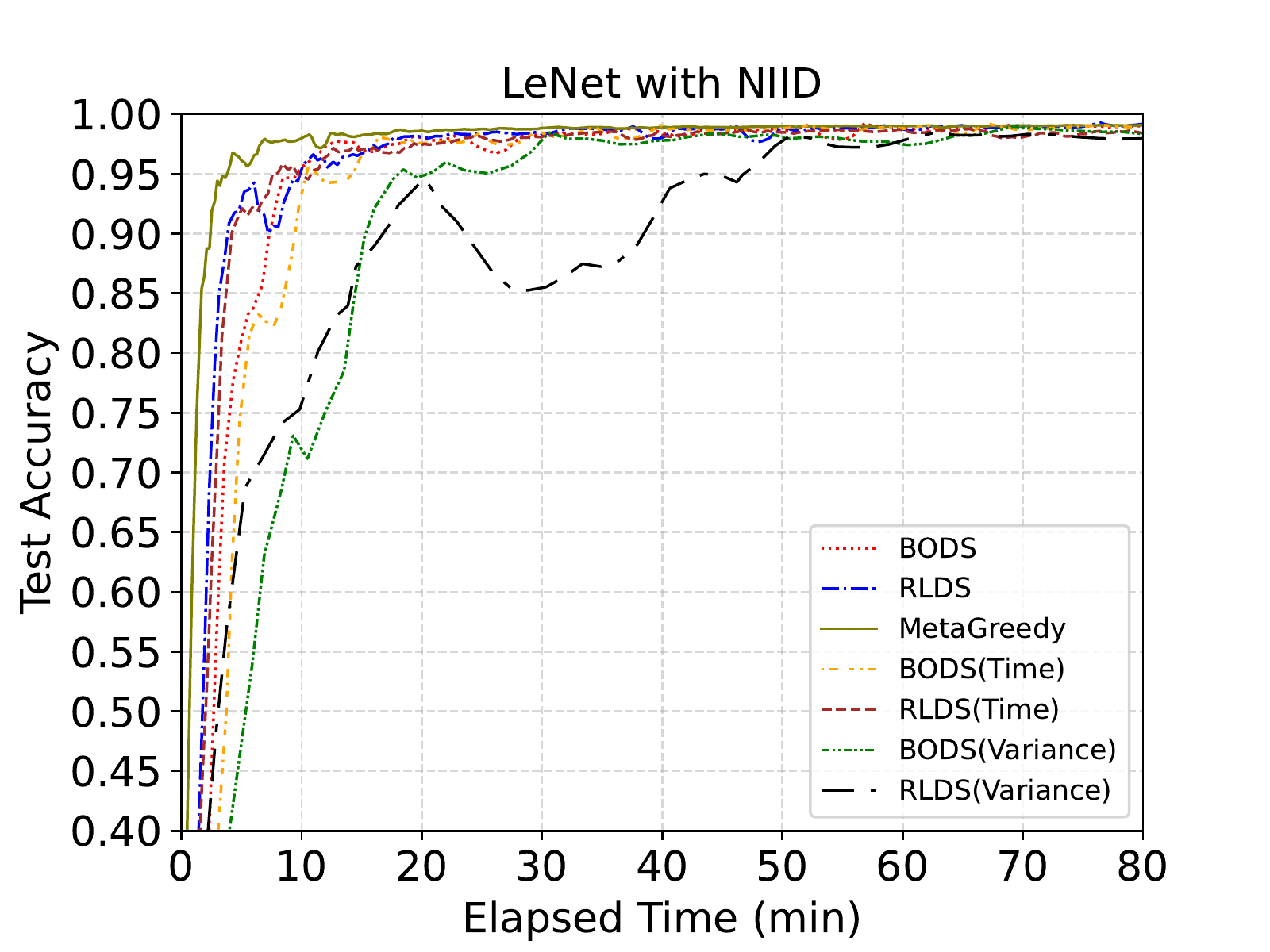}
\vspace{-4mm}
\caption{}
\label{figLeNiidAb}
\end{subfigure}
\begin{subfigure}{0.3\linewidth}
\includegraphics[width=\linewidth]{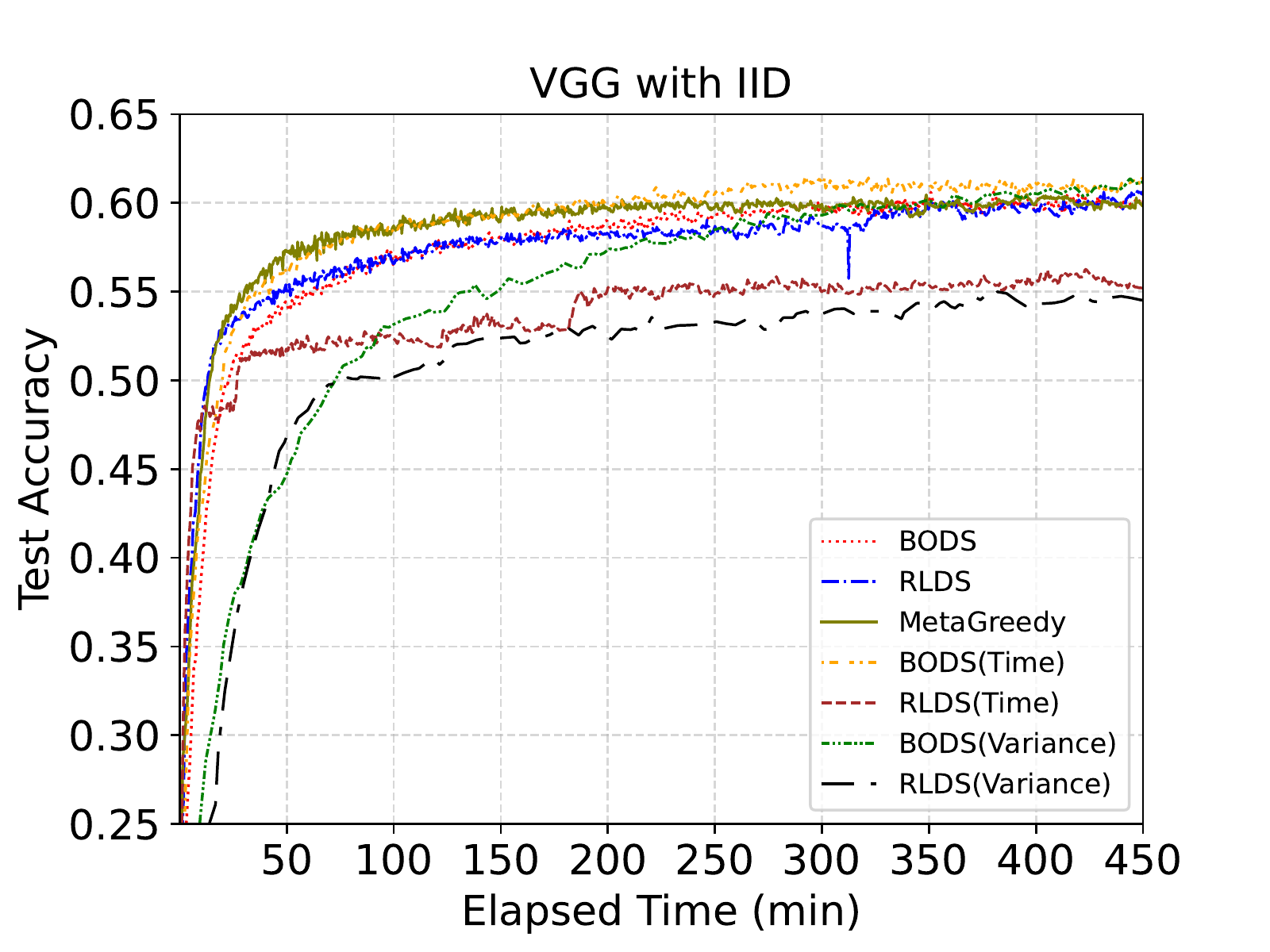}
\vspace{-4mm}
\caption{}
\label{figVgiidAb}
\end{subfigure}
\begin{subfigure}{0.3\linewidth}
\includegraphics[width=\linewidth]{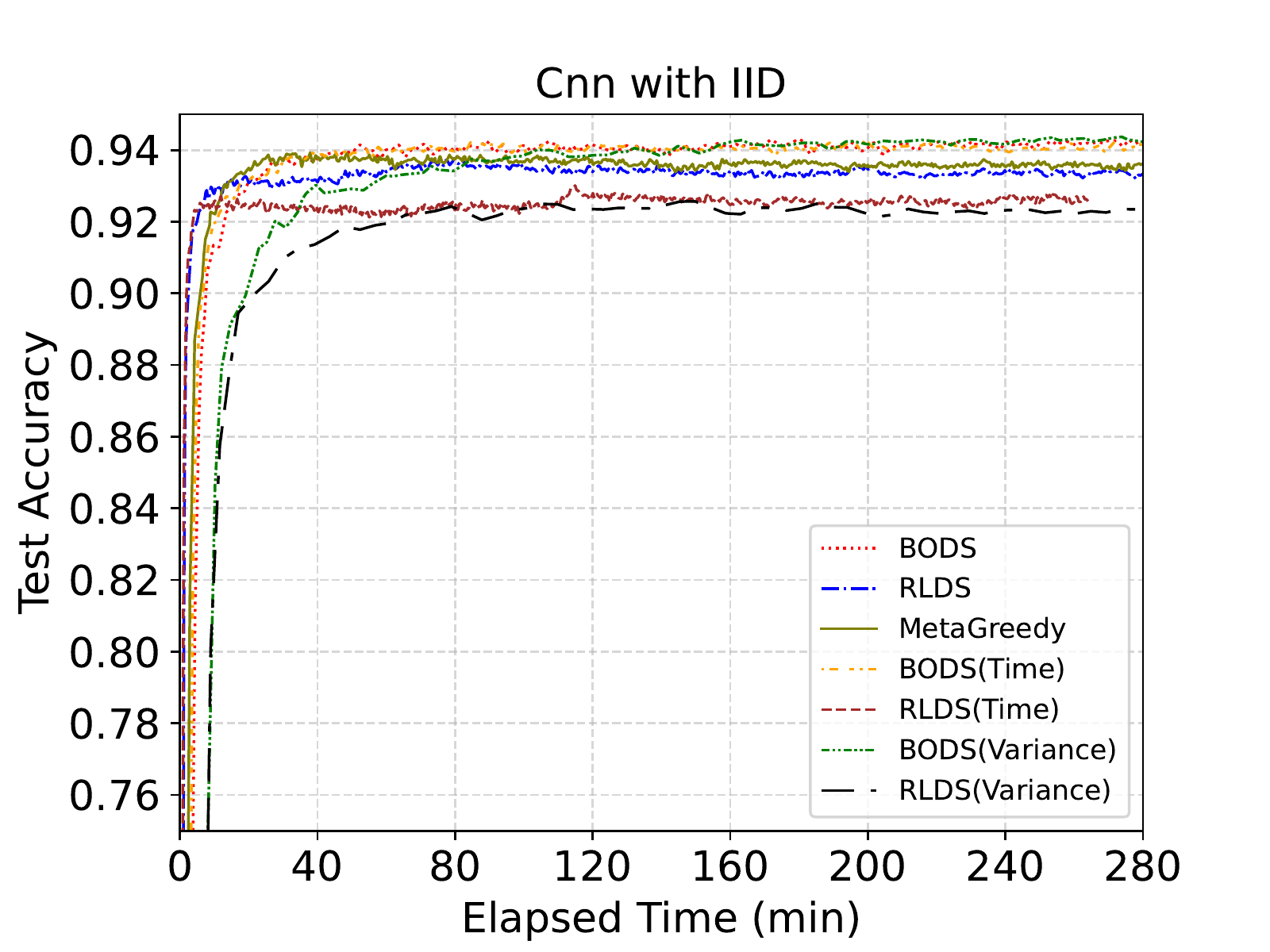}
\vspace{-4mm}
\caption{}
\label{figCniidAb}
\end{subfigure}
\begin{subfigure}{0.3\linewidth}
\includegraphics[width=\linewidth]{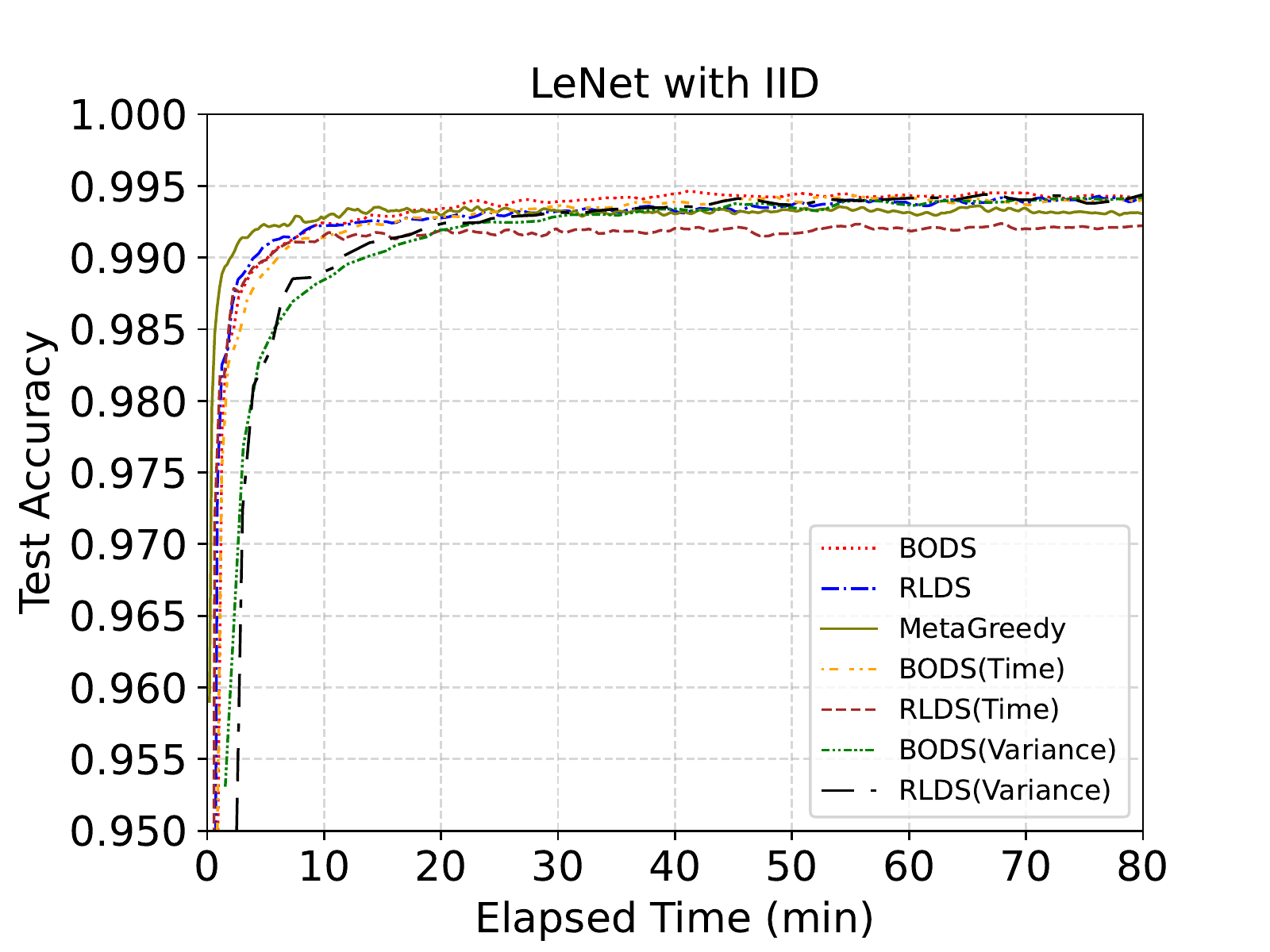}
\vspace{-4mm}
\caption{}
\label{figLeiidAb}
\end{subfigure}
\vspace{4mm}
\caption{The convergence accuracy of different jobs in Group A changes over time with the ablation setting, i.e., ``time'' represents the cost model with execution time and ``variance'' represents the cost model with data fairness.}
\label{fig:groupAAb}
\end{figure*}

\begin{figure*}[htbp]
\centering
\begin{subfigure}{0.3\linewidth}
\includegraphics[width=\linewidth]{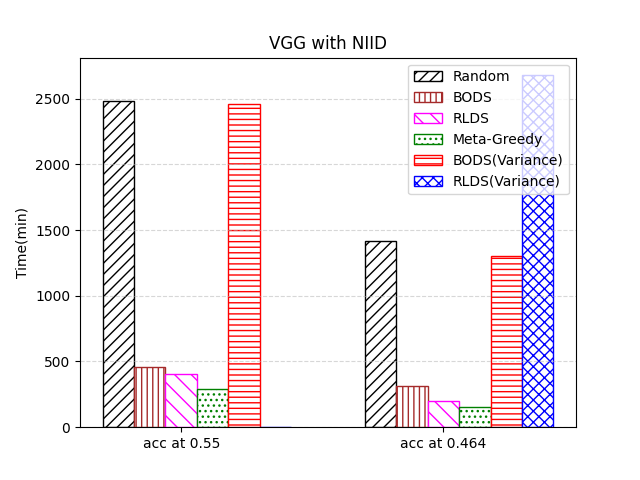}
\vspace{-4mm}
\caption{}
\label{figVgNiidAb-bar}
\end{subfigure}
\begin{subfigure}{0.3\linewidth}
\includegraphics[width=\linewidth]{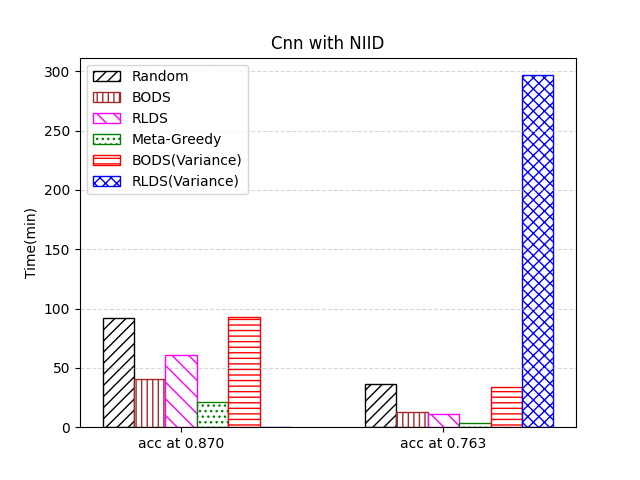}
\vspace{-4mm}
\caption{}
\label{figCnNiidAb-bar}
\end{subfigure}
\begin{subfigure}{0.3\linewidth}
\includegraphics[width=\linewidth]{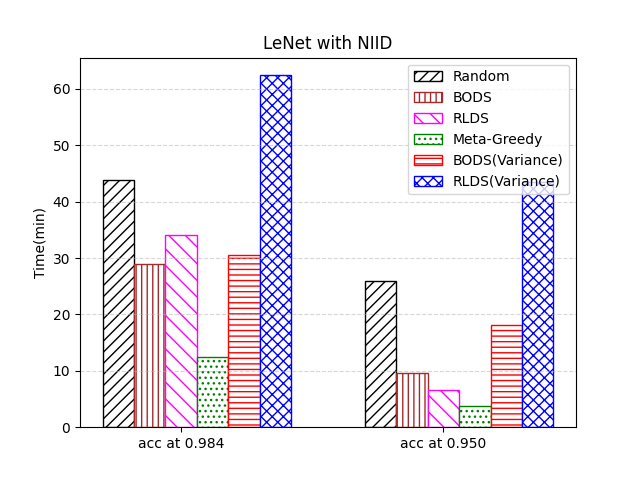}
\vspace{-4mm}
\caption{}
\label{figLeNiidAb-bar}
\end{subfigure}
\begin{subfigure}{0.3\linewidth}
\includegraphics[width=\linewidth]{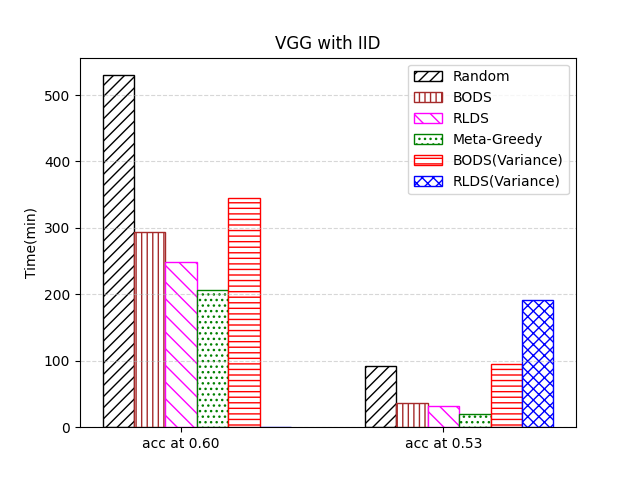}
\vspace{-4mm}
\caption{}
\label{figVgiidAb-bar}
\end{subfigure}
\begin{subfigure}{0.3\linewidth}
\includegraphics[width=\linewidth]{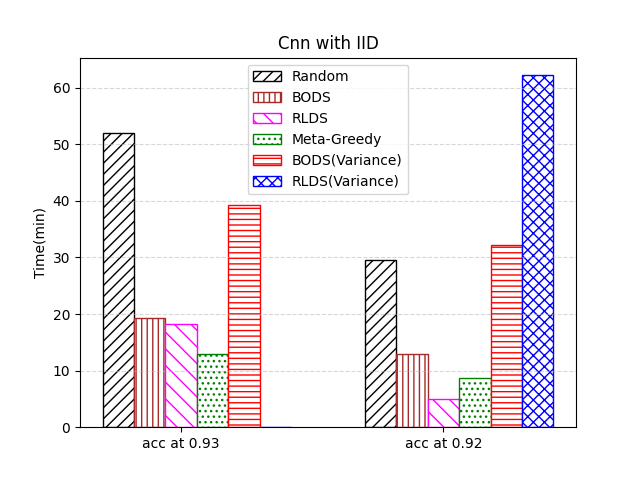}
\vspace{-4mm}
\caption{}
\label{figCniidAb-bar}
\end{subfigure}
\begin{subfigure}{0.3\linewidth}
\includegraphics[width=\linewidth]{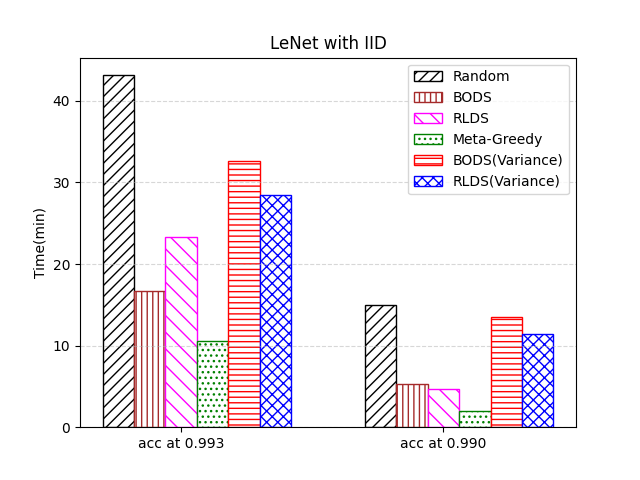}
\vspace{-4mm}
\caption{}
\label{figLeiidAb-bar}
\end{subfigure}
\vspace{-2mm}
\caption{The time required for each job of Group A to achieve the target convergence accuracy on the non-IID and IID distribution with the ablation setting, i.e., ``time'' represents the cost model with execution time and ``variance'' represents the cost model with data fairness.}
\label{fig:groupAAb-bar}
\vspace{-6mm}
\end{figure*}

\begin{figure*}[htbp]
\centering
\begin{subfigure}{0.3\linewidth}
\includegraphics[width=\linewidth]{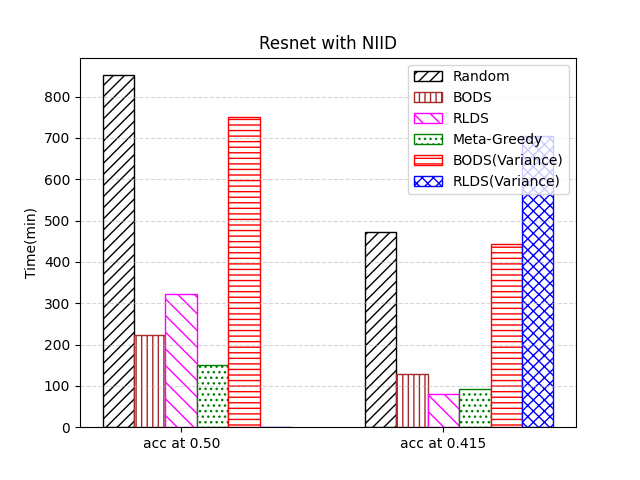}
\vspace{-4mm}
\caption{}
\label{figReNiidAb-bar}
\end{subfigure}
\begin{subfigure}{0.3\linewidth}
\includegraphics[width=\linewidth]{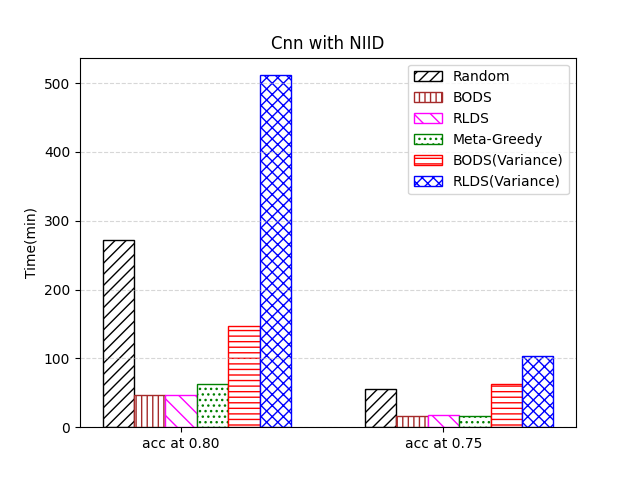}
\vspace{-4mm}
\caption{}
\label{figCBNiidAb-bar}
\end{subfigure}
\begin{subfigure}{0.3\linewidth}
\includegraphics[width=\linewidth]{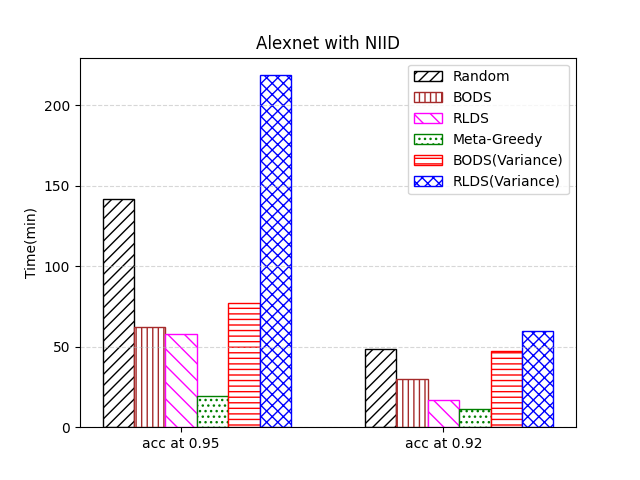}
\vspace{-4mm}
\caption{}
\label{figAlNiidAb-bar}
\end{subfigure}
\begin{subfigure}{0.3\linewidth}
\includegraphics[width=\linewidth]{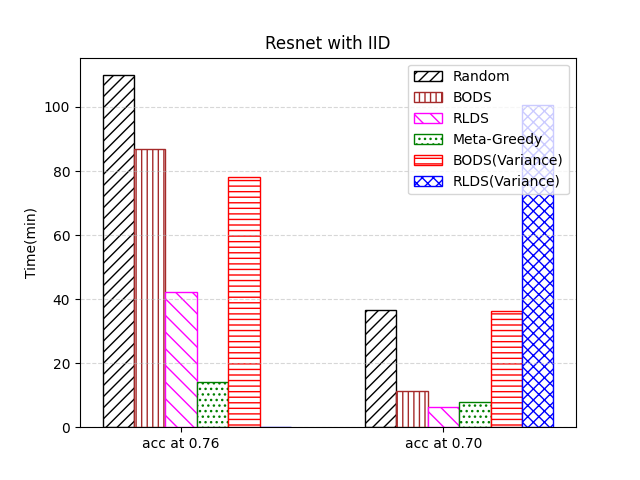}
\vspace{-4mm}
\caption{}
\label{figReiidAb-bar}
\end{subfigure}
\begin{subfigure}{0.3\linewidth}
\includegraphics[width=\linewidth]{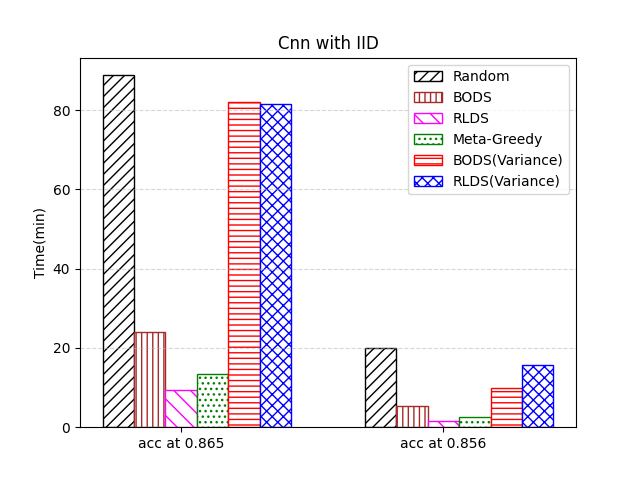}
\vspace{-4mm}
\caption{}
\label{figCBiidAb-bar}
\end{subfigure}
\begin{subfigure}{0.3\linewidth}
\includegraphics[width=\linewidth]{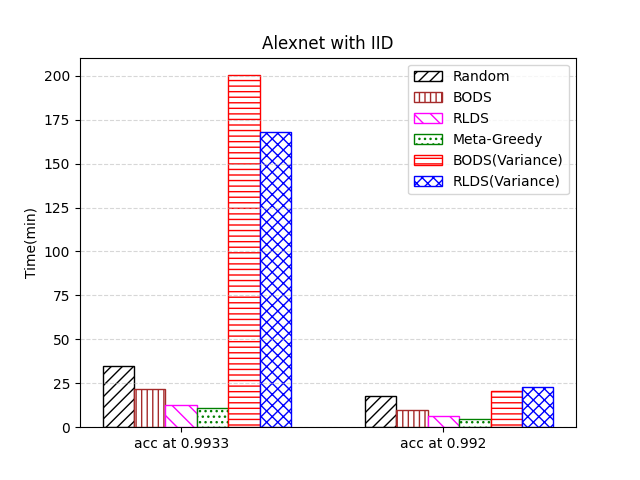}
\vspace{-4mm}
\caption{}
\label{figAliidAb-bar}
\end{subfigure}
\vspace{-2mm}
\caption{The time required for each job of Group B to achieve the target convergence accuracy on the non-IID and IID distribution with the ablation setting, i.e., ``time'' represents the cost model with execution time and ``variance'' represents the cost model with data fairness.}
\vspace{-6mm}
\label{fig:groupBAb-bar}
\end{figure*}

As RLDS can learn more information through a complex neural network, RLDS outperforms BODS for complex jobs (\liuR{0.008 and 0.029 in terms of accuracy with VGG19 and ResNet18, and 46.7\% and 34.8\% faster for the target accuracy of 0.7 with VGG19 and 0.5 with ResNet18; see details in Appendix}). Due to the emphasis on the combination of data fairness and device capabilities, i.e.,computation and communication capabilities, BODS can lead to high convergence accuracy and fast convergence speed for simple jobs (\liuR{0.018 in terms of accuracy and 38\% faster for the target accuracy of 0.97 with CNN; see details in Appendix}). \liu{Meta-Greedy reconstructs the cost model based on six scheduling schemes and dynamically adjusts the parameters so that the impact of data fairness increases with the number of rounds, which leads to high convergence accuracy and fast convergence speed in both complex and simple jobs. BODS, RLDS and Meta-Greedy significantly outperform the baseline methods, while there are also differences among the four methods.} The Greedy method prefers devices with high capacity, which leads to a significant decline in terms of the final convergence accuracy. The Genetic method can exploit randomness to achieve data fairness while generating scheduling plans, and the convergence performance is better than the Greedy method. The FedCS method optimizes the scheduling plan by randomly selecting devices, which improves the fairness of the device to a certain extent, and convergences faster than the Random method. \liuR{We carry out an experiment with Group A with both IID and non-IID distribution, and find that Meta-Greedy with 6 methods significantly outperforms that with 2 methods in terms of accuracy (up to 0.091) and training speed (up to 65\%) (see details in Appendix). In addition, after analyzing the log of the training process (see details in Appendix), we find Greedy and Genetic are extensively exploited, RLDS is selected at the beginning of the training process, BODS is chosen at the end of the training process, FedCS participates with less frequency, and Random is seldomly utilized. As Meta-Greedy can intelligently select a proper scheduling plan based on available methods, it results in a more efficient training process.}

\subsubsection{Ablation Study}

In this section, we first present the ablation study to show the impact of the execution time and the data fairness. Then, we analyze the influence of $\Omega(r)$. 

{\bf{Impact of data fairness:}}
\liu{
As shown in Figures \ref{fig:groupAAb} and \ref{fig:groupBAb}, we conduct ablation experiments and find an evident decline in convergence accuracy when the cost model is only composed of execution time cost. From Figures \ref{fig:groupAAb} and \ref{fig:groupBAb}, we can see that the accuracy of the cost model composed of execution time (i.e., BODS (Time) and RLDS (Time)) is worse than that of both data fairness and execution time (i.e., BODS and RLDS) in most cases (up to 35.83\% lower for RLDS (Time)). The abnormal case, when the convergence accuracy of BODS (Time) is slightly higher (about 1\%) than the convergence accuracy of BODS (please see details in Figure \ref{fig:groupAAb} (d)) should be caused by randomness. Besides, from the Figures \ref{fig:groupAAb} and \ref{fig:groupBAb}, we can find that the impact of execution time on complex jobs (VGG in Figure \ref{fig:groupAAb} and Resnet in Figure \ref{fig:groupBAb}) is significantly greater than that on simple jobs (LeNet in Figure \ref{fig:groupAAb} and AlexNet in Figure \ref{fig:groupBAb}).}

{\bf{Impact of execution time:}} \liu{
As shown in Figures \ref{fig:groupAAb} and \ref{fig:groupBAb}, the data fairness improves both the convergence speed (up to 9.35 times faster) and the accuracy (up to 15.3\%). From Figures \ref{fig:groupAAb-bar} and \ref{fig:groupBAb-bar}, we can find that the convergence speed significantly decreases (up to 25.46 times slower for RLDS (Variance) compared with RLDS) when given a target accuracy with the cost model composed of only data fairness. In most cases, considering data fairness reduces the convergence speed and essentially does not affect the convergence accuracy. There are few cases that the job fails to achieve the convergence accuracy. In addition, from the results shown in Figures \ref{fig:groupAAb} and \ref{fig:groupBAb}, we find that the convergence accuracy in complex jobs (VGG in Figure \ref{fig:groupAAb-bar} and ResNet in Figure \ref{fig:groupBAb-bar}) are more likely to decline when the cost model is only composed of data fairness, compared with that with simple jobs. Besides, the performance corresponding to the cost model with only data fairness is close to that of the 'Random' method in most setups. }

{\bf{Influence of $\Omega(r)$:}}\liu{
As shown in Figures 
of Appendix, Meta-Greedy with $\Omega(r)=\sqrt{r}$ outperforms other methods, i.e., $\Omega(r)=r$ and $\Omega(r)=\log r$, in terms of both convergence accuracy and convergence speed. 
Although the convergence accuracy of Meta-Greedy with $\Omega(r)=r$ is slightly higher than that with $\Omega(r) = \sqrt{r}$ in a few experiments, the convergence accuracy of Meta-Greedy with $\Omega(r) = \sqrt{r}$ in complex jobs is significantly lower than that with $\Omega(r) = r$. In addition, Meta-Greedy with $\Omega(r) = \sqrt{r}$ takes less time to reach the target accuracy in previous rounds compared with Meta-Greedy with $\Omega(r)=r$. This indicates that when round $r$ is linearly related to $\Omega(r)$, the data fairness in later rounds influences the cost model too much and may lead to a lower convergence accuracy. Therefore, we avoid using $\Omega(r) = r$ due to the dramatic changes of the magnitude in the linear relationship. 
Meta-Greedy with $\Omega(r) = \sqrt{r}$ takes less time to reach the target accuracy in previous rounds compared with Meta-Greedy with $\Omega(r)=\log r$. In the meanwhile, the convergence accuracy of Meta-Greedy with $\Omega(r)=\log r$ is lower than that of Meta-Greedy with $\Omega(r) = \sqrt{r}$, which implies that the relatively gentle variation of $\Omega(r)$ with round $r$ leads to lower convergence accuracy and slower convergence speed. By contrast, Meta-Greedy with $\Omega(r) = \sqrt{r}$ performs best among them in terms of both convergence accuracy and convergence speed. Thus, we adopt Meta-Greedy with $\Omega(r) = \sqrt{r}$, which dynamically changes the influence of data fairness to obtain relatively optimal solution in both complex and simple jobs.
}

\section{Conclusion}
\label{sec:conclusion}

In this work, we proposed a new Multi-Job Federated Learning framework, i.e., MJ-FL. The framework is composed of a system model and three device scheduling methods. The system model is composed of a process for the parallel execution of multiple jobs and a cost model based on the capability of devices and data fairness. \liu{We proposed three device scheduling methods, i.e., RLDS for complex jobs and BODS for simple jobs, while Meta-Greedy for both complex and simple jobs, to efficiently schedule proper devices for each job based on the cost model. We carried out extensive experimentation with six real-life models and four datasets with IID and non-IID distribution. The experimental results show that MJ-FL outperforms the single-job FL, and that our proposed scheduling methods, i.e., BODS and RLDS, significantly outperform baseline methods (up to 44.6\% in terms of accuracy, 12.6 times faster for a single job and 5.81 times faster for the total time). In addition, Meta-Greedy, the intelligent scheduling approach based on multiple scheduling methods (including the two proposed methods, i.e., RODS and BODS) significantly outperforms other methods (up to 46.4\% in terms of accuracy, 12.73 times faster for a single job and 8.16 times faster for the total time). 
}




\ifCLASSOPTIONcaptionsoff
  \newpage
\fi



\bibliographystyle{IEEEtran}
\bibliography{references}
%




%

\section*{Appendix}

\subsection*{Proof of Theorem 1}

For simplicity, we take $\nabla f^m_k(w^{m,k}_{r,h})$ instead of $\nabla f^m_k(w^{m,k}_{r,h}; \zeta^{m,r}_{k,h})$ in the proof. 
By the smoothness of $F^m$, we have
\begin{align}
    &\mathbb{E}[F^m(\bar{w}^m_{r,h+1})] \notag \\
    \leq &~\mathbb{E}[F^m(\bar{w}^m_{r,h})] \notag + \underbrace{\mathbb{E}[\langle \nabla F^m(\bar{w}^m_{r,h}), \bar{w}^m_{r,h+1}-\bar{w}^m_{r,h}\rangle]}_{A} \notag \\
    &+ \frac{L}{2} \underbrace{\mathbb{E}[\parallel \bar{w}^m_{r,h+1}-\bar{w}^m_{r,h}\parallel^2]}_{B}
\end{align}
Where for $B$ we have, $\eta^m_{r,h}$ is learning rate.
\begin{align}\label{eq:B}
    &\mathbb{E}[\parallel \bar{w}^m_{r,h+1}-\bar{w}^m_{r,h}\parallel^2] \notag \\
    = &~\mathbb{E}\parallel \eta^m_{r,h} \sum_{k \in V^r_m}p^m_{k,r} g^{m,k}_{r,h} \parallel^2 \notag \\
    = &~{\eta^m_{r,h}}^2 \mathbb{E}\parallel\sum_{k \in V^r_m}p^m_{k,r} g^{m,k}_{r,h}\parallel^2 \notag\\
    = &~{\eta^m_{r,h}}^2 \mathbb{E}\parallel \sum_{k \in V^r_m}p^m_{k,r} (g^{m,k}_{r,h} - \nabla F^m_k(w^{m,k}_{r,h})) \parallel^2 \notag \\ &+  {\eta^m_{r,h}}^2 \mathbb{E}\parallel \sum_{k \in V^r_m}p^m_{k,r} \nabla F^m_k(w^{m,k}_{r,h})) \parallel^2 \notag\\
    \leq &~{\eta^m_{r,h}}^2 \sum_{k \in V^r_m} p^m_{k,r} \mathbb{E} \parallel g^{m,k}_{r,h} - \nabla F^m_k(w^{m,k}_{r,h})\parallel^2 \notag\\ 
    &+ {\eta^m_{r,h}}^2 \mathbb{E}\parallel \sum_{k \in V^r_m}p^m_{k,r} \nabla F^m_k(w^{m,k}_{r,h})) \parallel^2 \notag\\
    \leq &~{\eta^m_{r,h}}^2\sum_{k \in V^r_m}p^m_{k,r} \sigma^2+{\eta^m_{r,h}}^2 \mathbb{E}\parallel \sum_{k \in V^r_m}p^m_{k,r} \nabla F^m_k(w^{m,k}_{r,h})) \parallel^2
\end{align}
where the first inequality results from Jensen's inequality. The last inequality results from Assumption 3.
Next, for $A$ we have
\begin{equation}
\begin{split}
    &\mathbb{E}[\langle \nabla F^m(\bar{w}^m_{r,h}), \bar{w}^m_{r,h+1}-\bar{w}^m_{r,h}\rangle] \notag \\
    = &~\mathbb{E}[\langle \nabla F^m(\bar{w}^m_{r,h}), -\eta^m_{r,h}\sum_{k \in V^r_m}p^m_{k,r} g^{m,k}_{r,h} \rangle] \notag \\
    = &~-\eta^m_{r,h}\mathbb{E}[\langle \nabla F^m(\bar{w}^m_{r,h}), \sum_{k \in V^r_m}p^m_{k,r} \nabla F^m_k(w^{m,k}_{r,h}) \rangle] \notag \\
    = &~-\frac{\eta^m_{r,h}}{2} \mathbb{E}\bigl[\parallel \nabla F^m(\bar{w}^m_{r,h}) \parallel^2 + \parallel \sum_{k \in V^r_m}p^m_{k,r} \nabla F^m_k(w^{m,k}_{r,h}) \parallel^2 \notag \\
    &- \parallel \nabla F^m(\bar{w}^m_{r,h}) -  \sum_{k \in V^r_m}p^m_{k,r} \nabla F^m_k(w^{m,k}_{r,h})\parallel^2\bigr] \notag \\
    = &~-\frac{\eta^m_{r,h}}{2} \mathbb{E}\parallel \nabla F^m(\bar{w}^m_{r,h}) \parallel^2 \notag \\
    & -\frac{\eta^m_{r,h}}{2}\mathbb{E}\parallel \sum_{k \in V^r_m} p^m_{k,r}\nabla F^m_k(w^{m,k}_{r,h}) \parallel^2 \notag \\
    &+ \frac{\eta^m_{r,h}}{2}\parallel \nabla F^m(\bar{w}^m_{r,h}) -  \sum_{k \in V^r_m}p^m_{k,r} \nabla F^m_k(w^{m,k}_{r,h})\parallel^2 
\end{split}
\end{equation} 
where the fourth equality results from the basic identity $\langle a, b \rangle = \frac{1}{2}\bigl( \parallel a \parallel^2 + \parallel b \parallel^2 - \parallel a-b\parallel^2 \bigr)$.\\

Then, combine $A$ and $B$ together 
\begin{align}\label{eq:28}
    &\mathbb{E}[F^m(\bar{w}^m_{r,h+1})] \notag \\
    \leq &~\mathbb{E}[F^m(\bar{w}^m_{r,h})] -\frac{\eta^m_{r,h}}{2} \mathbb{E}\parallel \nabla F^m(\bar{w}^m_{r,h}) \parallel^2 \notag \\
    &-\frac{\eta^m_{r,h}}{2}\mathbb{E}\parallel \sum_{k \in V^r_m} p^m_{k,r}\nabla F^m_k(w^{m,k}_{r,h}) \parallel^2 + \frac{L}{2}{\eta^m}^2 \sum_{k \in V^r_m}p^k_m \sigma^2 \notag \\
    &+ \frac{\eta^m_{r,h}}{2}\parallel \nabla F^m(\bar{w}^m_{r,h}) -  \sum_{k \in V^r_m}p^m_{k,r} \nabla F^m_k(w^{m,k}_{r,h})\parallel^2 \notag \\
    & +\frac{{\eta^m}^2 L}{2} \mathbb{E}\parallel \sum_{k \in V^r_m}p^m_{k,r} \nabla F^m_k(w^{m,k}_{r,h})) \parallel^2  \notag \\
    = &~\mathbb{E} [F^m(\bar{w}^m_{r,h})] -\frac{\eta^m_{r,h}}{2} \mathbb{E}\parallel \nabla F^m(\bar{w}^m_{r,h}) \parallel^2 + \frac{L}{2}{\eta^m}^2 \sum_{k \in V^r_m}p^k_m \sigma^2 \notag \\ &-\frac{\eta^m_{r,h}-{\eta^m_{r,h}}^2L}{2}\underbrace{\mathbb{E}\parallel\sum_{k \in V^r_m}p^m_{k,r}  \nabla F^m_k(w^{m,k}_{r,h})\parallel^2}_{C} \notag \\ &+\frac{1}{2}\eta^m_{r,h}\mathbb{E}\underbrace{\parallel \nabla F^m(\bar{w}^m_{r,h}) -  \sum_{k \in V^r_m}p^m_{k,r} \nabla F^m_k(w^{m,k}_{r,h})\parallel^2}_{D}
\end{align}

\begin{figure*}[htbp]
\centering
\begin{subfigure}{0.3\linewidth}
\includegraphics[width=\linewidth]{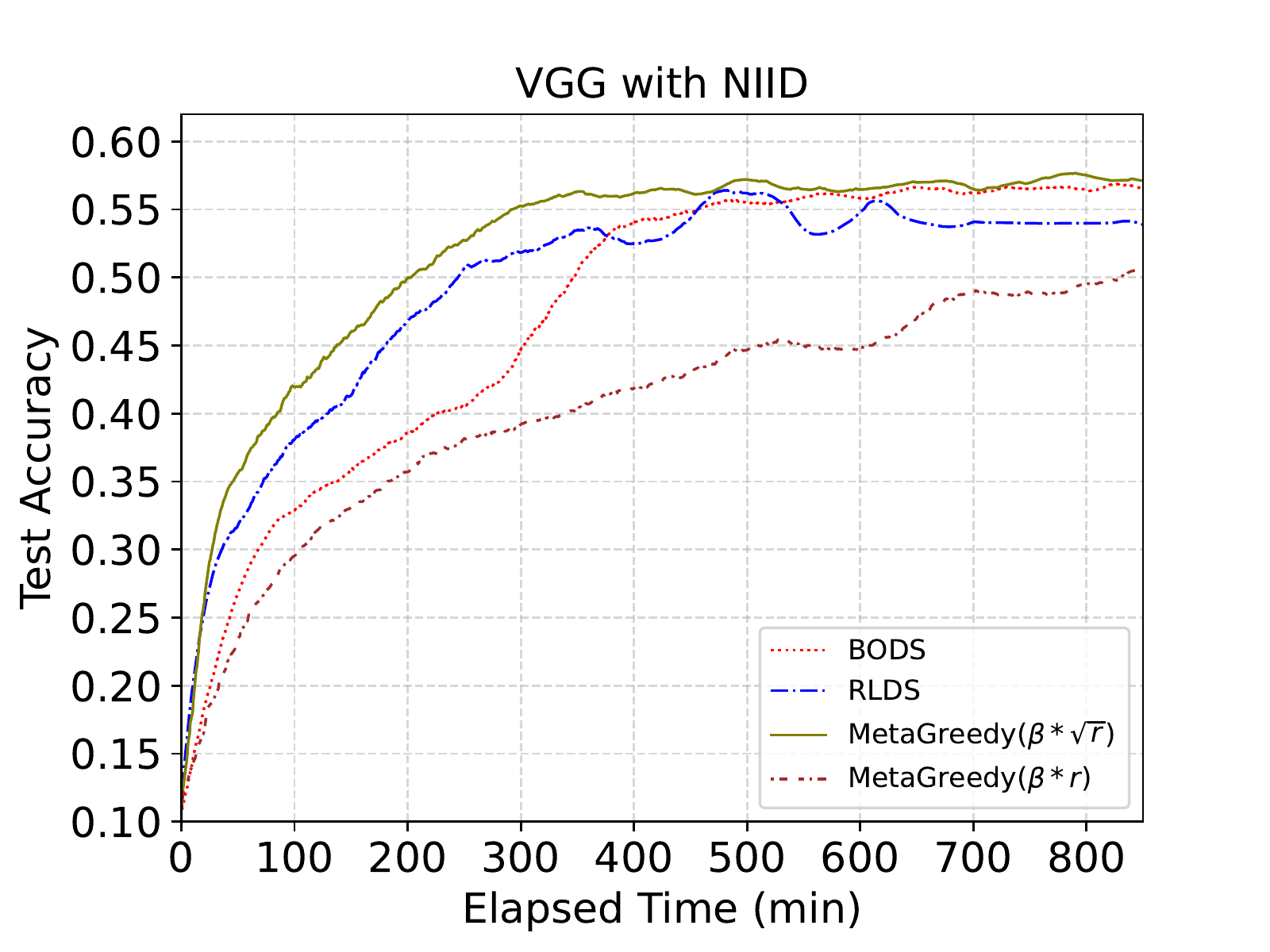}
\vspace{-4mm}
\caption{}
\label{figVgNiidBeta}
\end{subfigure}
\begin{subfigure}{0.3\linewidth}
\includegraphics[width=\linewidth]{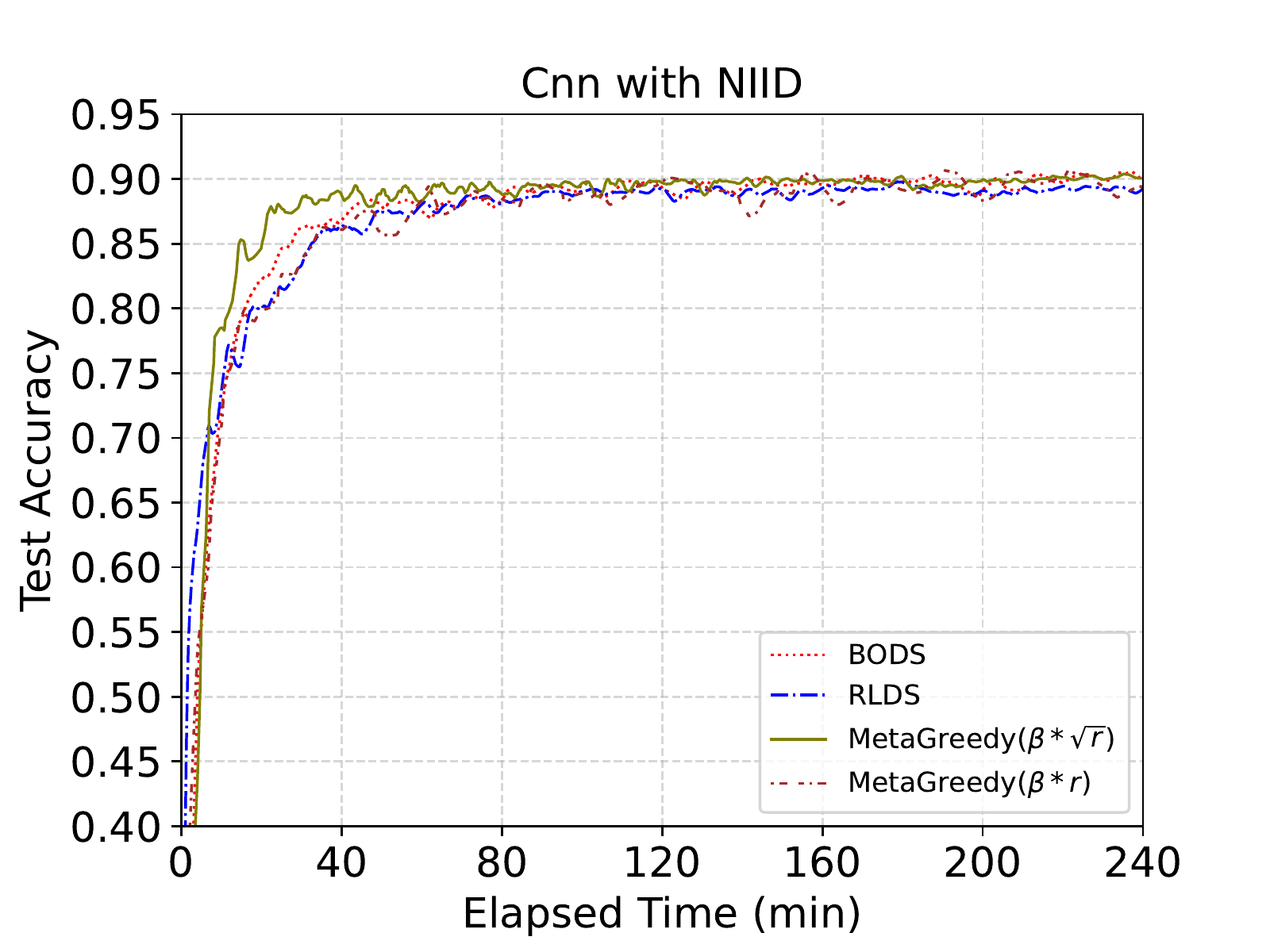}
\vspace{-4mm}
\caption{}
\label{figCnNiidBeta}
\end{subfigure}
\begin{subfigure}{0.3\linewidth}
\includegraphics[width=\linewidth]{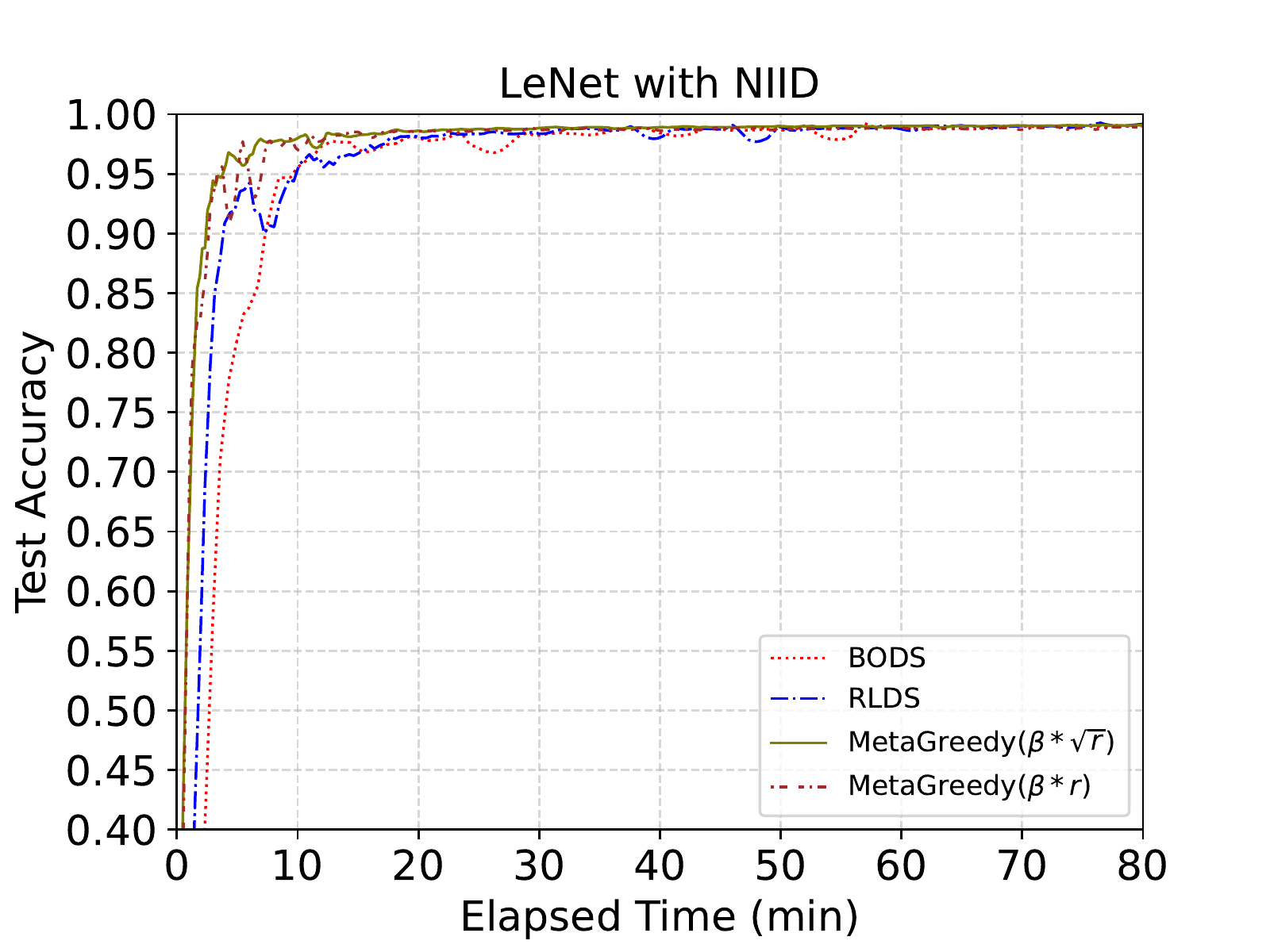}
\vspace{-4mm}
\caption{}
\label{figLeNiidBeta}
\end{subfigure}
\begin{subfigure}{0.3\linewidth}
\includegraphics[width=\linewidth]{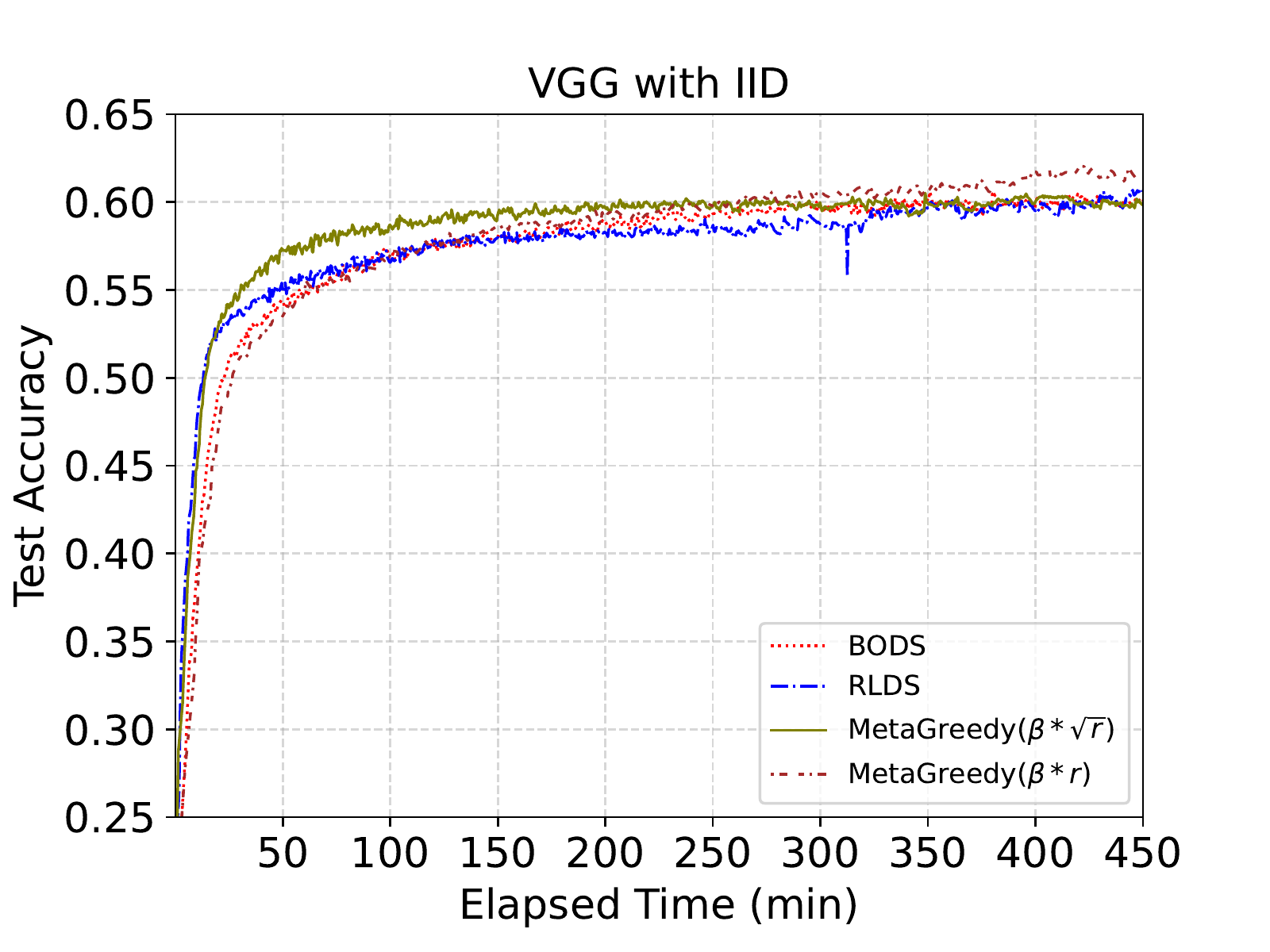}
\vspace{-4mm}
\caption{}
\label{figVgiidBeta}
\end{subfigure}
\begin{subfigure}{0.3\linewidth}
\includegraphics[width=\linewidth]{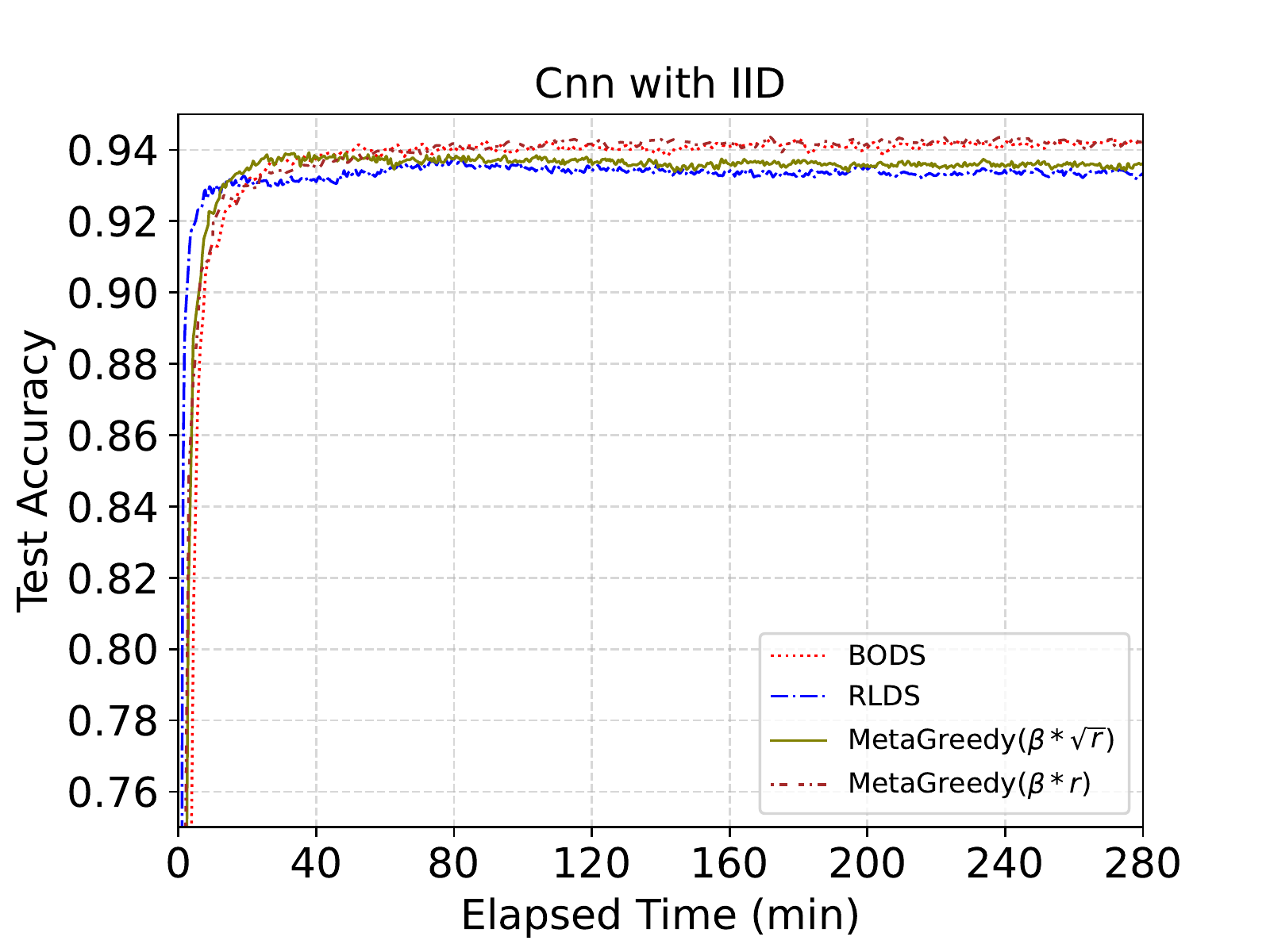}
\vspace{-4mm}
\caption{}
\label{figCniidBeta}
\end{subfigure}
\begin{subfigure}{0.3\linewidth}
\includegraphics[width=\linewidth]{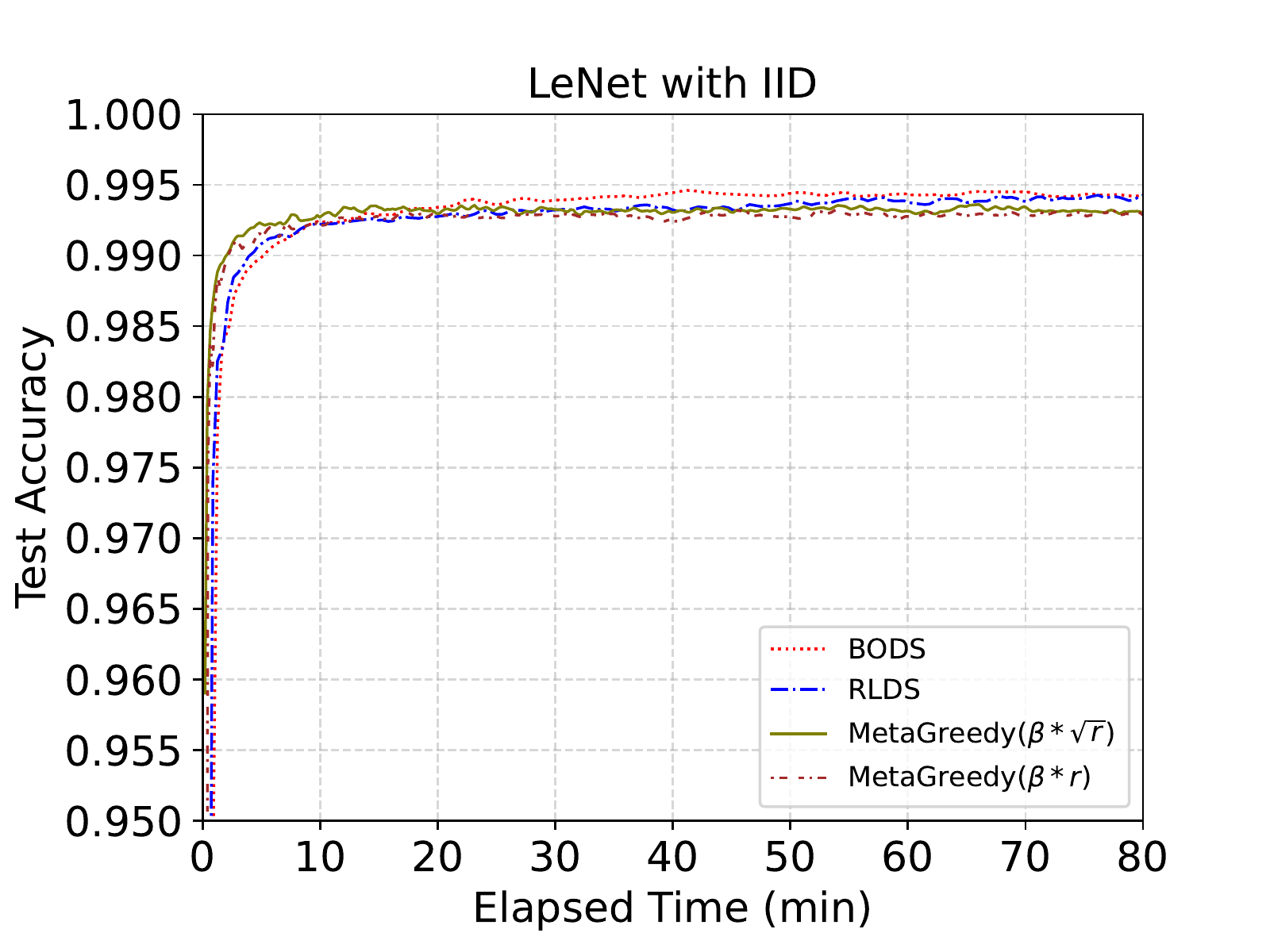}
\vspace{-4mm}
\caption{}
\label{figLeiidBeta}
\end{subfigure}
\vspace{-2mm}
\caption{The convergence accuracy of different jobs in Group A changes over time with diverse settings of $\beta$.}
\label{fig:groupABeta}
\end{figure*}
\begin{figure*}[htbp]
\centering
\begin{subfigure}{0.3\linewidth}
\includegraphics[width=\linewidth]{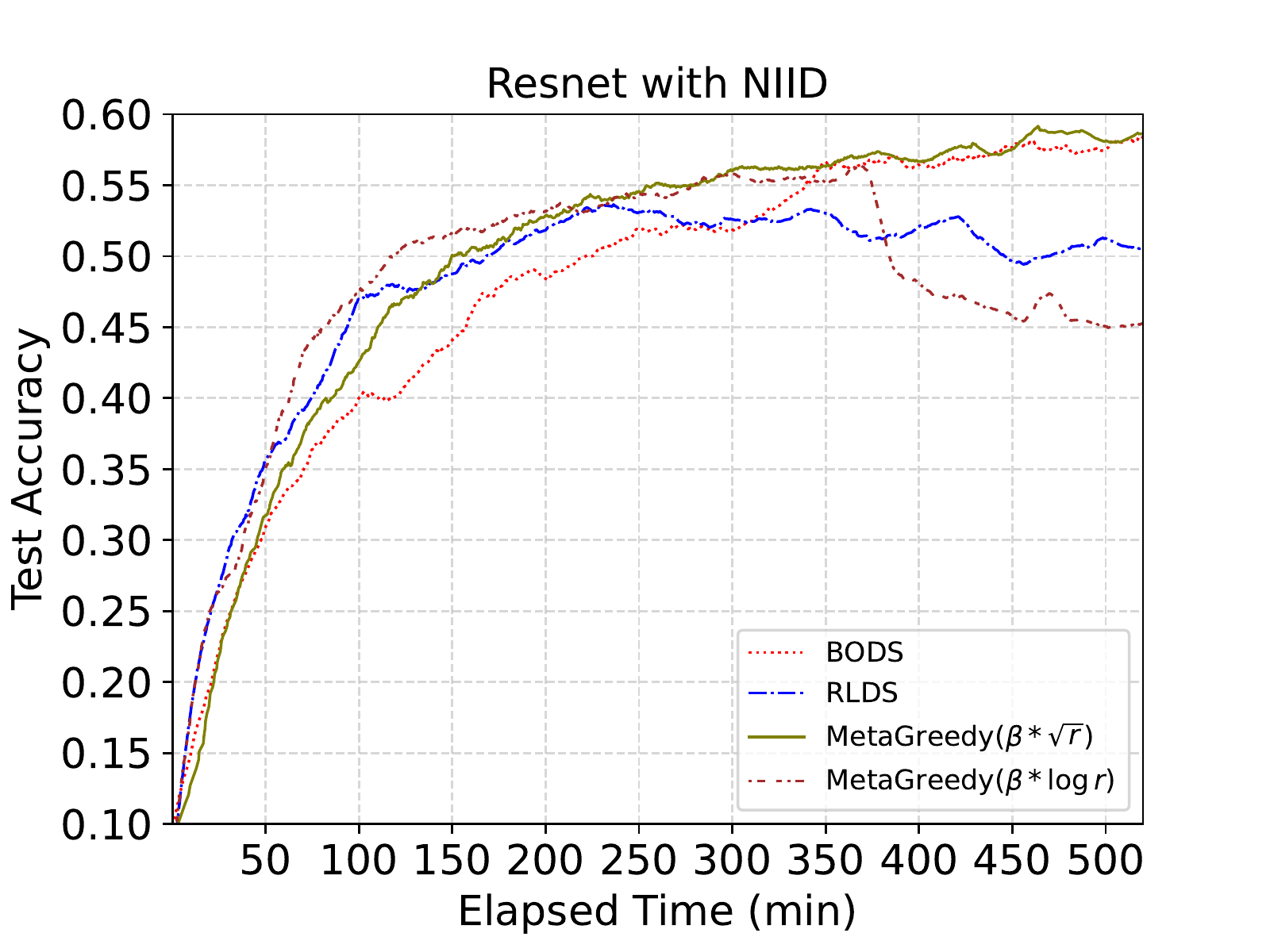}
\vspace{-4mm}
\caption{}
\label{figReNiidBeta}
\end{subfigure}
\begin{subfigure}{0.3\linewidth}
\includegraphics[width=\linewidth]{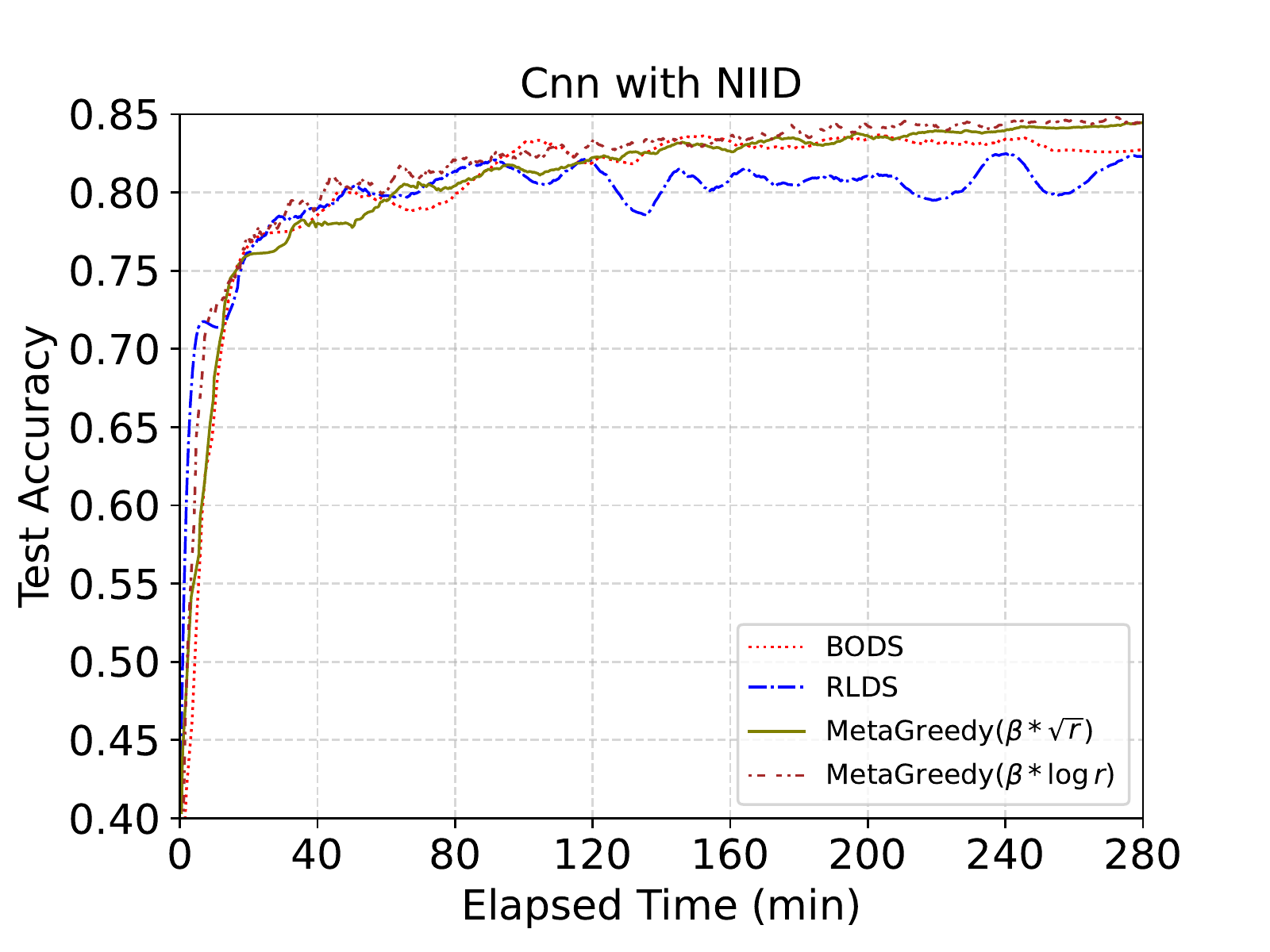}
\vspace{-4mm}
\caption{}
\label{figCBNiidBeta}
\end{subfigure}
\begin{subfigure}{0.3\linewidth}
\includegraphics[width=\linewidth]{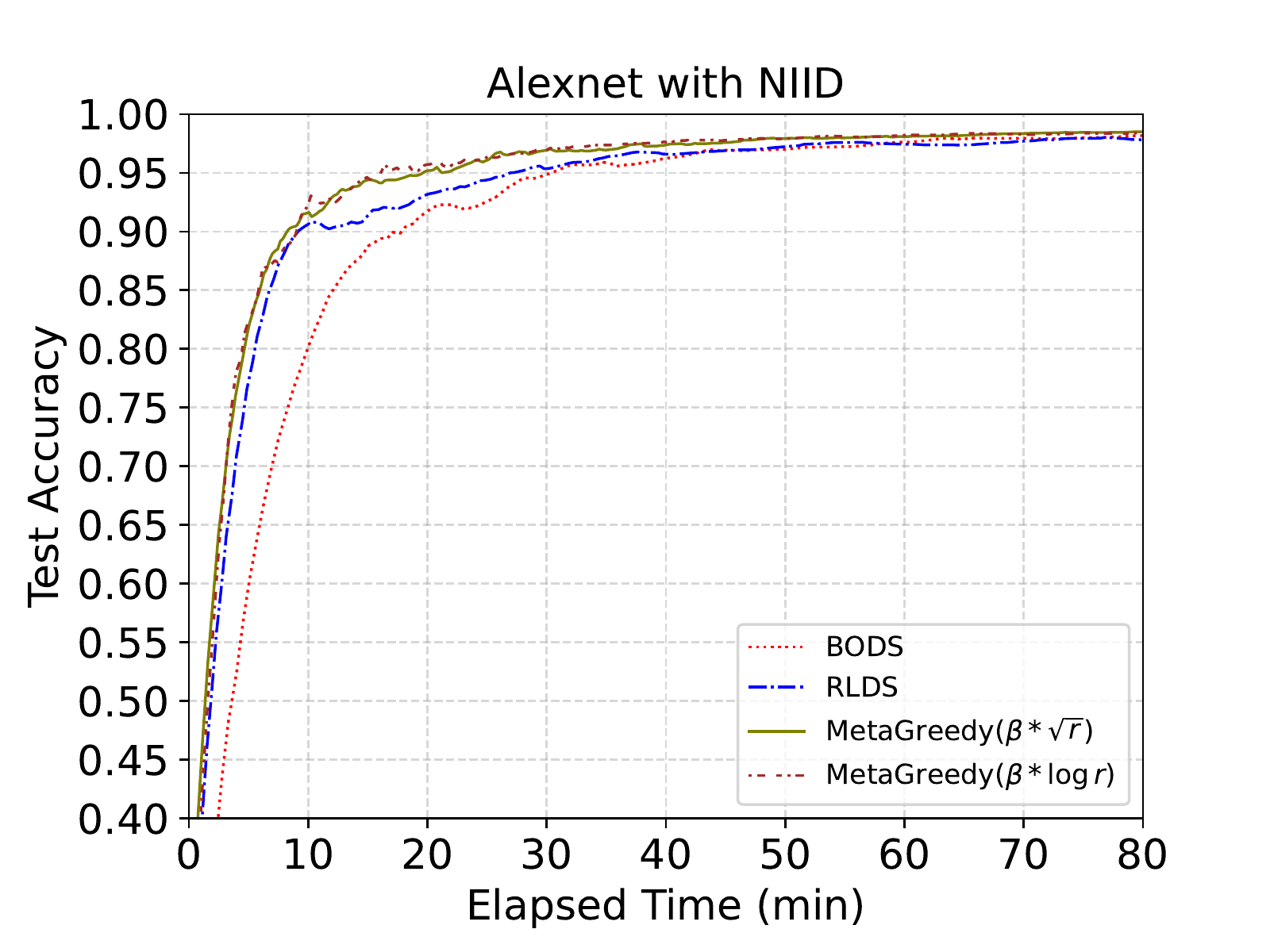}
\vspace{-4mm}
\caption{}
\label{figAlNiidBeta}
\end{subfigure}
\begin{subfigure}{0.3\linewidth}
\includegraphics[width=\linewidth]{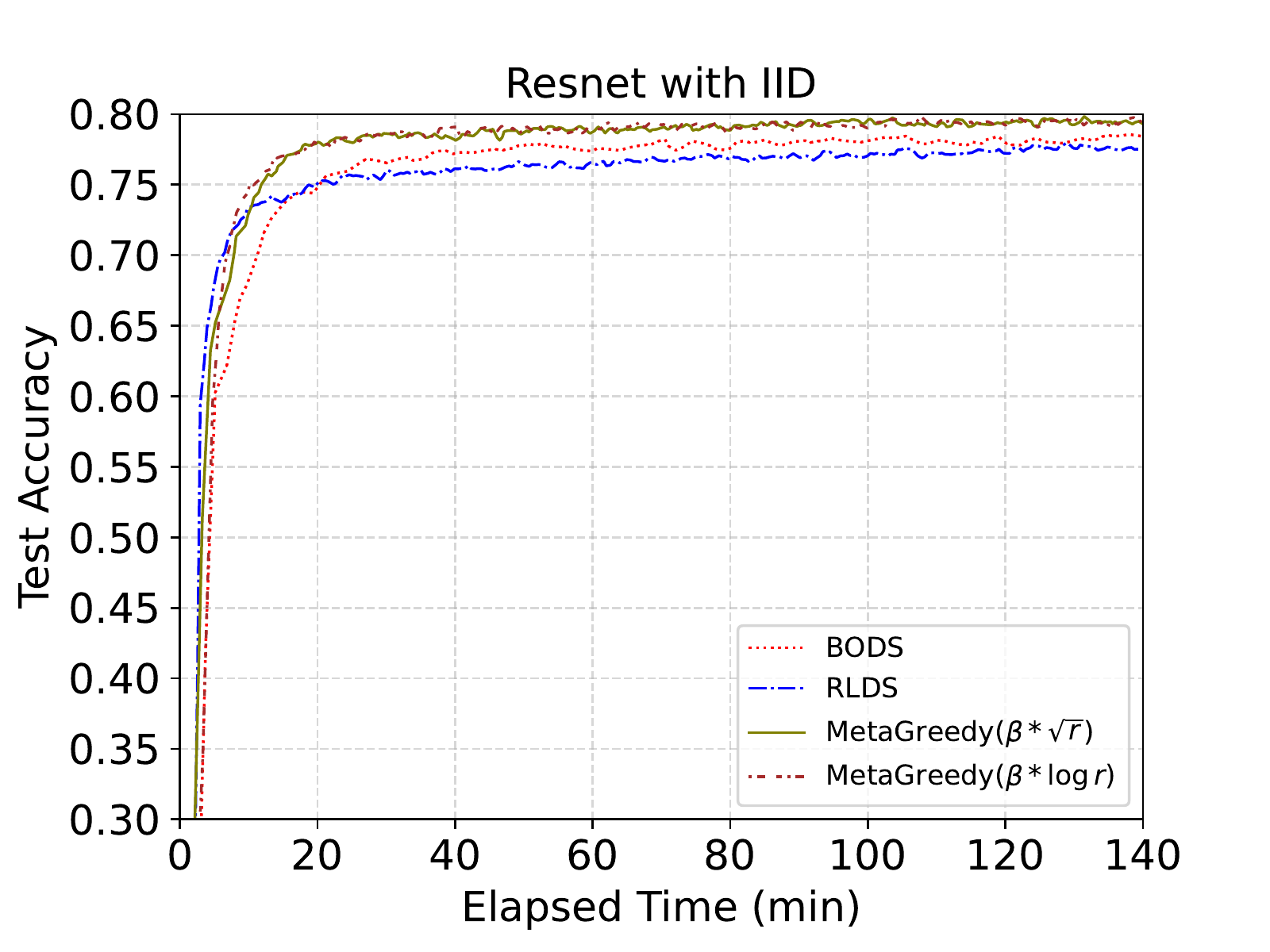}
\vspace{-4mm}
\caption{}
\label{figReiidBeta}
\end{subfigure}
\begin{subfigure}{0.3\linewidth}
\includegraphics[width=\linewidth]{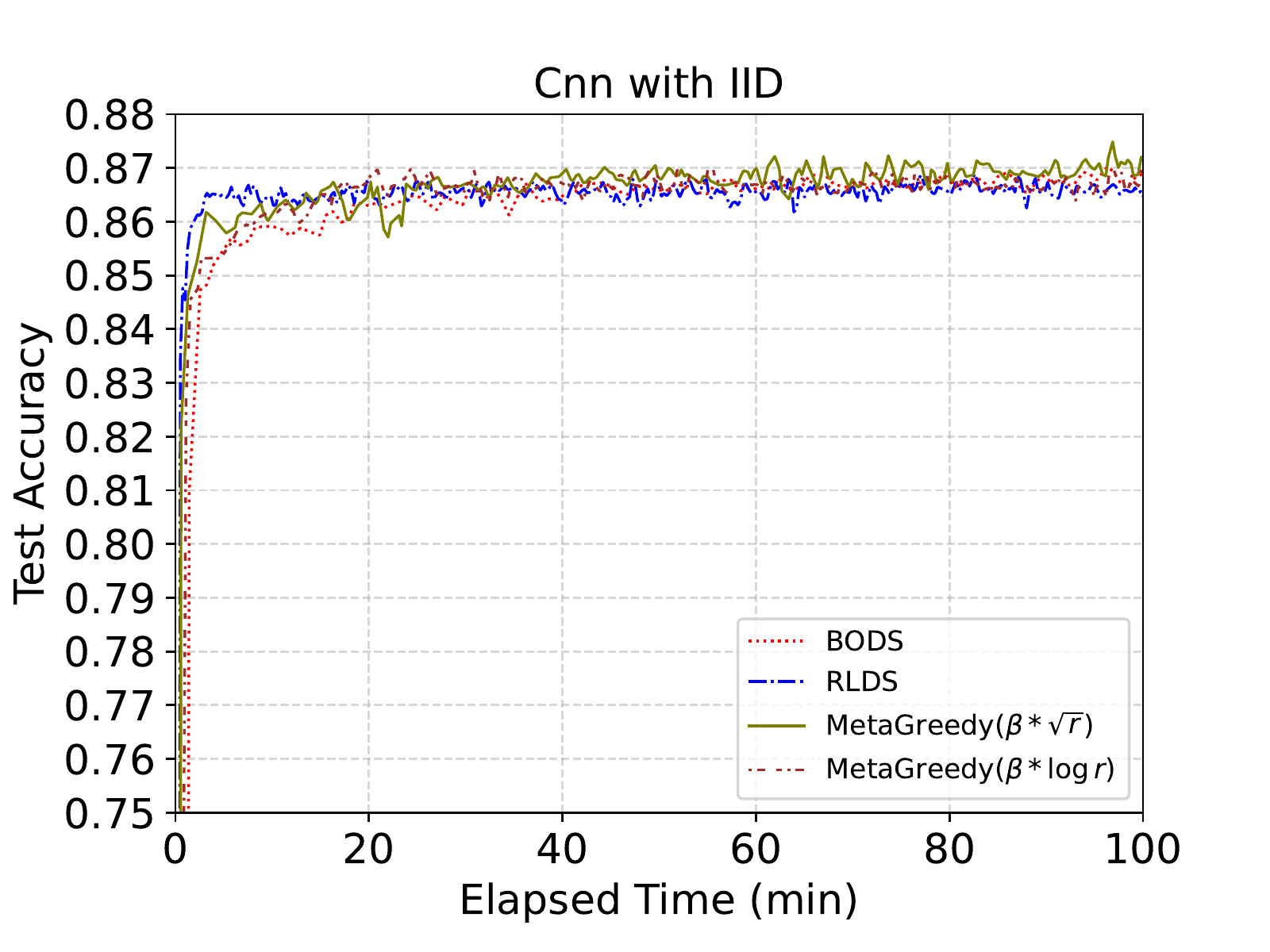}
\vspace{-4mm}
\caption{}
\label{figCBiidBeta}
\end{subfigure}
\begin{subfigure}{0.3\linewidth}
\includegraphics[width=\linewidth]{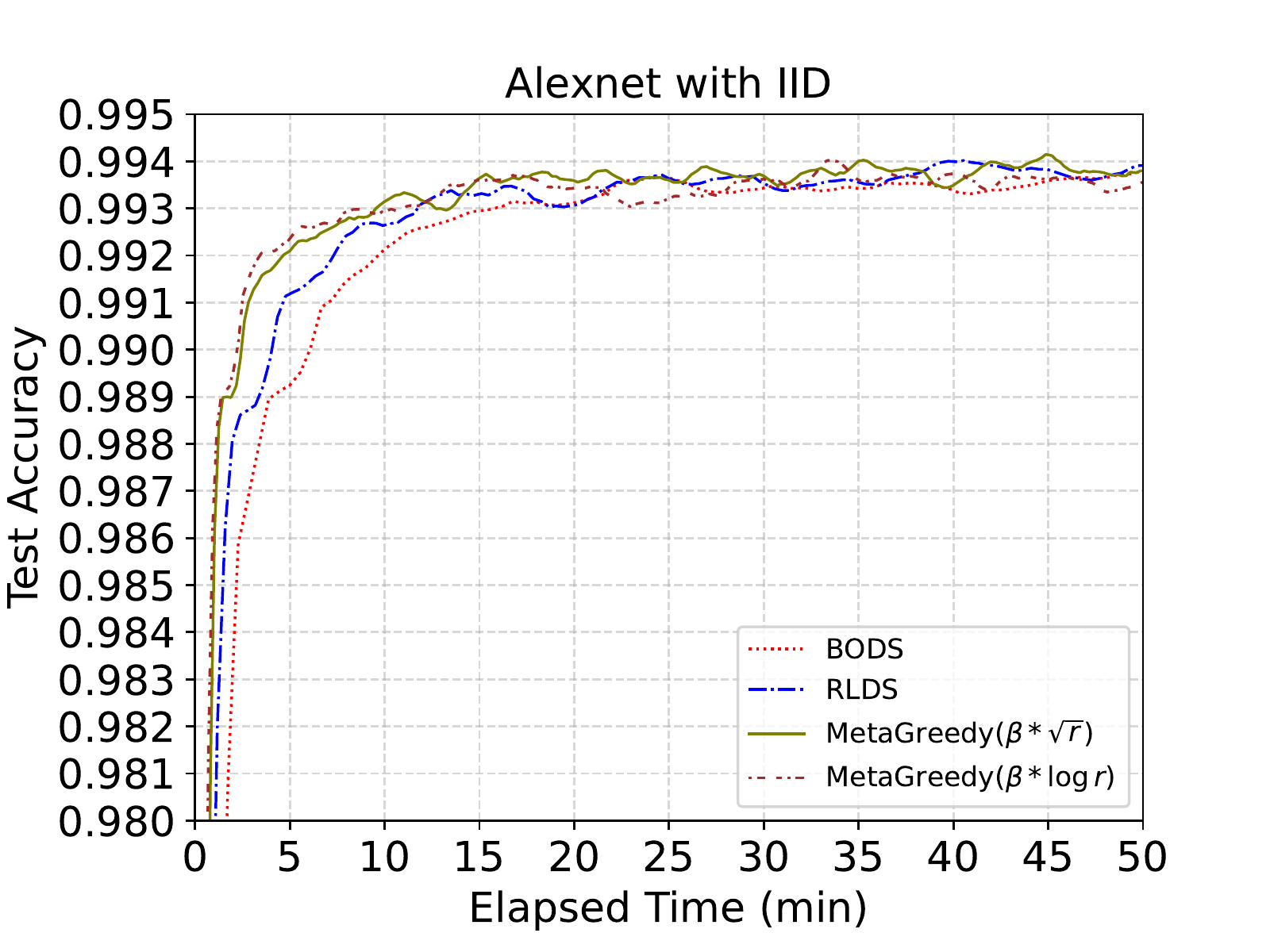}
\vspace{-4mm}
\caption{}
\label{figAliidBeta}
\end{subfigure}
\vspace{-2mm}
\caption{The convergence accuracy of different jobs in Group B changes over time with diverse settings of $\beta$.}
\label{fig:groupBBeta}
\end{figure*}

We take $0 < \eta^m_{r,h} \leq \frac{1}{L}$. As 
$C$ in (\ref{eq:28}) is positive, we have
\begin{align}
     &\mathbb{E}[F^m(\bar{w}^m_{r,h+1})] \notag \\
     \leq &~\mathbb{E}[F^m(\bar{w}^m_{r,h})] -\frac{\eta^m_{r,h}
     }{2}  \mathbb{E}\parallel \nabla F^m(\bar{w}^m_{r,h}) \parallel^2 + \frac{L}{2}{\eta^m}^2 \sum_{k \in V^r_m}p^k_m \sigma^2
    \notag \\
    &+\frac{\eta^m_{r,h}}{2}\mathbb{E}\underbrace{\parallel \nabla F^m(\bar{w}^m_{r,h}) -  \sum_{k \in V^r_m}p^m_{k,r} \nabla F^m_k(w^{m,k}_{r,h})\parallel^2}_{D}
\end{align}
For D we have,
\begin{align}
    &\mathbb{E}\parallel \nabla F^m(\bar{w}^m_{r,h}) - \sum_{k \in V^r_m}p^m_{k,r} \nabla F^m_k(w^{m,k}_{r,h})\parallel^2 \notag \\
    = &~\mathbb{E}\parallel \sum_{k \in V^r_m} p^m_{k,r} \nabla F^m_k(\bar{w}^m_{r,h})
    - \sum_{k \in V^r_m}p^m_{k,r} \nabla F^m_k(w^{m,k}_{r,h})\parallel^2 \notag\\
    = &~\mathbb{E}\parallel \sum_{k \in V^r_m} p^m_{k,r} \bigl(\nabla F^m_k(\bar{w}^m_{r,h})
    -  \nabla F^m_k(w^{m,k}_{r,h})\bigr)\parallel^2 \notag \\
    \leq &~\mathbb{E} [\sum_{k \in V^r_m} p^m_{k,r}  \parallel \nabla F^m_k(\bar{w}^m_{r,h})
    -  \nabla F^m_k(w^{m,k}_{r,h}) \parallel^2] \notag \\
    \leq &~L^2 \mathbb{E}\sum_{k \in V^r_m} p^m_{k,r} 
     \parallel \bar{w}^m_{r,h} - w^{m,k}_{r,h} \parallel^2
\end{align}
where the first inequality results from Jensen's inequality. The second inequality results from Assumption 1. Then we have,
\begin{align}
    &\mathbb{E}\sum_{k \in V^r_m} p^m_{k,r} 
    \parallel w^{m,k}_{r,h} - \bar{w}^m_{r,h} \parallel^2  \notag \\
    = &~\mathbb{E}\sum_{k \in V^r_m} p^m_{k,r} 
    \parallel (w^{m,k}_{r,h} - w^{m,k}_{r,0}) - (\bar{w}^m_{r,h} - w^{m,k}_{r,0}) \parallel^2 \notag \\
    \leq &~\mathbb{E}\sum_{k \in V^r_m} p^m_{k,r}
    \parallel w^{m,k}_{r,h} - w^{m,k}_{r,0} \parallel^2 \notag \\
    = &~\mathbb{E}\sum_{k \in V^r_m} p^m_{k,r}
    \parallel \sum^{h-1}_{h'=0} {\eta^m_{r,h'}} g^{m,k}_{r,h'} \parallel^2 \notag \\
    = &~\mathbb{E}\sum_{k \in V^r_m} p^m_{k,r}
    \parallel \sum^{h-1}_{h'=0} {\eta^m_{r,h'}} \nabla f^m_k(w^{m,k}_{r,h'}) \parallel^2 \notag \\
    \leq &~\sum_{k \in V^r_m} p^m_{k,r} (H-1) \sum^{h-1}_{h'=0} {\eta^m_{r,h'}}^2 \mathbb{E} \parallel \nabla f^m_k(w^{m,k}_{r,h}) \parallel^2 \notag \\
    \leq &~\sum_{k \in V^r_m} p^m_{k,r} (H-1)^2 {\eta^m_{r,0}}^2 G^2 \notag \\
    \leq &~Q^2 (H-1)^2 {\eta^m_{r,h}}^2\sum_{k \in V^r_m} p^m_{k,r} G^2
\end{align}
In the first inequality, we use $\mathbb{E}||X - \mathbb{E}X||^2 \le{\mathbb{E}|| X ||^2}$ where $X = w^{m,k}_{r,h}- w^{m,k}_{r,0}$. In the second inequality from, We use the following steps:
\begin{align}
    Var(x_k) = \mathbb{E}({x_k}^2) - (\mathbb{E}(x_k))^2 &\ge 0 \notag \\
    \frac{1}{h} \sum^{h-1}_{h'=0} \parallel x_k\parallel^2 - \parallel \frac{1}{h} \sum^{h-1}_{h'=0} x_k\parallel^2 &\ge 0 \notag \\
    \frac{1}{h^2} \parallel \sum^{h-1}_{h'=0} x_k \parallel^2 \le \frac{1}{h} \sum^{h-1}_{h'=0} \parallel x_k \parallel^2 \notag \\
    \parallel \sum^{h-1}_{h'=0} x_k \parallel^2 \le h\sum^{h-1}_{h'=0} \parallel x_k \parallel^2 \notag
\end{align}
we take $\eta^m_{r, h'} \leq \eta^m_{r, 0} = \eta^m_{r-1, H}$ with $0 \leq h' \leq H-1$ and $0 \leq h \leq H-1$. The third inequality results from Assumption 4. The last inequality, we assume that $\eta^m_{r, 0} \leq Q \eta^m_{r, h'} $. Therefore, we have
\begin{align}
    &\mathbb{E}\parallel \nabla F^m(\bar{w}^m_{r,h}) - \sum_{k \in V^r_m}p^m_{k,r} \nabla F^m_k(w^{m,k}_{r,h})\parallel^2 \notag \\
    \leq &~L^2 Q^2 (H-1)^2 {\eta^m_{r,h}}^2\sum_{k \in V^r_m} p^m_{k,r} G^2
\end{align}
Then, we have
\begin{align} \label{eq: ineq}
     \mathbb{E}[F^m(\bar{w}^m_{r,h+1})] &\leq \mathbb{E}[F^m(\bar{w}^m_{r,h})] -\frac{\eta^m_{r,h}
     }{2}  \mathbb{E}\parallel \nabla F^m(\bar{w}^m_{r,h}) \parallel^2 \notag \\
    & + \frac{L}{2}{\eta^m_{r,h}}^2 \sum_{k \in V^r_m}p^m_{k,r} \sigma^2
    \notag \\
    &+\frac{L^2}{2} Q^2 (H-1)^2 {\eta^m_{r,h}}^3\sum_{k \in V^r_m} p^m_{k,r} G^2
\end{align}
Divide (\ref{eq: ineq}) both sides by $\frac{{\eta^m_{r,h}}}{2}$ and rearrange it yields,
\begin{align}
    \mathbb{E} \parallel \nabla F^m(\bar{w}^m_{r,h}) \parallel^2 &\leq
    \frac{2}{{\eta^m_{r,h}}} (\mathbb{E}[F^m(\bar{w}^m_{r,h})] - \mathbb{E}[F^m(\bar{w}^m_{r,h+1})]) \notag \\
    &+ L {\eta^m_{r,h}} \sum_{k \in V^r_m}p^m_{k,r} \sigma^2 \notag \\
    &+ L^2 Q^2 (H-1)^2 {\eta^m_{r,h}}^2 \sum_{k \in V^r_m} p^m_{k,r} G^2
\end{align}
Summing from $r = 1, h = 1$ to $r = R, h = H$ and dividing both sides by $RH$ yields 
\begin{align}
    &\frac{1}{RH} \sum^{R}_{r=1} \sum^H_{h=1} \mathbb{E} \parallel \nabla F^m(\bar{w}^m_{r,h}) \parallel^2 \notag \\
    \leq &~\frac{2}{{\eta^m_{r,h}}RH}(F^m(\bar{w}^m_{1,1}) - \mathbb{E}[F^m(\bar{w}^m_{R,H})]) \notag \\
    &+  L {\eta^m_{r,h}} \sum_{k \in V^r_m}p^m_{k,r} \sigma^2 + L^2 Q^2 (H-1)^2 {\eta^m_{r,h}}^2 \sum_{k \in V^r_m} p^m_{k,r} G^2 \notag \\
    \leq &~\frac{2}{{\eta^m_{r,h}}RH}(F^m(\bar{w}^m_{1,1}) - \mathbb{E}[F^m(\bar{w}^m_{R,H})]) + L {\eta^m_{r,h}} \sigma^2\notag \\
    &+ L^2 Q^2 (H-1)^2 {\eta^m_{r,h}}^2 G^2 \notag \\
    \leq &~\frac{2}{{\eta^m_{r,h}}RH}(F^m(\bar{w}^m_{1,1}) - {F^m}^*) + L^2 Q^2 (H-1)^2 {\eta^m_{r,h}}^2 G^2 \notag \\
    &+ L {\eta^m_{r,h}} \sigma^2
\end{align}
We choose $\eta^m_{r} = \frac{1}{L\sqrt{RH}}$. Then we have
\begin{align}
    &\frac{1}{RH} \sum^{R}_{r=1} \sum^H_{h=1} \mathbb{E} \parallel \nabla F^m(\bar{w}^m_{r,h}) \parallel^2 \notag \\
    \leq &~\frac{2L}{\sqrt{RH}}(F^m(\bar{w}^m_{1,1}) - {F^m}^*) + \frac{1}{\sqrt{RH}} \sigma^2 + \frac{Q^2}{RH} (H-1)^2 G^2 
\end{align}
If we further choose $Q \leq (RH)^{\frac{1}{4}}$, we have
\begin{align}
    &\frac{1}{RH} \sum^{R}_{r=1} \sum^H_{h=1} \mathbb{E} \parallel \nabla F^m(\bar{w}^m_{r,h}) \parallel^2 \notag \\
    \leq &~\frac{2L}{\sqrt{RH}}(F^m(\bar{w}^m_{1,1}) - {F^m}^*) + \frac{1}{\sqrt{RH}} \sigma^2\notag \\
    &+ \frac{1}{\sqrt{RH}} (H-1)^2 G^2 \notag \\
    = &~\mathcal{O}(\frac{1}{\sqrt{RH}})
\end{align}

\begin{table*}[t!]
\centering
\caption{The time to achieve target accuracy for divers models and methods. The ``()'' after the model represents the target accuracy. ``/'' represents that the training cannot achieve the target accuracy with the corresponding scheduling method while the ``()'' represents the highest accuracy during the training process.}
\label{tab:realTime}
\begin{tabular}{cccccccc}
\hline
\multicolumn{8}{c}{Time (s)}                                                                                  \\ \hline
              & Random    & Genetic & FedCS     & Greedy    & BODS & RLDS   & Meta-Greedy      \\
CNN (0.928)    & / (0.908) & 1997.71 & 1909.72   & / (0.900) & 1369.10  & / (0.915)        & \textbf{1351.37} \\
VGG (0.870)    & 2414.70   & 1714.55 & / (0.848) & / (0.848) & 1553.10  & 2591.16          & \textbf{1493.06} \\
ResNet (0.680) & 2533.00   & 2839.45 & 2474.19   & 2876.11   & 2209.64  & 2446.88          & \textbf{2198.79} \\
ResNet (0.808) & 5552.51   & 5117.59 & / (0.776) & / (0.806) & 4553.23  & \textbf{3539.88} & 4983.88          \\ \hline
\end{tabular}
\end{table*}

\subsection*{Experimental Results}

\subsubsection*{Experimentation with Real Mobile Devices}

We carried out an experimentation with 20 real devices (mobile devices) and a parameter server on the Baidu AI Cloud \cite{BaiduCloud}. The devices are summarized in Table \ref{tab:devices}. We carry out the experimentation with a synthetic CNN model of 4 layers, a VGG model of 6 layers, and a ResNet model of 13 layers. 

The time to achieve target accuracy is shown in Table \ref{tab:realTime}. From the table, we can find that Meta-Greedy corresponds to the shortest time (up to 42.4\% shorter than others) to achieve the target accuracy of simple models, i.e. CNN and VGG, while BODS outperforms baseline methods (up to 35.7\%). With a complex model, i.e., ResNet, Meta-Greedy corresponds to excellent efficient training (up to 23.5\% compared with others) for a low target accuracy (0.680) while RLDS significantly outperforms baseline methods (up to 36.2\%) for a high target accuracy (0.808). This result implies that RLDS favors complex models while BODS favors simple models. 

\begin{table}[t!]
\centering
\caption{Summary of devices. }
\label{tab:devices}
\begin{tabular}{@{}ll@{}}
\toprule
Device type              & RAM size  \\ \midrule
HUAWEI Mate20    & 6G  \\
OPPO A72         & 8G  \\
Galaxy M11       & 8G  \\
Redmi Note9 Pro  & 8G  \\
HUAWEI P40 Pro   & 8G  \\
Realme GT2       & 8G  \\
Smartisan R2     & 8G  \\
HUAWEI nova2     & 4G  \\
Redmi K20        & 6G  \\
HUAWEI MatePad       & 8G  \\
HONOR 60          & 8G  \\
HUAWEI M6        & 4G  \\
Galaxy 20U       & 8G  \\
HONOR V10        & 4G  \\
Redmi Note11     & 8G  \\
HUAWEI nova5i    & 6G  \\
Redmi K50 Pro     & 12G \\
Galaxy S21       & 8G  \\
HONOR Play4      & 8G  \\
HUAWEI MatPad       & 6G  \\ \bottomrule
\end{tabular}
\end{table}

\subsubsection*{Impact of $\Omega(r)$}

As shown in Figures \ref{fig:groupABeta} and \ref{fig:groupBBeta}, Meta-Greedy with $\Omega(r)=\sqrt{r}$ outperforms other methods, i.e., $\Omega(r)=r$ and $\Omega(r)=\log r$, in terms of both convergence accuracy and convergence speed. 

\subsubsection*{RLDS \& BODS with Simple and Complex Jobs}

\begin{figure*}[htbp]
\centering
\begin{subfigure}{0.3\linewidth}
\includegraphics[width=\linewidth]{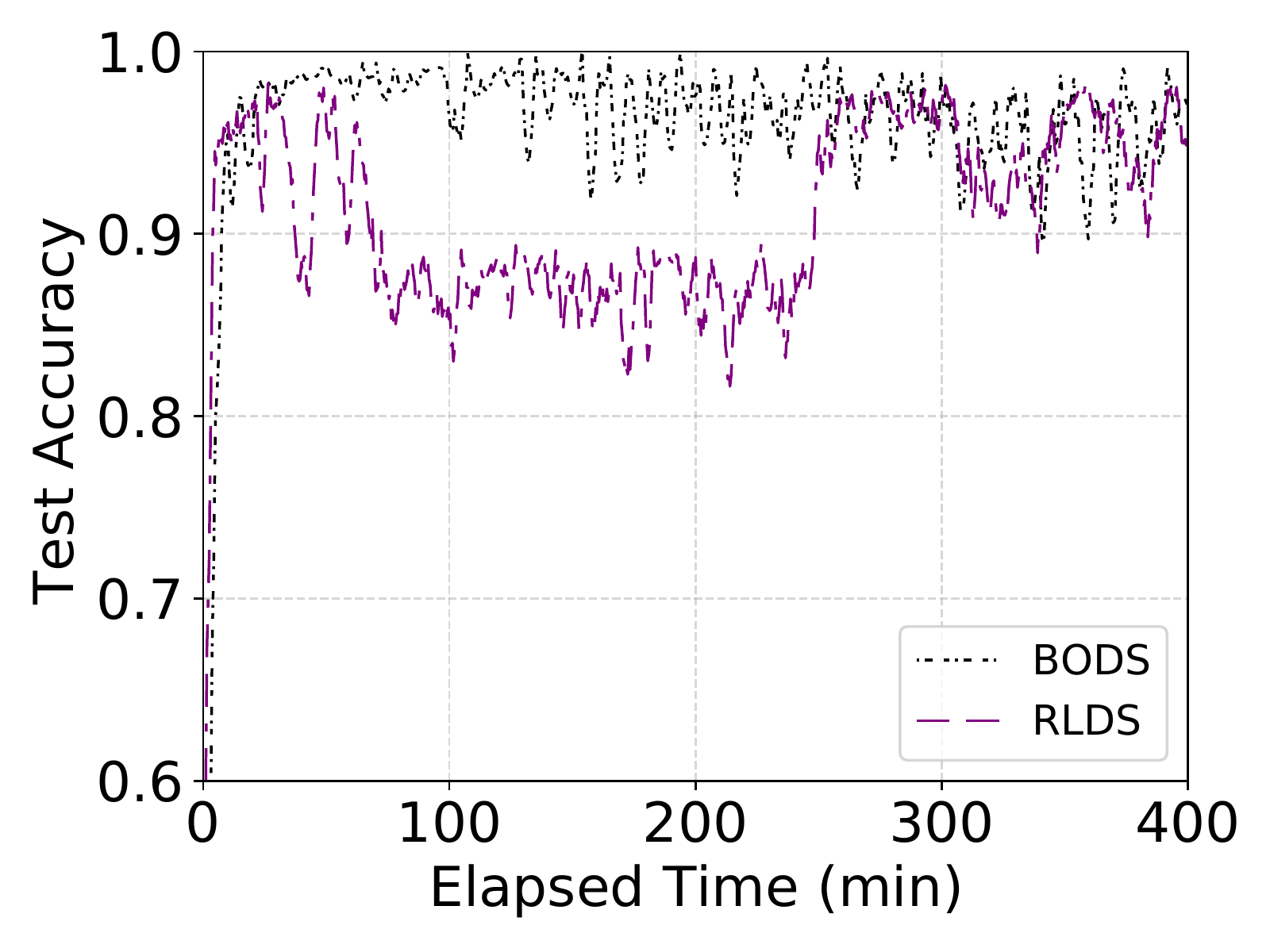}
\vspace{-4mm}
\caption{CNN}
\label{fig:groupR1Exp1CNNEMnistDigitalNIID}
\end{subfigure}\\
\begin{subfigure}{0.3\linewidth}
\includegraphics[width=\linewidth]{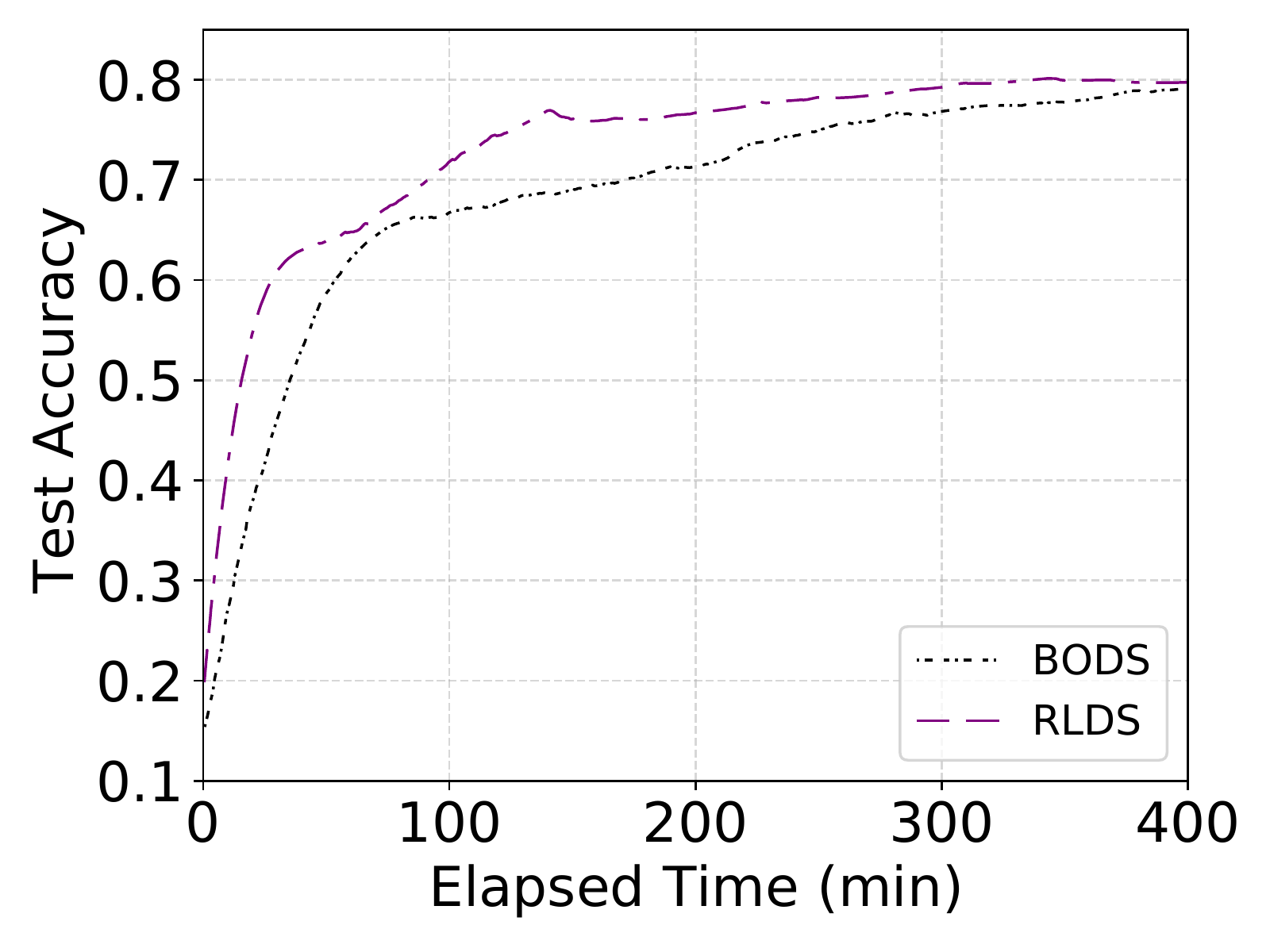}
\vspace{-4mm}
\caption{VGG}
\label{fig:groupR1Exp1VGGCifar10NIID}
\end{subfigure}
\begin{subfigure}{0.3\linewidth}
\includegraphics[width=\linewidth]{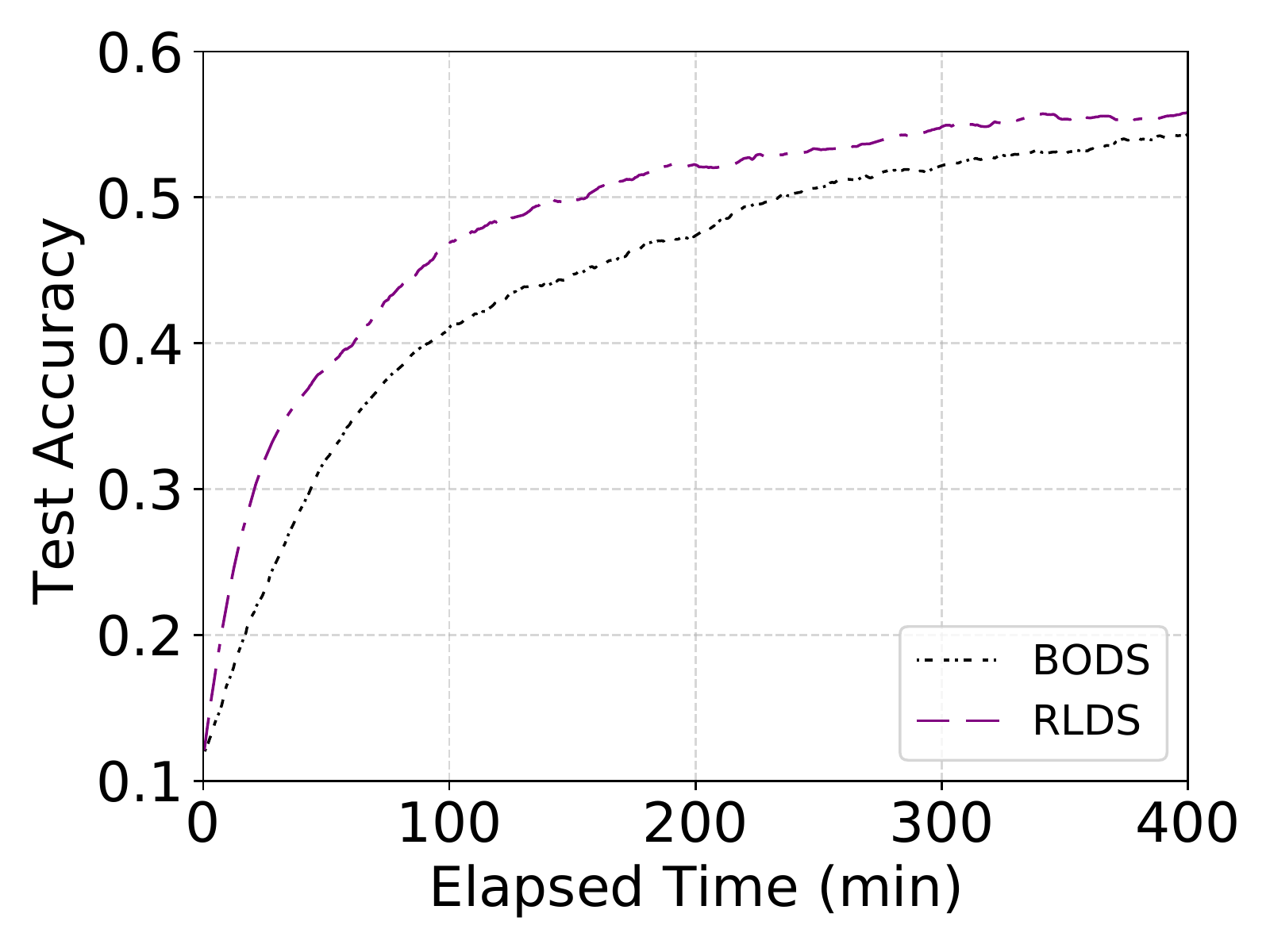}
\vspace{-4mm}
\caption{ResNet}
\label{fig:groupR1Exp1ResNetCifar10NIID}
\end{subfigure}
\vspace{-2mm}
\caption{The accuracy of different jobs (CNN, VGG, and ResNet) with non-IID data.}
\label{fig:groupR1Exp1}
\end{figure*}

\begin{figure*}[htbp]
\centering
\begin{subfigure}{0.3\linewidth}
\includegraphics[width=\linewidth]{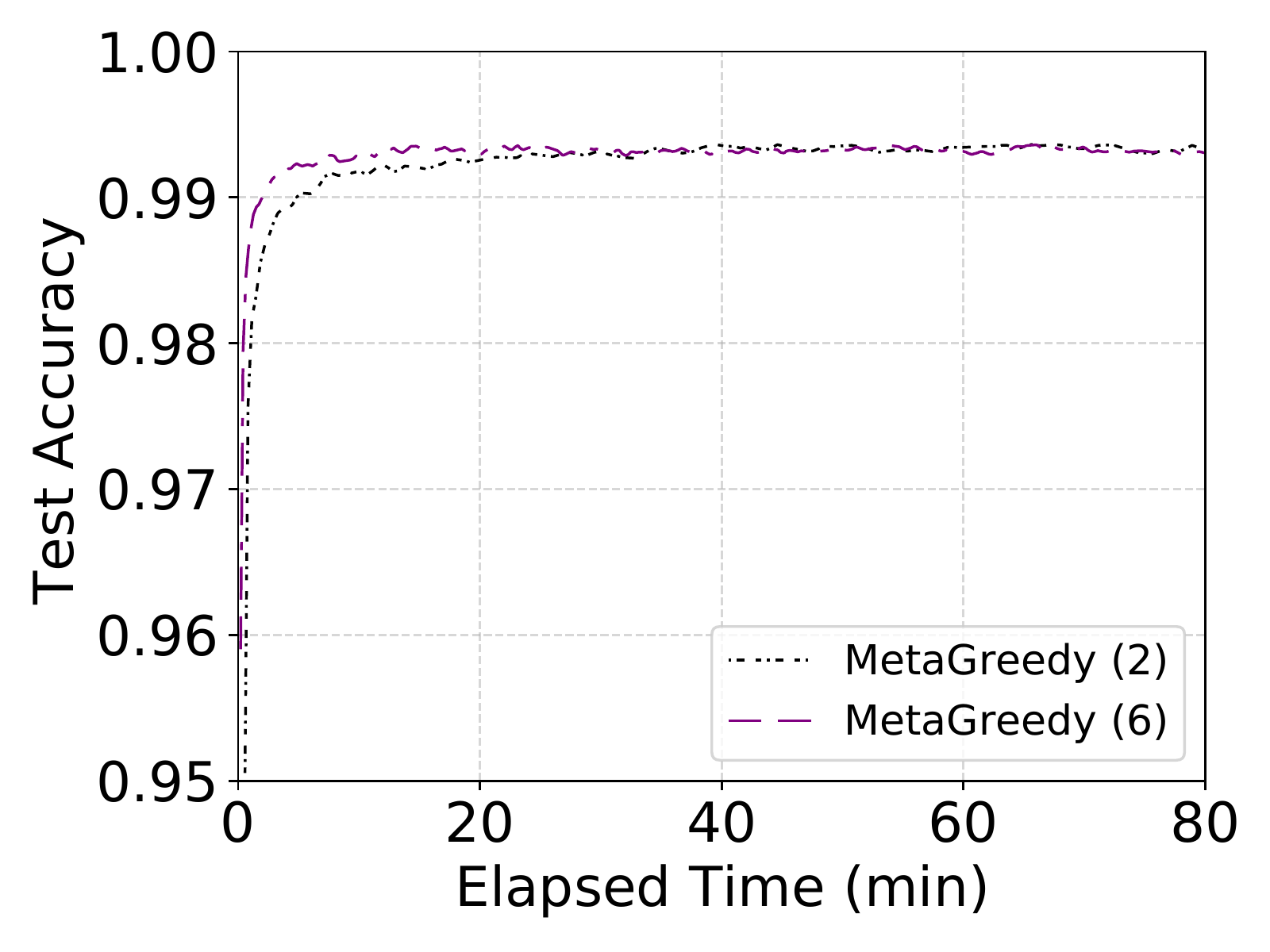}
\vspace{-4mm}
\caption{LeNet with IID data}
\label{fig:groupR1Exp2LeNetIID}
\end{subfigure}
\begin{subfigure}{0.3\linewidth}
\includegraphics[width=\linewidth]{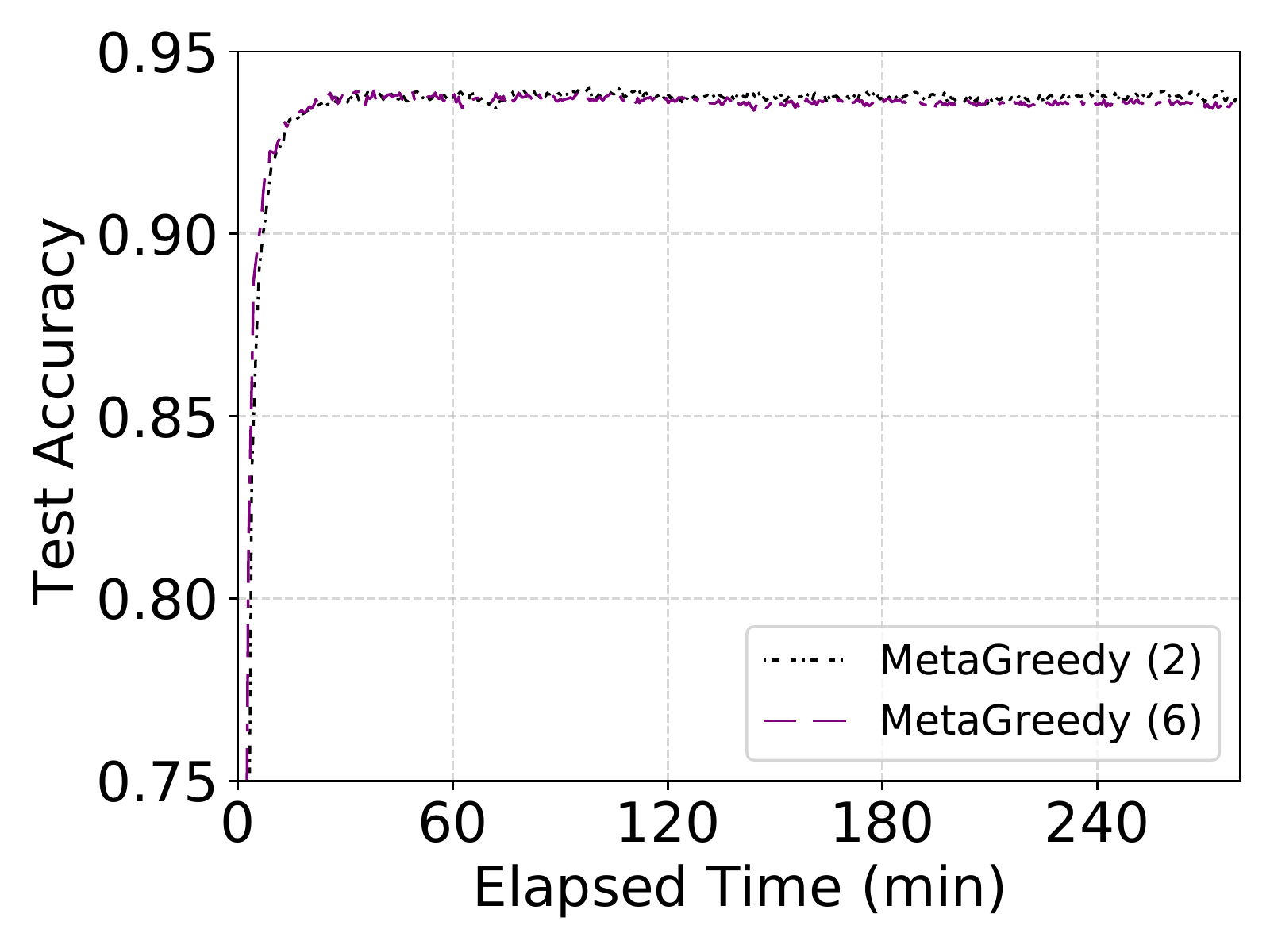}
\vspace{-4mm}
\caption{CNN with IID data}
\label{fig:groupR1Exp2CNNIID}
\end{subfigure}
\begin{subfigure}{0.3\linewidth}
\includegraphics[width=\linewidth]{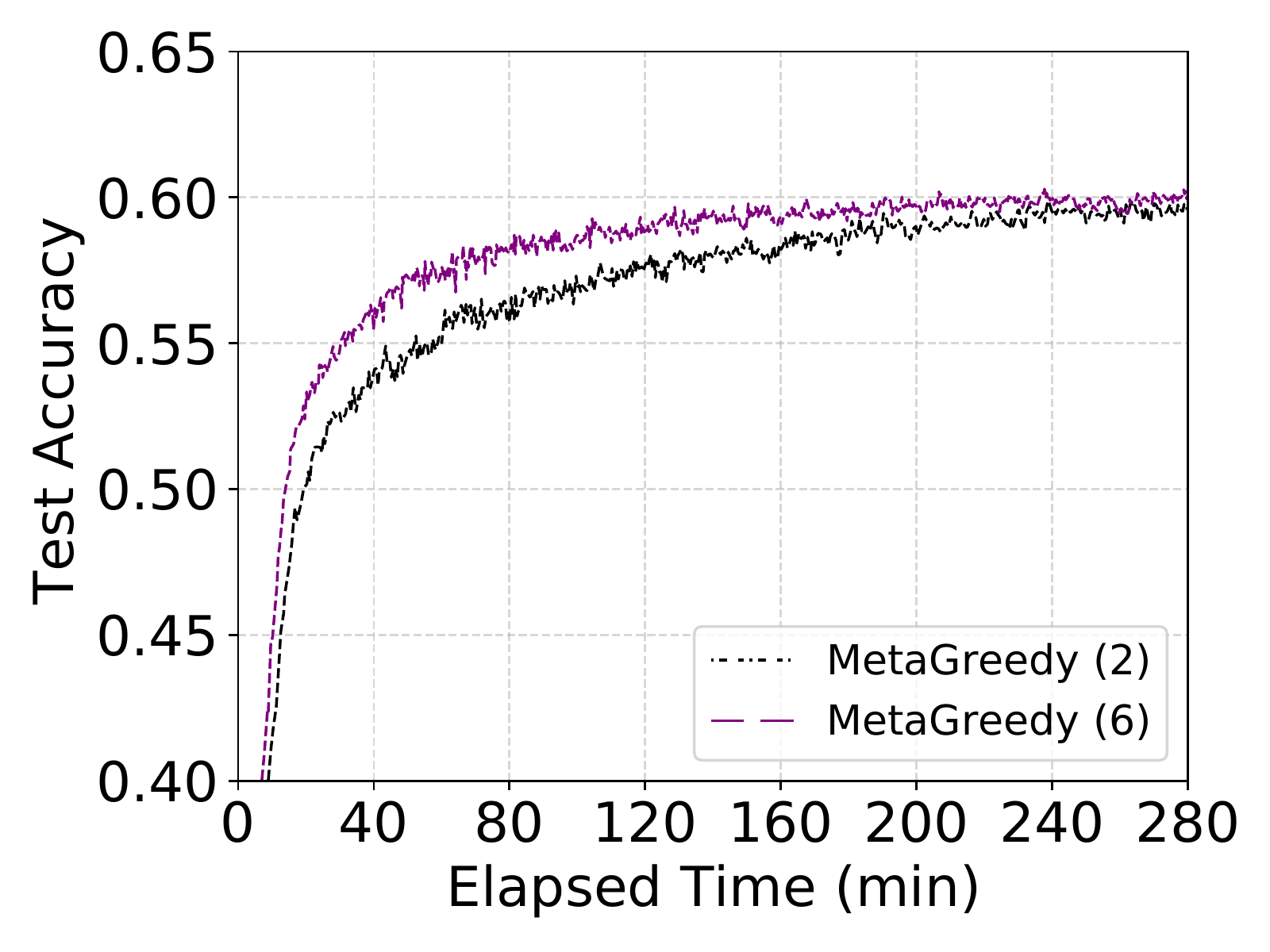}
\vspace{-4mm}
\caption{VGG with IID data}
\label{fig:groupR1Exp2VGGIID}
\end{subfigure}
\begin{subfigure}{0.3\linewidth}
\includegraphics[width=\linewidth]{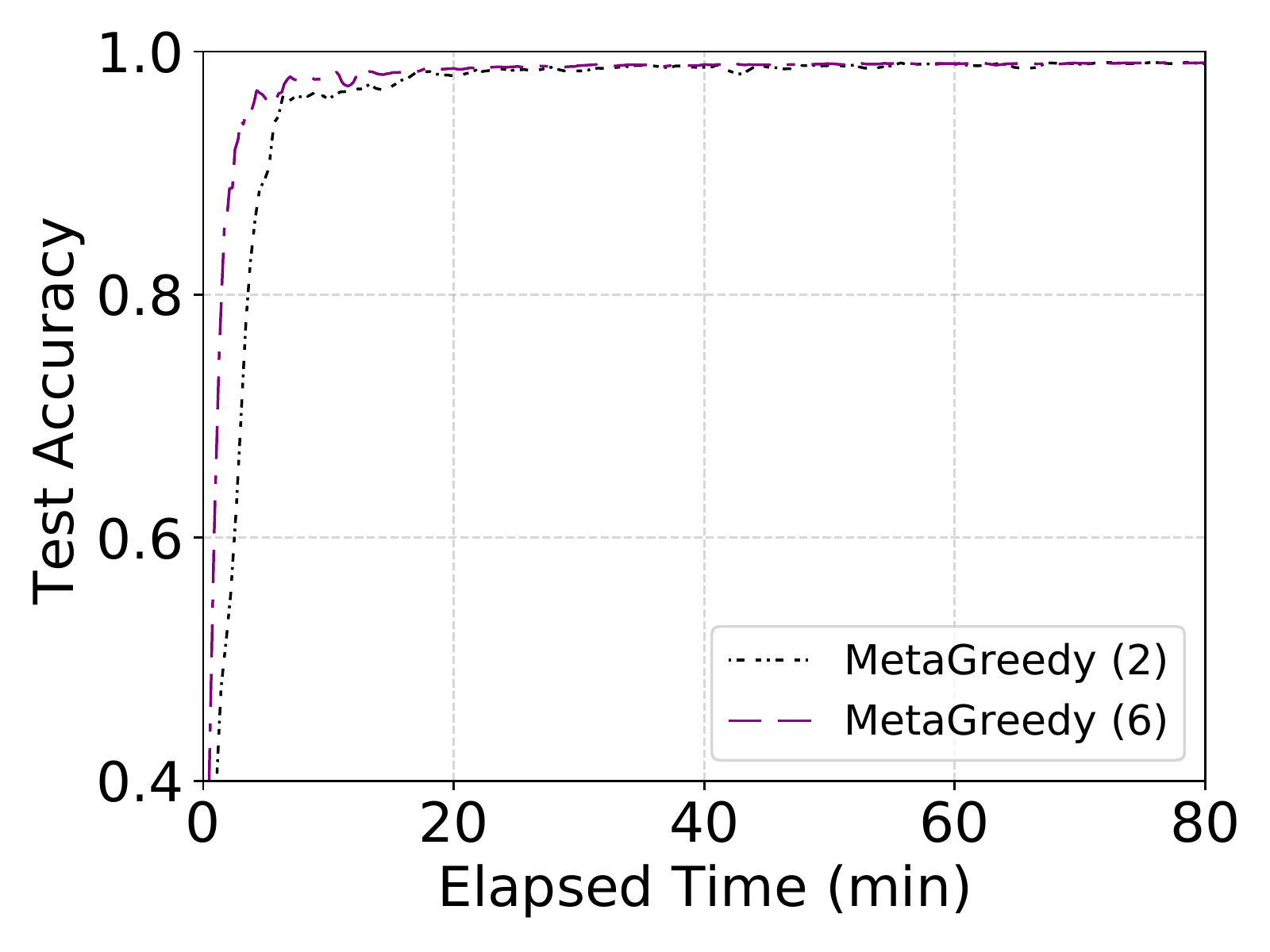}
\vspace{-4mm}
\caption{LeNet with non-IID data}
\label{fig:groupR1Exp2LeNetNIID}
\end{subfigure}
\begin{subfigure}{0.3\linewidth}
\includegraphics[width=\linewidth]{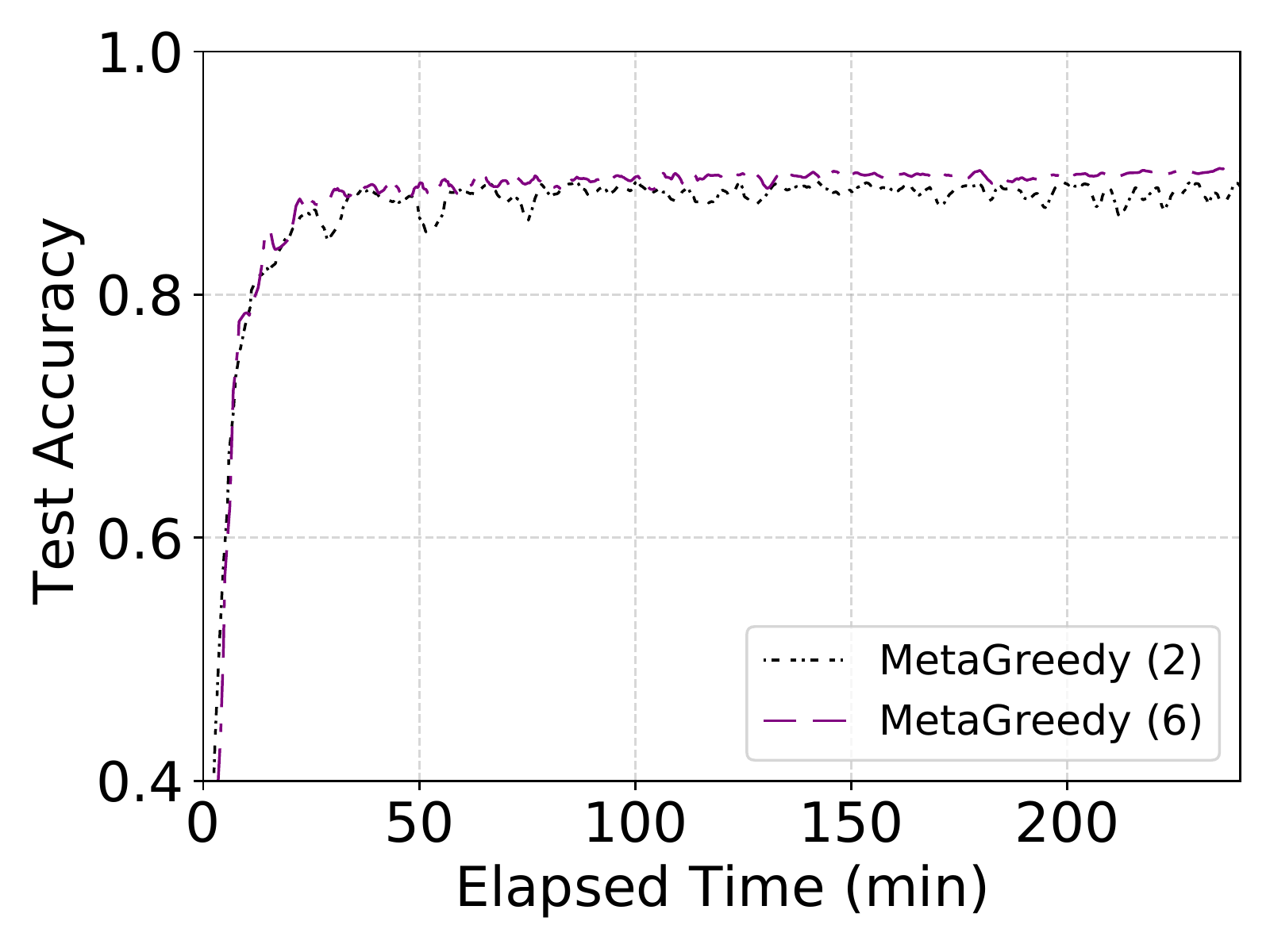}
\vspace{-4mm}
\caption{CNN with non-IID data}
\label{fig:groupR1Exp2CNNNIID}
\end{subfigure}
\begin{subfigure}{0.3\linewidth}
\includegraphics[width=\linewidth]{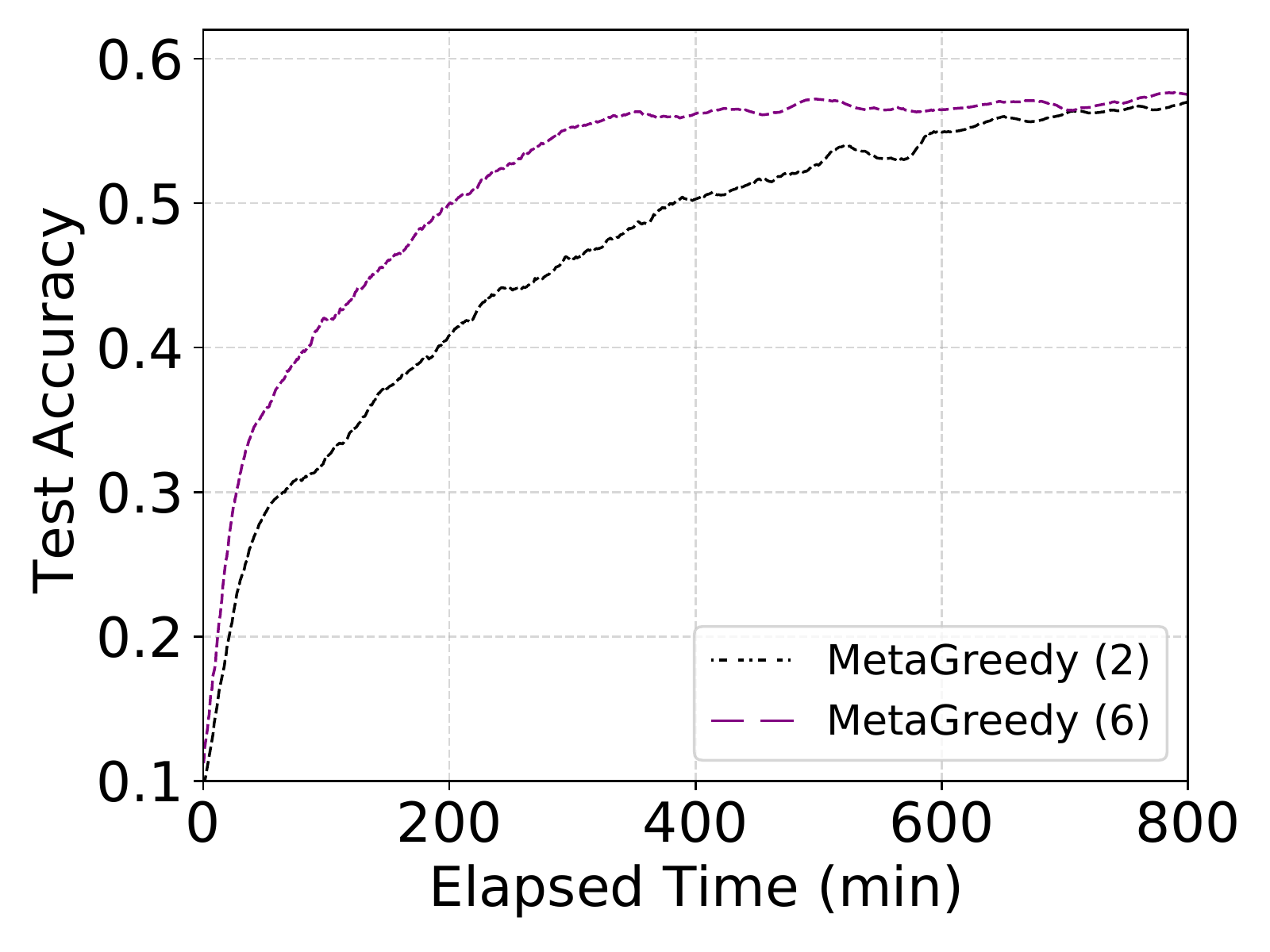}
\vspace{-4mm}
\caption{VGG with non-IID data}
\label{fig:groupR1Exp2VGGNIID}
\end{subfigure}
\vspace{-2mm}
\caption{The accuracy of different jobs in Group A with MetaGreedy (2) and MetaGreedy (6). MetaGreedy (2) represents Meta-Greedy with 2 methods and MetaGreedy (6) represents Meta-Greedy with 6 methods.}
\label{fig:groupR1Exp2}
\end{figure*}

\begin{figure*}[htbp]
\centering
\begin{subfigure}{0.45\linewidth}
\includegraphics[width=\linewidth]{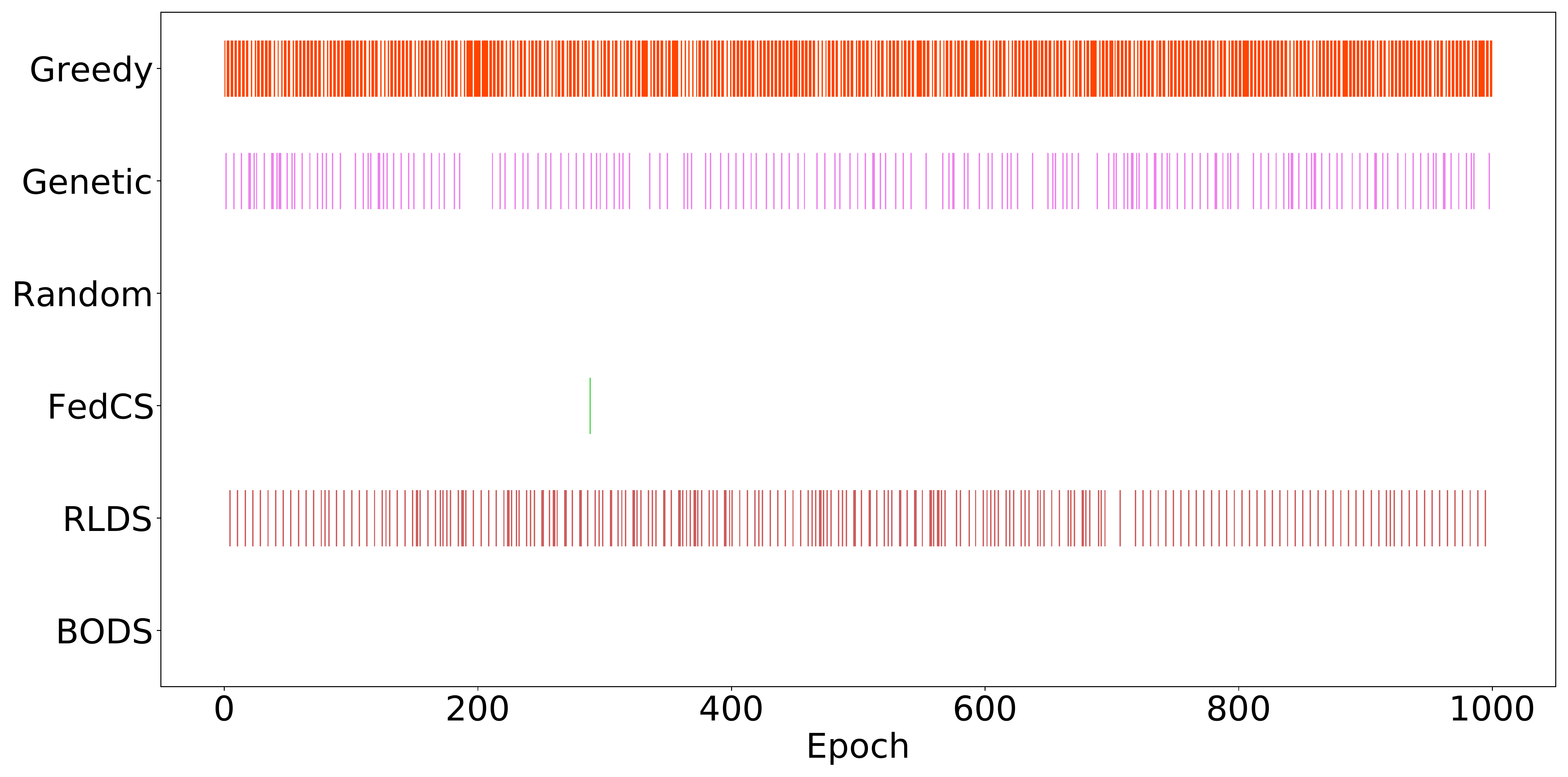}
\vspace{-4mm}
\caption{LeNet with IID data}
\label{fig:groupR1Exp3LeNetIID}
\end{subfigure}
\begin{subfigure}{0.45\linewidth}
\includegraphics[width=\linewidth]{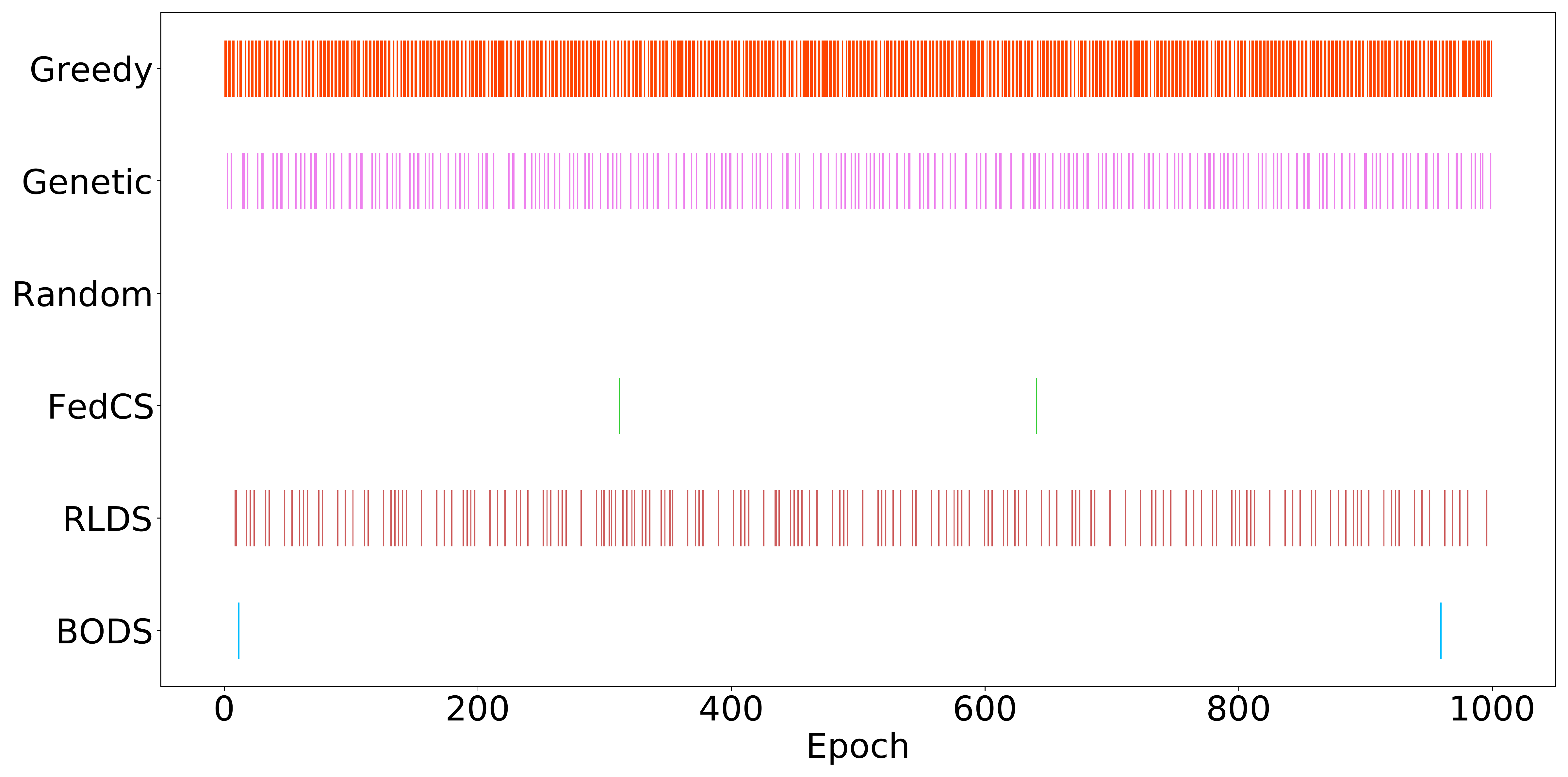}
\vspace{-4mm}
\caption{LeNet with non-IID data}
\label{fig:groupR1Exp3LeNetNIID}
\end{subfigure}
\begin{subfigure}{0.45\linewidth}
\includegraphics[width=\linewidth]{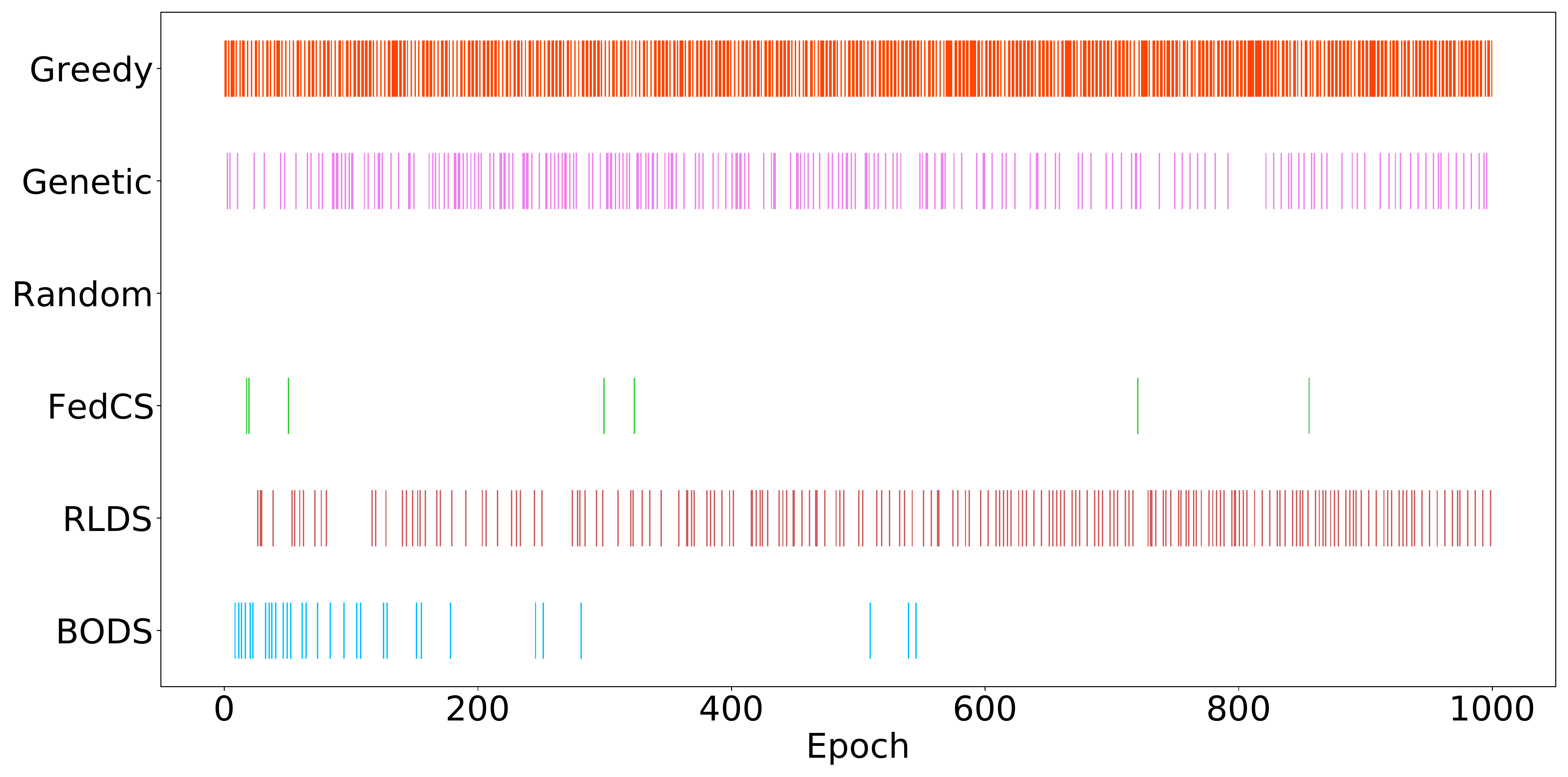}
\vspace{-4mm}
\caption{CNN with IID data}
\label{fig:groupR1Exp3CNNIID}
\end{subfigure}
\begin{subfigure}{0.45\linewidth}
\includegraphics[width=\linewidth]{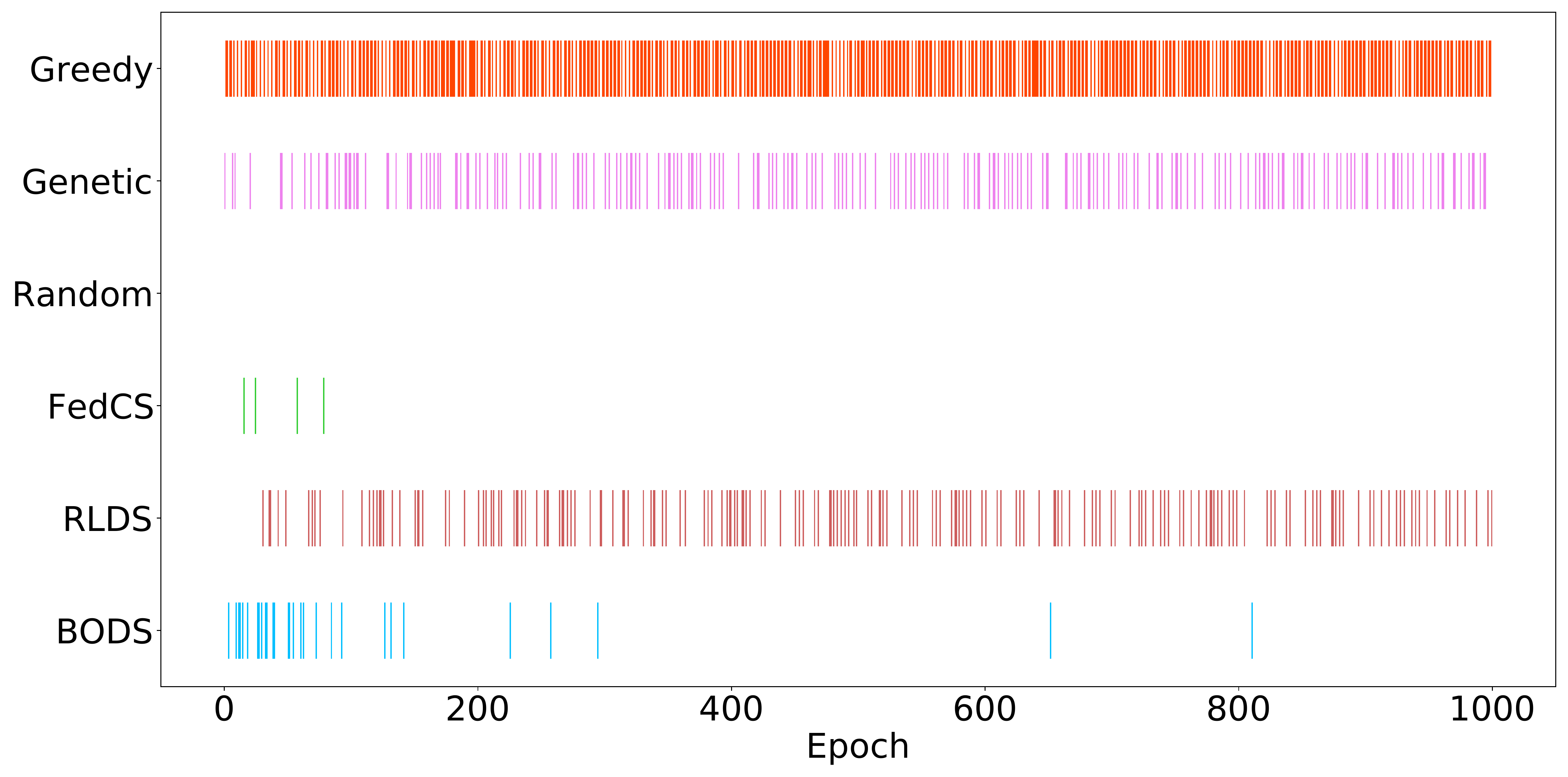}
\vspace{-4mm}
\caption{CNN with non-IID data}
\label{fig:groupR1Exp3CNNNIID}
\end{subfigure}
\begin{subfigure}{0.45\linewidth}
\includegraphics[width=\linewidth]{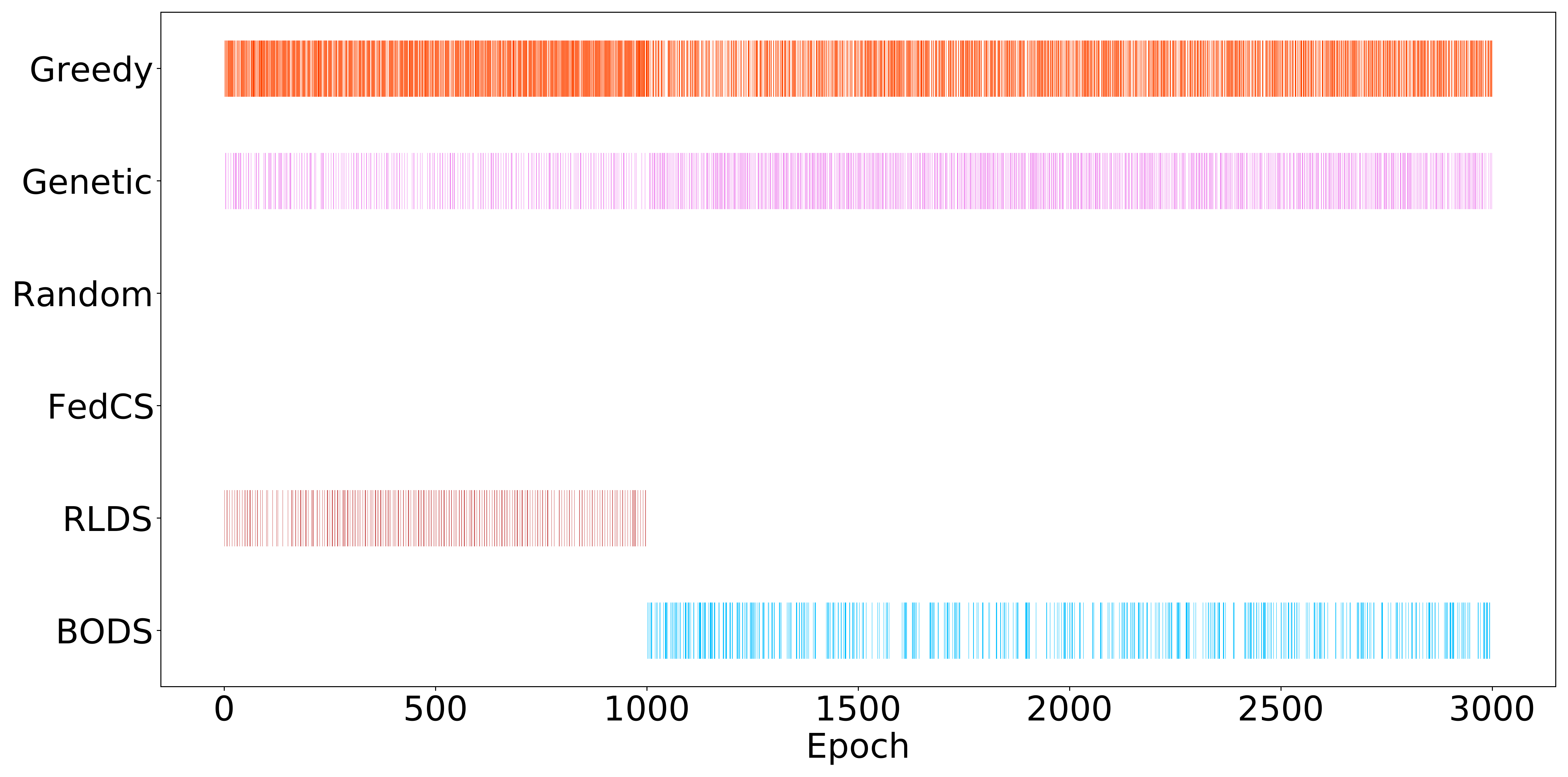}
\vspace{-4mm}
\caption{VGG with IID data}
\label{fig:groupR1Exp3VGGIID}
\end{subfigure}
\begin{subfigure}{0.45\linewidth}
\includegraphics[width=\linewidth]{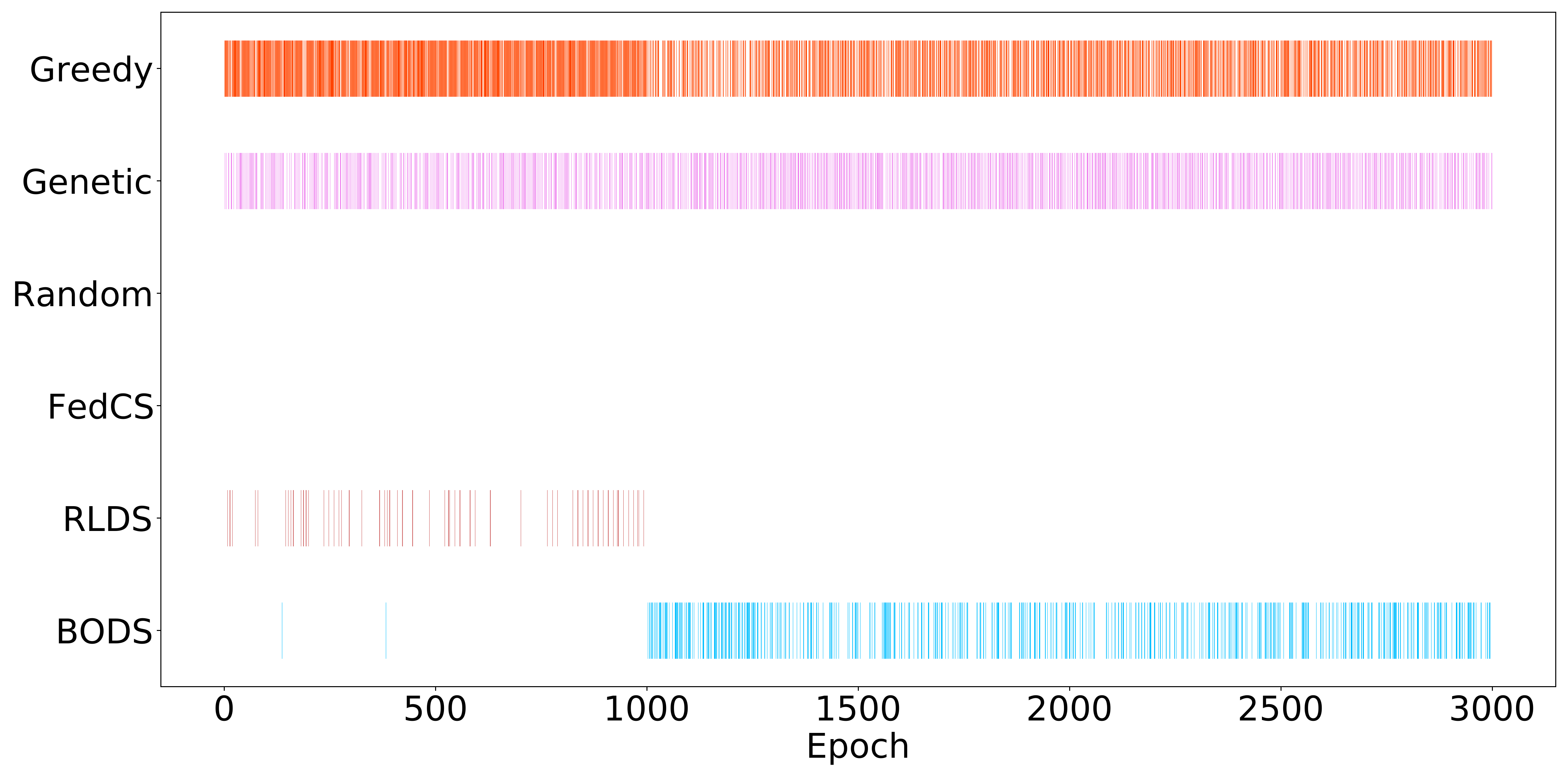}
\vspace{-4mm}
\caption{VGG with non-IID data}
\label{fig:groupR1Exp3VGGNIID}
\end{subfigure}
\vspace{-2mm}
\caption{The participation frequency of diverse methods in the training process of Group A with Meta-Greedy and 6 methods.}
\label{fig:groupR1Exp3}
\end{figure*}

As shown in Figure \ref{fig:groupR1Exp1}, RLDS favors complex jobs (VGG and ResNet) while BOSD corresponds to better performance for a simple job (CNN). We exploit VGG19 (21,240,010) \cite{simonyan2015very} and ResNet18 (595,466) \cite{He2016} to train models with non-IID Cifar10 \cite{krizhevsky2009learning} dataset. We exploit CNN with 491,920 parameters to train a model with the emnist-digital dataset \cite{cohen2017emnist}.

As RLDS can learn more information through a complex neural network, RLDS outperforms BODS for complex jobs (0.008 and 0.029 in terms of accuracy with VGG19 and ResNet18, and 46.7\% and 34.8\% faster for the target accuracy of 0.7 with VGG19 and 0.5 with ResNet18). Due to the emphasis on the combination of data fairness and device capabilities, i.e.,computation and communication capabilities, BODS can lead to high convergence accuracy and fast convergence speed for simple jobs (0.018 in terms of accuracy and 38\% faster for the target accuracy of 0.97 with CNN; see details in Appendix).

\subsubsection*{Meta-Greedy with 2 Methods and 6 Methods}

As shown in Figure \ref{fig:groupR1Exp2}, Meta-Greedy with 2 methods (BODS and RLDS) significantly outperforms that with 6 methods (BODS, RLDS, Genetic, Greedy, Random, and FedCS).

\subsubsection*{Frequency of Each Method in Meta-Greedy}

As shown in Figure \ref{fig:groupR1Exp3}, Meta-Greedy with 2 methods (BODS and RLDS) significantly outperforms that with 6 methods (BODS, RLDS, Genetic, Greedy, Random, and FedCS). We find that Greedy and Genetic are extensively exploited, RLDS is selected at the beginning of the training process, BODS is chosen at the end of the training process, FedCS participates with less frequency, and Random is seldomly utilized. This result shows that when there are Greedy, Genetic, FedCS, and RLDS, Meta-Greedy can combine them to generate better scheduling plans for a simple job (LeNet with IID). When the job becomes complex, Meta-Greedy combines Greedy, Genetic, FedCS, RLDS, and BODS to generate proper scheduling plans (LeNet with non-IID, CNN with both IID and non-IID). However, when the job becomes even more complex, Meta-Greedy exploits more frequently RLDS and BODS. RLDS is utilized at the beginning of the training process because of its superior performance while BODS is exploited at the end. As BODS may introduce some randomness to the scheduling process, it may correspond to better data fairness and higher accuracy at the end. Please note that does not contradict with the claim ``BODS favors simple jobs and RLDS favors complex jobs'' as the combination of Greedy and Genetic can well address the simple jobs and the training process is quite different from that of a single scheduling method. As Meta-Greedy can intelligently select a proper scheduling plan based on method, it corresponds to efficiency training process. 


\end{document}